\documentclass[a4paper,9.5pt]{article}
\usepackage[USenglish]{babel}
\usepackage[T1]{fontenc}
\usepackage[ansinew]{inputenc}
\usepackage{lmodern}

\usepackage[margin=0.7in, paperwidth=8.4in, paperheight=11.5in]{geometry}

\usepackage{natbib}
\usepackage{fancyhdr}
\usepackage{lastpage}
\usepackage{textcomp}
\usepackage{xcolor}

\usepackage{amsmath}
\usepackage{amsfonts}
\usepackage{graphicx}
\usepackage{multirow}
\usepackage[bookmarks = true, colorlinks = true]{hyperref}

\def \Remark{\noindent\underline{\textbf{Remark }}}

\DeclareMathOperator*{\argsup}{arg\,sup\,}
\DeclareMathOperator*{\arginf}{arg\,inf\,}

\newtheorem{theorem}{Theorem}

\newtheorem{example}[theorem]{Example}

\newtheorem{remark}[theorem]{Remark}

\newenvironment{proof}[1][Proof]{\noindent\textbf{#1.} }{\ \rule{0.5em}{0.5em}}

\begin{document}

\pagestyle{headings}  % switches on printing of running heads
%%%%%%%%%%%%%%%%%%%%%%%%%%%%%%%%%%%%%%%%%%%%%%%%%%%%%%%%%%%%%
%% DOCUMENT
%%%%%%%%%%%%%%%%%%%%%%%%%%%%%%%%%%%%%%%%%%%%%%%%%%%%%%%%%%%%%

\title{Towards a better understanding of the dual representation of phi divergences}
\author{Diaa AL MOHAMAD\\
Laboratoire de Statistique Th\'eorique et Appliqu\'ee, Universit\'e Pierre et Marie Curie\\
4 place Jussieu 75005 PARIS}
%\and
%Higher Institute for Applied Sciences and Technology, Damascus, Syria}
\date{\today}
\maketitle

\begin{abstract}
The aim of this paper is to study different estimation procedures based on $\varphi-$divergences. The dual representation of $\varphi-$divergences based on the Fenchel-Legendre duality is the main interest of this study. It provides a way to estimate $\varphi-$divergences by a simple plug-in of the empirical distribution without any smoothing technique. Resulting estimators are thoroughly studied theoretically and with simulations showing that the so called minimum $\varphi-$divergence estimator (MD$\varphi$DE) is generally non robust and behaves similarly to the maximum likelihood estimator. We give some arguments supporting the non robustness property, and give insights on how to modify the classical approach. An alternative class of $\varphi-$divergences robust estimators based on the dual representation is presented. We study consistency and robustness properties from an influence function point of view of the new estimators. In a second part, we invoke the Basu-Lindsay approach for approximating $\varphi-$divergences and provide a comparison between these approaches. The so called dual $\varphi-$divergence is also discussed and compared to our new estimator. A full simulation study of all these approaches is given in order to compare efficiency and robustness of all mentioned estimators against the so-called minimum density power divergence, showing encouraging results in favor of our new class of minimum dual $\varphi-$divergences.
\end{abstract}
\section*{Introduction}
The maximum likelihood method is a simple and an efficient method to estimate unknown parameters of a given model. The most common drawback of such method is its sensibility to contamination and misspecification. From the first years of the twentieth century, many researchers such as Pearson, Hellinger, Kullback and Liebler, Neymann and others started using different approaches using distant-like functions between probability density functions called as divergences. Resulting estimators have shown a good robustness against outliers. Nowadays, we have several divergence-based techniques which perform well under noise presence such as $\varphi-$divergences (\cite{Csiszar}, \cite{AliSilvey}), $S-$divergences (\cite{GoshSDivergence}), R\'eny\'e pseudodistances (see for example \cite{TomaAubin}), Bregman divergences and many others. We are particularly interested here in this paper in $\varphi-$divergences and in comparing it with maximum likelihood (calculated using EM algorithm for mixtures) and some particular cases of $S-$divergences and Bregman divergences.\\
We define a $\varphi-$divergence in the sense of \cite{Csiszar} as follows. Let $\varphi:[0,\infty)\rightarrow (0,\infty)$ be a proper closed convex function. Let $P$ and $Q$ be two probability measures defined on the same measurable space $(\mathcal{A},\mathbb{B})$ such that $Q$ is absolutely continuous with respect to $P$. Denote $dQ/dP$ the corresponding Radon-Nikodym density. The $\varphi-$divergence between $Q$ and $P$ is defined by:
\begin{equation}
D_{\varphi}(Q,P) = \int{\varphi\left(\frac{dQ}{dP}(y)\right)dP(y)}
\label{eqn:PhiDivergence}
\end{equation}
If $Q$ is not absolutely continuous with respect to $P$, we set $D_{\varphi}(Q,P)=\infty$. For the class of Cressie-Read $\varphi_{\gamma}(t) = \frac{x^{\gamma}-\gamma x + \gamma -1}{\gamma (\gamma -1)}$, we get the power divergences which contain the Hellinger ($\gamma=0.5$), the Pearson $\chi^2$ ($\gamma=2$), the Neymann $\chi^2$ ($\gamma=-1$) and other classical divergences.\\
When working with discreet models, $\varphi-$divergences are simply approximated using the empirical distribution $P_n$ since both the model and the empirical distribution are absolutely continuous with respect to the Dirac measure. Efficient and robust estimators were derived and extensively studied, see for example \cite{Simpson} and (\cite{LindsayRAF}). \\
For continuous models, the empirical distribution is no longer suitable to replace directly the true distribution since the model has a continuous support. Thus, the model cannot be absolutely continuous with respect to $P_n$ and no interesting estimation procedure is produced (see \cite{BroniatowskiSeveralApplic} for a proper explanation). Authors such as \cite{Beran} proposed to simply smooth the empirical distribution using kernels. \cite{BasuLindsay} proposed to smooth both the model and the empirical distribution in order to avoid consistency conditions and rates of convergence imposed on the kernel estimator (provided the existence of a \emph{transparent} kernel). Their method can be reread in some basic examples as the calculus of the $\varphi-$divergence between a kernel estimator and a weighted version of the model $p_{\theta}^* = \text{cte }p_{\theta}^a$. Although the smoothed model may result in a loss of information, Basu and Lindsay show that this loss is rather small. They also admit that there is still a difficulty in the choice of the window and the kernel for the smoothing since providing transparent kernels for a given model is a hard task\footnote{The authors provide however, three standard examples which admit transparent kernels; they are the gaussian the gamma and the poisson models.}.\\

Recently, an approach based on some convexity arguments have been proposed by \cite{LieseVajdaDivergence} and \cite{BroniaKeziou2006}. In both articles, the authors provide similar "supremal" representations of $\varphi-$divergences where a simple plug-in of the empirical distribution is possible without any smoothing techniques. The resulting estimators were called as minimum dual $\varphi-$ divergence estimators (MD$\varphi$DE). Since their appearance, no complete study about the robustness of such estimators were proposed except for the calculus of the influence function in \cite{TomaBronia}. There were even no simulation studies either, except for the paper of \cite{Frydlova}. However, the authors have considered only the case of normal model where the MD$\varphi$DE is proved to coincide with the maximum likelihood estimator, see \cite{Broniatowski2014}. Although they get a robust estimator in only one case\footnote{They get robust results when adding outliers drawn from a Cauchy distribution.}, we believe that it was due to a calculation error which we explain later.\\
The dual representation proposed by both \cite{LieseVajdaDivergence} and \cite{BroniaKeziou2006} performs well under the model. It even coincides with the MLE in full exponential families, and hence have the same efficiency as the MLE. Weak and strong consistency is reached under classical conditions (see \cite{BroniatowskiKeziou2007}). Limit laws of the MD$\varphi$DE and the estimated divergence are simple and were exploited to build statistical tests. However, when we are not under the model, this approach seems to be inconvenient and suffer from lack of robustness. When we are in contamination models or under misspecification, this approach does not approximate well the $\varphi-$divergence between the empirical distribution and the model. It even remarkably underestimates its value. We propose in this paper a brief explanation of this problem and provide a general solution. We also give two particular solutions. The first is based on kernels which avoids the supremal form (hence no double optimization). The second is devoted to contamination models which appears as a slight modification of the classical MD$\varphi$DE. We study the consistency and the robustness from an influence function point of view of the kernel-based estimator and check corresponding conditions on simple examples. \\
In a second part of this paper, we briefly recall the Basu-Lindsay approach (\cite{BasuLindsay}) and discuss what happens in the context of densities defined on $(0,\infty)$. We show that symmetric kernels are not suitable and provide some solutions through asymmetric kernels. We also discuss some of the positive and drawbacks of the so-called dual $\varphi-$divergence estimator (\cite{BroniatowskiKeziou2007}) which is another estimator derived through the dual representation of the divergence. The sensitivity to the choice of the escort parameter is invoked. We compare this estimator with the density power divergence of \cite{BasuMPD} and show a strong relation between these two methods.\\
Finally, we provide several simulation results in a simple gaussian model, a mixture of two gaussian components, a generalized Pareto distribution (GPD) and a mixture of two Weibull components. We make a comparison with the classical MD$\varphi$DE, the MLE (calculated using EM for the mixture case), the Basu-Lindsay approach, the so called dual $\varphi-$divergence estimator, Beran's approach and the minimum power density divergence (MPD) when the data is drawn under the model and when it is contaminated by $10\%$ of observations from other distributions. Our new estimator is as efficient as MLE under the model and is robust against outliers. Although not being the best robust estimator in simple examples (but close enough), it shows promising performances in difficult ones conquering other methods.\\

The paper is organized as follows. In Section 1, we give a theoretical introduction of the dual representation of $\varphi-$divergences. We explain the problem of the existing approach based on duality and introduce a solution to robustify the classical MD$\varphi$DE. Section 2 is devoted to the asymptotic properties of our kernel-based MD$\varphi$DE where a set of conditions ensuring consistency are given and verified on a gaussian model. The influence function is also calculated and proved to be bounded on a simple example. In Section 3, we recall the Basu-Lindsay approach and see how it is applied when one uses asymmetric kernels. In Section 4, we discuss some of the positive points of our new estimator in comparison to the classical MD$\varphi$DE and the Basu-Lindsay approach. We also enlist some of its drawbacks. The so called dual $\varphi$ divergence estimator is discussed in Section 5 and we show a convergence with the density power divergence introduced in \cite{BasuMPD}. In Section 6, we give another estimator for contamination models and discuss briefly its properties. Finally, Section 7 contains an extensive simulation study and a comparative discussion with other estimators presented in the paper.

% -----------------------------------------------------------------------
%%%%%%%%%%%%%%%%%%%%%%%%%%%%%%%%%%%%%%%%%%%%%%%%%%%%%%%%%%%
%========================================================
%%%%%%%%%%%%%%%%%%%%%%%%%%%%%%%%%%%%%%%%%%%%%%%%%%%%%%%%%%%
% -----------------------------------------------------------------------
\section{The dual representation of \texorpdfstring{$\varphi$}{TEXT}--divergences}
\subsection{The theoretical approaches}
\cite{LieseVajdaDivergence} propose the following "supremal" representation of $\varphi-$divergences. Let $\mathcal{P}$ be a class of mutually absolutely continuous distributions such that for any triplet $P,P_T$ and $Q$, $\varphi'(dP_T/dQ)$ is $P$-integrable. Theorem 17 in \cite{LieseVajdaDivergence} states that:
\begin{equation}
D_{\varphi}(P_T,P) = \sup_{Q\in\mathcal{P}}\int{\varphi'\left(\frac{dQ}{dP}\right)dP_T} + \int{\varphi\left(\frac{dQ}{dP}\right)dP} - \int{\varphi'\left(\frac{dQ}{dP}\right)dQ}
\label{eqn:LieseVajdaRep}
\end{equation}
and the supremum is attained when $Q=P_T$.
\cite{BroniaKeziou2006} have also developed a similar and a more general representation of $D_{\varphi}(P,P_T)$. Let $\mathcal{F}$ be some class of $\mathcal{B}-$measurable real valued functions. Let $\mathcal{M}_{\mathcal{F}}$ be the subspace of the space of probability measures $\mathcal{M}$ defined by $\mathcal{M}_{F} = \{P\in\mathcal{M} | \int{|f|dP<\infty, \forall f\in\mathcal{F}}\}$. Assume that $\varphi$ is differentiable and strictly convex. Then, for all $P\in\mathcal{M}_{\mathcal{F}}$ such that $D_{\varphi}(P,P_T)$ is finite and $\varphi'(dP/dP_T)$ belongs to $\mathcal{F}$, the $\varphi-$divergence admits the dual representation (see Theorem 4.4 in \cite{BroniaKeziou2006}):
\begin{equation}
D_{\varphi}(P,P_T) = \sup_{f\in\mathcal{F}} \int{f dP} - \int{\varphi^*(f)dP_T},
\label{eqn:GeneralDualRep}
\end{equation}
where $\varphi^*(x)=\sup_{t\in\mathbb{R}} tx-\varphi(t)$ is the Fenchel-Legendre convex conjugate. Moreover, the supremum is attained at $f=dP/dP_T$.\\
When substituting $\mathcal{F}$ by the class of functions $\{\varphi'(dP/dQ)\}$, and using the property $\varphi^*(\varphi'(t)) = t\varphi'(t)-\varphi(t)$, we obtain the same representation given above in (\ref{eqn:LieseVajdaRep}). Both formulations (\ref{eqn:LieseVajdaRep}) and (\ref{eqn:GeneralDualRep}) are interesting in their own and in their proofs. The second formula gives us the opportunity to reproduce many supremal forms for the $\varphi-$divergence.\\
In a parametric setup where $dP_{\phi} = p_{\phi}dx$ for $\phi\in\Phi\subset\mathbb{R}^d$ and the true distribution generating the data is a member of the model, i.e. $P_T=P_{\phi^T}$ for some $\phi^T\in\Phi$, \cite{BroniaKeziou2006} propose to use the class of functions $\mathcal{F}_{\phi} = \{\varphi'(p_{\phi}/p_{\alpha}),\alpha\in\Phi\}$. The dual representation of $D_{\varphi}$ is now written as:
\begin{equation}
D_{\varphi}(p_{\phi},p_{\phi_T}) = \sup_{\alpha\in\Phi}\left\{\int{\varphi'\left(\frac{p_{\phi}}{p_{\alpha}}\right)(x)p_{\phi}(x)dx} - \int{\left[\frac{p_{\phi}}{p_{\alpha}} \varphi'\left(\frac{p_{\phi}}{p_{\alpha}}\right) - \varphi'\left(\frac{p_{\phi}}{p_{\alpha}}\right)\right](y) p_{\phi^T}(y)dy}\right\}.
\label{eqn:ParametricDualForm}
\end{equation}
The idea behind this choice is that the supremum is attained when $\alpha = \phi^T$. Since $p_{\phi^T}$ is unknown, one think about replacing $p_{\phi^T}dy$ by the empirical distribution. This seems very natural and does not cause any problem of absolute continuity as in formula (\ref{eqn:PhiDivergence}). We now get:
\begin{equation}
\hat{D}_{\varphi}(p_{\phi},p_{\phi_T}) = \sup_{\alpha\in\Phi}\left\{\int{\varphi'\left(\frac{p_{\phi}}{p_{\alpha}}\right)(x)p_{\phi}(x)dx} - \frac{1}{n}\sum_{i=1}^n{\left[\frac{p_{\phi}}{p_{\alpha}} \varphi'\left(\frac{p_{\phi}}{p_{\alpha}}\right) - \varphi'\left(\frac{p_{\phi}}{p_{\alpha}}\right)\right](y_i)}\right\}
\label{eqn:DivergenceDef}
\end{equation}
This quantity is "nearly" the divergence between the empirical distribution and the model. Both \cite{BroniaKeziou2006} and \cite{LieseVajdaDivergence} propose to estimate the set of parameters $\phi^T$ by:
\begin{equation}
\hat{\phi}_n = \arginf_{\phi\in\Phi} \sup_{\alpha\in\Phi} \hat{D}_{\varphi}(p_{\phi},p_{\phi_T})
\label{eqn:MDphiDEClassique}
\end{equation}
This was called by \cite{BroniaKeziou2006} as the minimum dual $\varphi-$divergence estimator (MD$\varphi$DE) who have also studied the asymptotic properties and provided sufficient conditions for the consistency of this estimator. They have also built some test statistics based on it. \cite{TomaBronia} and \cite{BroniatowskiSeveralApplic} have studied the robustness of such an estimator from an influence function (IF) point of view. The IF is unfortunately unbounded in general and does not even depend on $\varphi$ for the classe of Cressie-Read functions $\varphi_{\gamma}$ presented in the introduction. This fact is still not sufficient to conclude the non robustness of the MD$\varphi$DE. It was pointed out by many authors in the context of $\varphi-$divergences that one may have an  unbounded influence function, still the resulting estimators enjoy a good robustness against outliers, see \cite{Beran} for the hellinger divergence in continuous models and \cite{LindsayRAF} for a general class of $\varphi-$divergences in discrete models. \\
Till this day, and to the best of our knowledge, there is not even a simulation study of the robustness of the MD$\varphi$DE although it is an estimator which, similarly to the power density estimator of \cite{BasuMPD}, does not require any smoothing or escort parameters. Besides, the asymptotic properties are proved with merely classical conditions on the model. The only simulation study, to our knowledge, is done by \cite{Frydlova} and focuses only on the normal model. In their results, the MD$\varphi$DE has comparative results to the maximum likelihood estimator when no contamination is present, while they get some cases where the MD$\varphi$DE is robust under contamination, although they \emph{should not} as we will see later in the following paragraph.
%%%%%%%%%%%%%%%%%%%%%%%%%%%%%%%%%%%%%%%%%%%%%%%%%%%%%%%%%%%
% -----------------------------------------------------------------------
\subsection{Does the MD\texorpdfstring{$\varphi$}{TEXT}DE have any chance to be robust?}\label{sec:NonRobsutMDphiDE}
\paragraph{Equality with MLE in exponential families.} An important aspect about the classical MD$\varphi$DE is that it coincides with the maximum likelihood estimator in full exponential models whenever the corresponding \emph{true} divergence $D_{\varphi}$ is finite, see \cite{Broniatowski2014}. This covers the standard gaussian model for which \cite{Frydlova} provided clear robust properties of the MD$\varphi$DE when outliers are generated by the standard Cauchy distribution. This contradicts with the theoretical result presented in \cite{Broniatowski2014} which is an exact result and depends only on analytic arguments. We have done similar simulations and found out that numerical problems may play a nice role here. Fortunately, we have no numerical integration since all integrals can be easily calculated, see \cite{Frydlova} or \cite{BroniatowskiSeveralApplic}. When using the standard Cauchy distribution to generate outliers, we get points with very large values superior to 100. These points participate only in the sum term in the MD$\varphi$DE (\ref{eqn:MDphiDEClassique}). A gaussian density with parameters not very far from the standard ones ($\mu=0,\sigma=1$) will produce a value equal to 0 in numerical computer programs. Thus, numerical problems of the form $0/0$ would appear when calculating the sum term in (\ref{eqn:DivergenceDef}) since the summand is of the form $g(p_{\theta}/p_{\alpha})(y_i)$. If one uses simple practical solutions to avoid this, such as adding a very small value (e.g. $10^{-100}$) to the denominator or the nominator, a thresholding effect is produced and the \emph{true} fraction is badly calculated. As a result, such outliers would have practically no effect in the procedure as if they were not added, and one would obtain "forged robust estimates". The same thresholding effect does not happen in the MLE since the likelihood function does not contain any fractions. On the other hand, if one calculates the fraction using the properties of the exponential function, i.e. $p_{\theta}(y_i)/p_{\alpha}(y_i) = \exp[(y_i-\alpha)^2/2 - (y_i-\phi)^2/2]$, the MD$\varphi$DE defined by (\ref{eqn:MDphiDEClassique}) gives the same result as the maximum likelihood estimator and never better\footnote{On the basis of 100 experiments, there were about 20 experiments where the MD$\varphi$DE suffered from numerical complications and exploses higher than the MLE.}.\\ 
We have performed further simulations on several models which do not belong the exponential family and found out that the MD$\varphi$DE have a very similar behavior to the MLE, see Sect. \ref{sec:Simulations} below. This should not be very surprising because of the convergence between exponential families and a large class of probability laws. Papers such as \cite{BarronSheu} discussed how one can estimate a probability density using an exponential families and proved interesting convergence rates.\\
\paragraph{Why should not it work well although being an estimator of a $\varphi$--divergence which is proved to be a robust procedure (see \cite{Donoho})?} We do not pretend to give a full answer about the non robustness of the MD$\varphi$. Our argument here is intuitive. When $P_T$ is a member of the model, the approximated dual formula converges to the $\varphi-$divergence, and the argument of the infimum to the corresponding one, as the number of observations increases. This consistency was discussed in Proposition 3.1 in \cite{BroniatowskiKeziou2007}. Their result, however, does not hold when $P_T$ is not a member of the model, i.e. under contamination or misspecification. Indeed, consistency is in the following sense:
\[\hat{D}_n(P_{\phi},P_T)\rightarrow \sup_{\alpha\in\Phi}\left\{\int{\varphi'\left(\frac{p_{\phi}}{p_{\alpha}}\right)(x)p_{\phi}(x)dx} - \int{\varphi^{\#}\left(\frac{p_{\phi}}{p_{\alpha}}\right)(y) dP_T(y)dy}\right\},\]
and the arginf of the left hand side to the arginf of the right hand side. The limiting quantity is the dual representation of the $\varphi-$divergence, and since the supremum is attained uniquely when $p_{\alpha}=dP_T/dy$, then it is never attained as long as $P_T$ is not a member of the model. Moreover, the limiting quantity is a lower bound of the divergence and minimizing the former does not guarantee the minimization of the later. Figure \ref{fig:UnderEstimation} represent this idea on a standard gaussian model where the mean is unknown. The true distribution is contaminated by a gaussian distribution $\mathcal{N}(\mu=10,\sigma=2)$. The minimum of the dual representation is attained at $\mu=1$ whereas it is attained at 0 for the true divergence. Figure (a) shows formula (\ref{eqn:ParametricDualForm}) and figure (b) shows formula (\ref{eqn:DivergenceDef}). The data contains 100 observations. We also represent the solution introduced in the following paragraph which overcomes this problem.
\begin{figure}[h]
\includegraphics[scale=0.31]{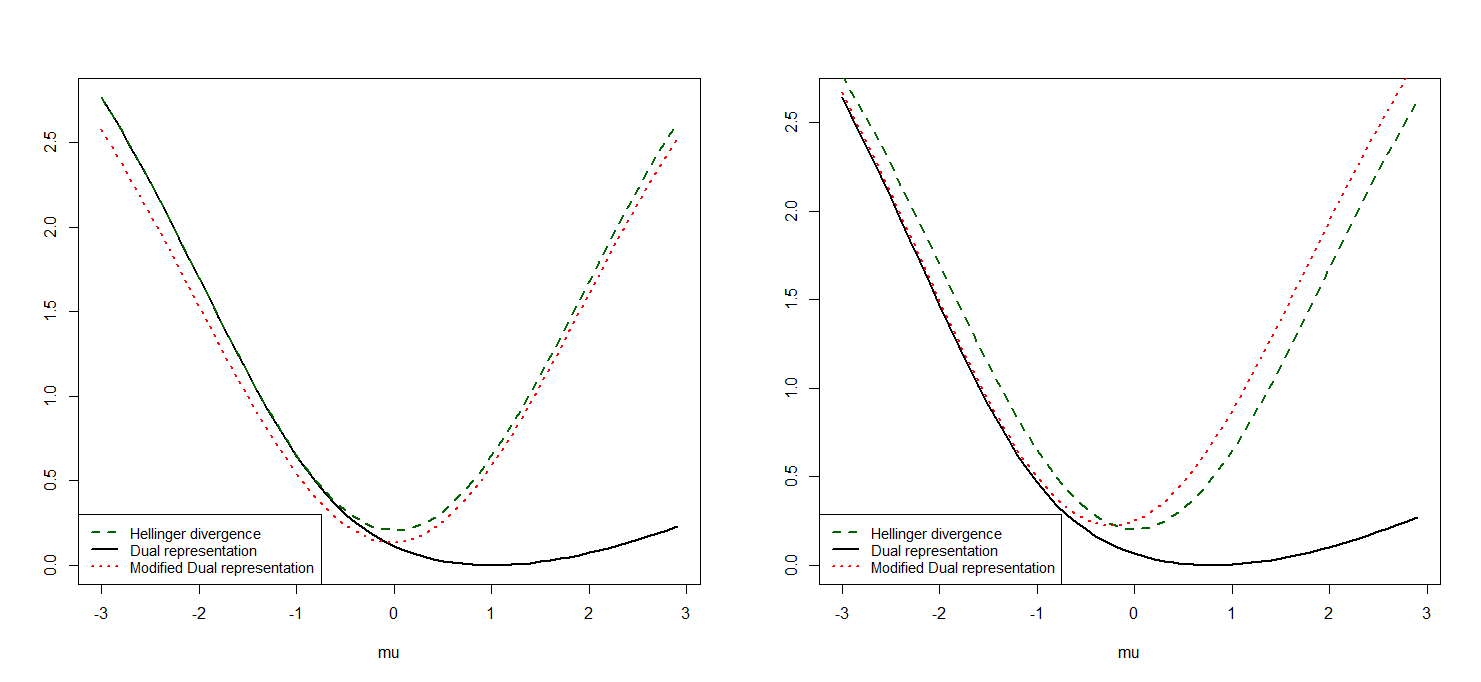}
\label{fig:UnderEstimation}
\caption{Underestimation caused by the classical dual representation compared to the new one. The true distribution is taken to be $0.9\mathcal{N}(\mu=0,\sigma=1)+0.1\mathcal{N}(\mu=10,\sigma=2)$. Figure (a) shows the dual representation defined by (\ref{eqn:ParametricDualForm}) in comparison with the new reformulation defined by (\ref{eqn:NewExactDualForm}). Figure (b) shows the corresponding approximations when we replace the true distribution by its empirical version.}
\end{figure}
\subsection{New reformulation of the dual representation}\label{subsec:KernelSolution}
As stated previously, when the data is contaminated, the supremum in (\ref{eqn:DivergenceDef}) is not attained and the approximation of the divergence between the model and the empirical distribution\footnote{Although this quantity is not well defined in the continuous case, the plug-in of the empirical distribution in the dual representation gives an idea about the divergence between the model and the empirical distribution.} is dramatically degraded. Since in the \emph{not approximated formula} (\ref{eqn:ParametricDualForm}) or (\ref{eqn:LieseVajdaRep}), the supremum is attained uniquely whenever $p_{\alpha}=p_T$, an intuitive idea is to replace $p_{\alpha}$ by an adaptive (nonparametric) estimator of $p_T$ which does not take into account the restriction of being in the model. We, then, have a dual representation where the supremum is, nearly, attained whether we are under the model or not. This way our criterion should inherit robustness properties against possible contamination as it approximates a $\varphi-$divergence.\\
One should be able to propose many solutions which correspond to this idea in order to reach a supremal attainment in the dual representation which may vary depending on the situation. For example, if we face a proportion of large-values outliers, one may add an extra component to $p_{\alpha}$, i.e. replace $p_{\alpha}$ with the mixture $\lambda p_{\alpha} + (1-\lambda) q_{\theta}$. The extra component covers the outliers part in a smooth way. This suggestion is still very specific and treats only the case of contaminated data. Any nonparametric estimator of $p_T$ can be used whose parameters may be determined automatically in the supremum calculus (since the supremum will be over the window parameter). In what follows $K_{n,w}$ denotes a kernel estimator\footnote{In formula (\ref{eqn:NewExactDualForm}) which comes next, the kernel function should not have a compact support such as the Epanechnikov kernel for the sake of integration existence. This is only temporary, and as we define the new estimator of the $\varphi-$divergence, the integral is replaced by a Monte-Carlo average where the kernel is only calculated on observed data and we get rid of the integration problem, and thus the use of a compact support kernel becomes possible.} of $p_T$ defined using a symmetric or asymmetric kernel with or without bias-correction treatment.\\
In order to introduce our new MD$\varphi$DE, let's go back to the beginning. We restart from formula (\ref{eqn:GeneralDualRep}) and use the following class of functions $\mathcal{F}_{\theta,n}=\{p_{\theta}/K_{n,w}, w>0\}$. The dual representation is now given by:
\begin{equation}
D_{\varphi}(p_{\phi},p_T) = \sup_{w>0}\left\{\int{\varphi'\left(\frac{p_{\phi}}{K_{n,w}}\right)(x)p_{\phi}(x)dx} - \int{\left[\frac{p_{\phi}}{K_{n,w}} \varphi'\left(\frac{p_{\phi}}{K_{n,w}}\right) - \varphi'\left(\frac{p_{\phi}}{K_{n,w}}\right)\right](y) p_T(y)dy}\right\}
\label{eqn:NewExactDualForm}
\end{equation}
The supremum calculus will produce a window for which the kernel $K_{n,w_{\text{opt}}}$ is the \emph{closest} (in some sense) to $p_{\phi^T}$. Now, we approximate $D_{\varphi}(p_{\phi},p_{\phi_T})$ by:
\[D_{\varphi}(p_{\phi},p_T) \approx \int{\varphi'\left(\frac{p_{\phi}}{K_{n,w_{\text{opt}}}}\right)(x)p_{\phi}(x)dx} - \int{\left[\frac{p_{\phi}}{K_{n,w_{\text{opt}}}} \varphi'\left(\frac{p_{\phi}}{K_{n,w_{\text{opt}}}}\right) - \varphi'\left(\frac{p_{\phi}}{K_{n,w_{\text{opt}}}}\right)\right](y) p_T(y)dy}\]
Since $p_T$ is the unknown object we hope to estimate, we replace it by its empirical version. Our final approximation is given by:
\[\hat{D}_{\varphi}(p_{\phi},p_T) = \int{\varphi'\left(\frac{p_{\phi}}{K_{n,w_{\text{opt}}}}\right)(x)p_{\phi}(x)dx} - \frac{1}{n}\sum_{i=1}^n{\left[\frac{p_{\phi}(y_i)}{K_{n,w_{\text{opt}}}(y_i)} \varphi'\left(\frac{p_{\phi}}{K_{n,w_{\text{opt}}}}\right)(y_i) - \varphi\left(\frac{p_{\phi}}{K_{n,w_{\text{opt}}}}\right)(y_i)\right]}\]
Define now the new minimum dual $\varphi-$divergence estimator by:
\begin{equation}
\hat{\phi}_{n} = \arginf_{\phi\in\Phi} \int{\varphi'\left(\frac{p_{\phi}}{K_{n,w_{\text{opt}}}}\right)(x)p_{\phi}(x)dx} - \frac{1}{n}\sum_{i=1}^n{\left[\frac{p_{\phi}(y_i)}{K_{n,w_{\text{opt}}}(y_i)} \varphi'\left(\frac{p_{\phi}}{K_{n,w_{\text{opt}}}}\right)(y_i) - \varphi\left(\frac{p_{\phi}}{K_{n,w_{\text{opt}}}}\right)(y_i)\right]}
\label{eqn:NewMDphiDE}
\end{equation}
An important question which arises now is: what should be the value of $w_{\text{opt}}$ since its calculus demands knowing the true distribution? For the time being, we do not have any specific propositions for the choice of the window. Taking into consideration that the window should be chosen in order to copy the true distribution, one needs a \emph{good} kernel estimator. In the literature of kernel estimation, there exists many rules (automatic or not) to determine sub-optimum windows such as the (Silverman's or Scott's) rule-of-thumb, cross-validation methods, etc. See for example \cite{VenablesRipley} Chap 5. Figure \ref{fig:UnderEstimation} shows in a gaussian example contaminated by a gaussian component $\mathcal{N}(10,2)$ the use of Silverman's rule with a gaussian kernel. The classical dual representation clearly underestimates the true divergence whereas the new reformulation stays close to it.\\
\begin{remark}
The new MD$\varphi$DE keeps the MLE as a member of its class for the choice of $\varphi(t)=-\log(t)+t-1$. Indeed, $\varphi'(t)=-1/t+1$ and $t\varphi'(t)-\varphi(t)=\log(t)$. Thus:
\begin{eqnarray*}
\int{\varphi'\left(\frac{p_{\phi}}{K_{n,w_{\text{opt}}}}\right)(x)p_{\phi}(x)dx} & = &  1\\
\frac{1}{n}\sum_{i=1}^n{\left[\frac{p_{\phi}(y_i)}{K_{n,w_{\text{opt}}}(y_i)} \varphi'\left(\frac{p_{\phi}}{K_{n,w_{\text{opt}}}}\right)(y_i) - \varphi\left(\frac{p_{\phi}}{K_{n,w_{\text{opt}}}}\right)(y_i)\right]} & = & \frac{1}{n}\sum_{i=1}^n{\log\left(p_{\phi}(y_i)\right)-\log\left(K_{n,w_{\text{opt}}}(y_i)\right)}
\end{eqnarray*}
and thus:
\begin{eqnarray*}
\hat{\phi}_{n} & = & \arginf_{\phi\in\Phi} 1 - \frac{1}{n}\sum_{i=1}^n{\log\left(p_{\phi}(y_i)\right)}+ \frac{1}{n}\sum_{i=1}^n{\log\left(K_{n,w_{\text{opt}}}(y_i)\right)}\\
 & = & \argsup_{\phi\in\Phi} \frac{1}{n}\sum_{i=1}^n{\log\left(p_{\phi}(y_i)\right)}\\
 & = & \text{MLE}
\end{eqnarray*}
\end{remark}
\begin{remark}
Replacing $p_T$ by the empirical distribution in (\ref{eqn:NewExactDualForm}) should not be a way to calculate an automatic window for the kernel. Indeed, although the proof of the dual representation supposes mutual absolute continuity between $p_{\alpha}, p_{\phi}$ and $p_T$, one still expects that the attainment condition of the supremum ($p_{\alpha}=p_T$) should hold as we replace $p_T$ by the empirical distribution. Indeed, if we insert directly $K_{n,w}$ in (\ref{eqn:DivergenceDef}) instead of $p_{\alpha}$, the maximization becomes on the window $w$, and the supremum will always be attained for $w=0$. When the kernel estimator is calculated by convolution, recall that $K_{n,w}=K_w*P_n\rightarrow P_n$ as w goes to zero. 
\end{remark}

% -----------------------------------------------------------------------
%%%%%%%%%%%%%%%%%%%%%%%%%%%%%%%%%%%%%%%%%%%%%%%%%%%%%%%%%%%
%========================================================
%%%%%%%%%%%%%%%%%%%%%%%%%%%%%%%%%%%%%%%%%%%%%%%%%%%%%%%%%%%
% -----------------------------------------------------------------------

\section{Asymptotic properties and robustness of the new reformulation}
We present in this section some of the asymptotic properties of the new MD$\varphi$DE defined by (\ref{eqn:NewMDphiDE}). We use Theorem 5.7 from the book of \cite{Vaart} which we restate here. Consistency of the kernel-based MD$\varphi$DE means that $\hat{\phi}_n$ defined by (\ref{eqn:NewMDphiDE}) converges in probability to $\phi^T$ the true vector of parameters when we are under the model, i.e. $P_T = P_{\phi^T}$. If we are not under the model, consistency becomes with respect to the projection of $P_T$ on the model in the sens of the divergence. In other terms, the projection $P_{\phi^T}$ is the member of the model $P_{\phi}$ whose parameters are defined by $\phi^T = \arginf_{\phi\in\Phi} D_{\varphi}(P_{\phi},P_T)$.\\ 
Similarly to \cite{BasuLindsay}, there are some cases (which are rare) in which consistency of the kernel-based MD$\varphi$DE does not need any condition on the kernel window. Thus, one may find simpler versions of the results we give below. We will see that a gaussian model with unknown mean is one of these examples where we give the corresponding conditions.\\
In a second part of this section, we calculate the influence function of the kernel-based MD$\varphi$DE for a given window, and show how the use of a kernel estimate instead of the model $p_{\alpha}$ in the dual formula interferes to make the IF bounded. \\
We use the same notations as in \cite{Vaart} to note integration. Thus, if $f$ is a $P-$integrable function, we denote $Pf$ to the integral $\int fdP$. Moreover, the notation $K_w*P$ denotes the operation of smoothing $dP$ by the kernel $K_w$ with bandwidth equal to $w$. This smoothing can be done by simple convolution as in the case of Rosenblatt-Parzen kernel estimator. Other kinds of smoothing are presented in Section \ref{subsec:BLapproach}. The smoothing is supposed to be an additive operator on distributions in the sense that $K_w*(P\pm Q) = K_w*P \pm K_w*Q$.
\subsection{Consistency}
Theorem 5.7 from \cite{Vaart} permits to treat the consistency of a general class of M-estimates. It is stated as follows:
\begin{theorem}
Let $M_n$ be random functions and let $M$ be a fixed function of $\phi$ such that for every $\varepsilon>0$
\begin{eqnarray}
\sup_{\phi\in\Phi} |M_n(\phi) - M(\phi)| \xrightarrow[]{\mathbb{P}} 0,\label{eqn:ConsistP1}\\
\inf_{\phi:\|\phi-\phi^T\|\geq\varepsilon} M(\phi) > M(\phi^T).\label{eqn:ConsistP2}
\end{eqnarray}
Then any sequence of estimators $\hat{\phi}_n$ with $M_n(\hat{\phi}_n)\leq M_n(\phi^T) - o_P(1)$ converges in probability to $\phi^T$.
\end{theorem}
In our approach, function $M_n$ corresponds to the criterion function $P_n H(P_n,\phi)$, where $H(P_n,\phi,y)$ is defined by:
\[
H(P_n,\phi,y) = \int{\varphi'\left(\frac{p_{\phi}}{K_{w}*P_n}\right)(x)p_{\phi}(x)dx} - \left[\frac{p_{\phi}(y)}{K_{w}*P_n(y)} \varphi'\left(\frac{p_{\phi}}{K_{w}*P_n}\right)(y) - \varphi\left(\frac{p_{\phi}}{K_{w}*P_n}\right)(y)\right].
\]
Function $M$ is simply defined by the \emph{expected}\footnote{In the literal sense and not mathematically.} limit in probability of $M_n$, since the Law of Large Numbers cannot be used because the average term is not a sum of i.i.d. random variables. It is given by $P_{\phi^T} h(P_{\phi^T},\phi)$ where $h(P_{\phi^T},\phi,x)$ is defined as:
\[h(P_{\phi^T},\phi,y) = \int{\varphi'\left(\frac{p_{\phi}}{p_{\phi^T}}\right)(x)p_{\phi}(x)dx} - \left[\frac{p_{\phi}(y)}{p_{\phi^T}(y)} \varphi'\left(\frac{p_{\phi}}{p_{\phi^T}}\right)(y) - \varphi\left(\frac{p_{\phi}}{p_{\phi^T}}\right)(y)\right]. \]
In order to prove (\ref{eqn:ConsistP1}), we propose to divide the argument into two parts. One can write:
\begin{equation}
\sup_{\phi\in\Phi} |P_nH(P_n,\phi) - P_{\phi^T}h(P_{\phi^T},\phi)| \leq \sup_{\phi\in\Phi} |P_nH(P_n,\phi) - P_nh(P_{\phi^T},\phi)| + \sup_{\phi\in\Phi} |P_{\phi^T}h(P_{\phi^T},\phi) - P_nh(P_{\phi^T},\phi)|.
\label{eqn:DecompProofIneq}
\end{equation}
Now, the second supremum tends to 0 in probability by the Glivenko-Cantelli theorem as soon as function $\phi\mapsto h(P_{\phi^T},\phi)$ is $P_{\phi^T}$--integrable, or more generally if $\{h(P,\phi), \phi\in\Phi\}$ is a Glivenko-Cantelli class of functions. The problem then resides in finding conditions under which the first supremum tends to 0 in probability. The remaining of the paragraph will be concerned with the search for such conditions. In the whole section concerning the consistency of our new estimator the window parameter $w$ is suppposed to depend directly on $n$ in order to be able to use Theorem 1 without any modifications. Besides, the construction of the estimator from (\ref{eqn:NewExactDualForm}) shows the explicit link of the window with $n$.\\
The following results are arranged in a way to give at first the most general case which one can offer. This result shows, according to our proof, that it is very difficult to derive a general and an applicable result in the same time. The ideas we provide are useful however to derive particular results according to a given divergence. We treat after that the case of divergences based on the Cressie-Read class of functions. Conditions of our result for this case still seem very restrictive. We finally discuss two particular cases when $\gamma$ is either in the interval $(0,1)$ or $(-1,0)$. Simpler conditions are derived and then verified in the gaussian model when we use a gaussian kernel.
%%%%%%%%%%%%%%%%%%%%%%%%%%%%%%%%%%%%%%%%%%%%%%%%%
\subsection{General Result}
We will derive in this paragraph a result which concerns the general class of divergence functions $\varphi$. Hereafter, simpler conditions will be proved for the particular class of Cressie-Read functions $\varphi_{\gamma}$. Let $\varphi^{\#}$ be the function $\varphi^{\#}(t) = t\varphi'(t)-\varphi(t)$, we then have:
\begin{multline*}
P_nH(P_n,\phi)-P_nh(P,\phi) = \int{\left[\varphi'\left(\frac{p_{\phi}}{K_{w}*P_n}\right)-\varphi'\left(\frac{p_{\phi}}{p_{\phi^T}}\right)\right](x)p_{\phi}(x)dx} \\ - \frac{1}{n}\sum_{i=1}^n{\varphi^{\#}\left(\frac{p_{\phi}}{K_{w}*P_n}\right)(y_i) - \varphi^{\#}\left(\frac{p_{\phi}}{p_{\phi^T}}\right)(y_i)}.
\end{multline*}
The key idea is to treat each term (the integral and the sum) separately and prove its uniform convergence in probability towards 0. Another important step is to apply the mean value theorem in order to transfer the difference from functions $\varphi'$ and $\varphi^{\#}$ into a difference between the kernel estimator and the true distribution where consistency of the former is exploited. We state now our general result:
\begin{theorem}
Assume that:
\begin{enumerate}
\item function $t\mapsto\varphi(t)$ is twice differentiable;
\item the kernel estimator is strongly consistent, i.e. $\sup_x\left|K_w*P_n(x) - p_{\phi^T}(x)\right|\rightarrow 0$ in probability;
\item function $x\mapsto \varphi^{\#}\left(\frac{p_{\phi}}{p_{\phi^T}(x)}\right)(x)$ is $P_T-$integrable for any $\phi$ in $\Phi$;
\item for any $\varepsilon>0$, there exists $n_0$ such that $\forall n\geq n_0$, the probability that the quantity 
\[\mathcal{A}_n = \sup_{\phi} \int{\frac{p_{\phi}^2(x)}{p_{\phi^T}(x) K_{w}*P_n(x)} \varphi''\left(\lambda_1(x)\frac{p_{\phi}}{K_{w}*P_n}(x)+(1-\lambda_1(x))\frac{p_{\phi}}{p_{\phi^T}}(x)\right) dx}\] 
is upper bounded independently of $n$, where $\lambda_1(x)\in(0,1)$, is greater than $1-\eta_n$ for $\eta_n\rightarrow 0$;
\item for any $\varepsilon>0$, there exists $n_0$ such that $\forall n\geq n_0$, the probability that the quantity 
\[\mathcal{B}_n = \sup_{\phi} \frac{1}{n}\sum_{i=1}^n{\frac{p_{\phi}}{p_{\phi^T} K_{w}*P_n}\left(\varphi^{\#}\right)'\left(\lambda_2(y_i)\frac{p_{\phi}}{K_{w}*P_n}+(1-\lambda_2(y_i))\frac{p_{\phi}}{p_{\phi^T}}\right)}\] 
is upper bounded, where $\lambda_2(y_i)\in(0,1)$, is greater than $1-\eta_n$ for $\eta_n\rightarrow 0$;
\item $\inf_{\phi:\|\phi-\phi^T\|\geq\varepsilon} P_T h(P_T,\phi) > P_T h(P_T,\phi^T)$,
\end{enumerate}
then the minimum dual $\varphi-$divergence estimator defined by (\ref{eqn:NewMDphiDE}) is consistent whenever it exists.
\end{theorem}
\begin{proof}
Let $\varepsilon>0$. We want to prove that $\lim_{n\rightarrow\infty} \mathbb{P}\left(\sup_{\phi\in\Phi} |P_nH(P_n,\phi) - P_nh(P_T,\phi)|<\varepsilon\right) = 1$. Since $\varphi$ is twice differentiable (which also implies the differentiability of $\varphi^{\#}$), then by the mean value theorem, there exist two functions $\lambda_1,\lambda_2:\mathbb{R}\rightarrow (0,1)$ such that:
\begin{eqnarray*}
\varphi'\left(\frac{p_{\phi}}{K_{w}*P_n}\right)-\varphi'\left(\frac{p_{\phi}}{p_{\phi^T}}\right) & = & \varphi''\left(\lambda_1(x)\frac{p_{\phi}}{K_{w}*P_n}+(1-\lambda_1(x))\frac{p_{\phi}}{p_{\phi^T}}\right) \left[\frac{p_{\phi}}{K_{w}*P_n} - \frac{p_{\phi}}{p_{\phi^T}}\right], \\
\varphi^{\#}\left(\frac{p_{\phi}}{K_{w}*P_n}\right)-\varphi^{\#}\left(\frac{p_{\phi}}{p_{\phi^T}}\right) & = & \left(\varphi^{\#}\right)'\left(\lambda_2(y_i)\frac{p_{\phi}}{K_{w}*P_n}+(1-\lambda_2(y_i))\frac{p_{\phi}}{p_{\phi^T}}\right) \left[\frac{p_{\phi}}{K_{w}*P_n} - \frac{p_{\phi}}{p_{\phi^T}}\right]. \\
\end{eqnarray*}
Let $n$ be sufficiently large such that:
\[\sup_{x}\left|K_{w}*P_n(x) - p_{\phi^T}(x)\right| \leq \min\left(\varepsilon,\frac{\varepsilon}{\mathcal{A}_n},\frac{\varepsilon}{\mathcal{B}_n}\right) \]
where $\mathcal{A}_n$ and $\mathcal{B}_n$ are as described in 4 and 5 in the theorem.
%\begin{eqnarray*}
%\mathcal{A}_n & = & \sup_{\phi} \int{\frac{p_{\phi}^2}{p_{\phi^T} K_{w}*P_n} \varphi''\left(\lambda_1(x)\frac{p_{\phi}}{K_{w}*P_n}+(1-\lambda_1(x))\frac{p_{\phi}}{p_{\phi^T}}\right) dx} \\
%\mathcal{B}_n & = & \sup_{\phi} \frac{1}{n}\sum_{i=1}^n{\frac{p_{\phi}}{p_{\phi^T} K_{w}*P_n}\left(\varphi^{\#}\right)'\left(\lambda_2(y_i)\frac{p_{\phi}}{K_{w}*P_n}+(1-\lambda_2(y_i))\frac{p_{\phi}}{p_{\phi^T}}\right)}
%\end{eqnarray*}
Provided that constants $\mathcal{A}_n$ and $\mathcal{B}_n$ exist and are bounded independently of $n$, this event occurs with probability $1-\eta_n$ with $\eta_n\rightarrow 0$ by the strong consistency assumption (point 2). This implies that both events:
\begin{eqnarray*}
\left|\int{\left[\varphi'\left(\frac{p_{\phi}}{K_{w}*P_n}\right)-\varphi'\left(\frac{p_{\phi}}{p_{\phi^T}}\right)\right]p_{\phi}}\right| &\leq & \frac{\varepsilon}{\mathcal{A}_n}\int{\frac{p_{\phi}^2}{p_{\phi^T} K_{w}*P_n} \varphi''\left(\lambda_1(x)\frac{p_{\phi}}{K_{w}*P_n}+(1-\lambda_1(x))\frac{p_{\phi}}{p_{\phi^T}}\right) dx}\\
& \leq & \varepsilon, \\
\left|\frac{1}{n}\sum_{i=1}^n{\varphi^{\#}\left(\frac{p_{\phi}}{K_{w}*P_n}\right)(y_i) - \varphi^{\#}\left(\frac{p_{\phi}}{p_{\phi^T}}\right)(y_i)}\right| & \leq & \frac{\varepsilon}{\mathcal{B}_n} \frac{1}{n}\sum_{i=1}^n{\frac{p_{\phi}}{p_{\phi^T} K_{w}*P_n}\left(\varphi^{\#}\right)'\left(\lambda_2(y_i)\frac{p_{\phi}}{K_{w}*P_n}+(1-\lambda_2(y_i))\frac{p_{\phi}}{p_{\phi^T}}\right)(y_i)} \\
& \leq & \varepsilon
\end{eqnarray*}
happen with probability greater than $1-\eta_n$ independently of $\phi$. Finally, we conclude that
\[\mathbb{P}\left(\sup_{\phi\in\Phi} |P_nH(P_n,\phi) - P_nh(P,\phi)|<2\varepsilon\right) \geq 1-\eta_n,\]
and hence $\sup_{\phi\in\Phi} |P_nH(P_n,\phi) - P_nh(P,\phi)|\rightarrow 0$ in probability. To end the proof, we use assumption 3 of the present theorem together with the Glivenko-Cantelli theorem to conclude that $\sup_{\phi\in\Phi} |P_{\phi^T}h(P_{\phi^T},\phi) - P_nh(P_{\phi^T},\phi)|\rightarrow 0$ in probability. Using inequality \ref{eqn:DecompProofIneq}, we conclude that $\sup_{\phi\in\Phi} |P_nH(P_n,\phi) - P_{\phi^T}h(P_{\phi^T},\phi)| \rightarrow 0$ in probability. We end with the use of Theorem 1. Condition (\ref{eqn:ConsistP1}) is verified by the previous arguments, and Condition (\ref{eqn:ConsistP2}) is what we have assumed in point 6 of the present theorem. By definition of the kernel-based MD$\varphi$DE as a minimum of the criterion function $\phi\mapsto P_nH(P_n,\phi)$ Theorem 1 entails the consistency of our new estimator.
\end{proof}

This result is very general since function $\varphi$ is only supposed to be twice differentiable\footnote{Recall that $\varphi$ should also verify other conditions related to the notion of $\varphi-$divergences as mentionned in the introduction.}. For consistency results one can consult for example \cite{WiedWeibbach}, \cite{ZambomDias} or \cite{Libengue} Chap. 1 for a brief survey on symmetric kernels. If one is using asymmetric kernels, unfortunately consistency is only proved on every compact subset of the support of the distribution function, see \cite{Bouezmarni} or \cite{Libengue} Chap. 3 for a more general approach. On the other hand, it is not simple to verify conditions 4 and 5 for the general class of functions $\varphi$, and one may derive for his own case study a simpler set of conditions on the basis of this result. For condition 4, if one is using for example the $\chi^2$ divergence, $\varphi''(t)=1$ so that function $\lambda_1$ is no longer there and the expression of $\mathcal{A}_n$ is simplified. The main subtlety in condition 5 is that the sum is over strongly dependent random variables. We will see in the case of divergences with $\varphi=\varphi_{\gamma}$ for $\gamma\in(-1,0)$ that this sum becomes over only i.i.d. random variables and is simple to be taken care of.\\
Assumption 6 means that function $\phi\mapsto P_Th(P_T,\phi)$ has a unique and well separated minimum. Uniqueness is already in our hands since function $\phi\mapsto P_Th(P_T,\phi)$ is non other than the dual representation (with the supremum calculated) of the $\varphi-$divergence $D_{\varphi}(p_{\phi},p_{\phi^T})$. Using the property that $D_{\varphi}(p_{\phi},p_{\phi^T})=0$ iff $p_{\phi}=p_{\phi^T}$, uniqueness is immediately verified as long as the model is identifiable.
%%%%%%%%%%%%%%%%%%%%%%%%%%%%%%%%%%%%%%%%
%%%%%%%%%%%%%%%%%%%%%%%%%%%%%%%%%%%%%%%%%%
\subsection{General Result for Power Divergences}
Power divergences are the divergences defined through the class of Cressie-Read functions defined by:
\[\varphi_{\gamma}(t) = \frac{t^{\gamma} - \gamma t + \gamma - 1}{\gamma(\gamma-1)}.\]
The kernel based MD$\varphi$DE is defined as:
\[\hat{\phi}_{n} = \arginf_{\phi\in\Phi} \frac{1}{\gamma-1}\int{\frac{p_{\phi}^{\gamma}}{\left(K_{w}*P_n\right)^{1-\gamma}}(x)dx} - \frac{1}{n\gamma}\sum_{i=1}^n{\left(\frac{p_{\phi}(y_i)}{K_{w}*P_n(y_i)}\right)^{\gamma} }.\]
The idea here is to get rid of $p_{\phi}$ from many places and replace the derivatives of $\varphi'$ and $\varphi^{\#}$ with a more explicit formulas. The idea of using the mean value theorem is kept, but this time it is applied on simpler functions. Here we have:
\begin{multline*}
P_nH(P_n,\phi)-P_nh(P,\phi) = \frac{1}{\gamma-1}\int{\frac{\left(K_{w}*P_n\right)^{1-\gamma}-p_{\phi^T}^{1-\gamma}}{p_{\phi}^{-\gamma}}(x)dx} - \frac{1}{n\gamma}\sum_{i=1}^n{\frac{\left(K_{w}*P_n\right)^{-\gamma}-p_{\phi^T}^{-\gamma}}{p_{\phi}^{-\gamma}}(y_i)}.
\end{multline*}
\begin{theorem}
For the class of power divergences defined through the class of Cressie-Read functions $\varphi_{\gamma}$, assume that:
\begin{enumerate}
\item the kernel estimator is strongly consistent, i.e. $\sup_x\left|K_w*P_n(x) - p_{\phi^T}(x)\right|\rightarrow 0$ in probability;
\item function $x\mapsto \left(\frac{p_{\phi}}{p_{\phi^T}(x)}\right)^{\gamma}(x)$ is $P_T-$integrable for any $\phi$ in $\Phi$;
\item for any $\varepsilon>0$, there exists $n_0$ such that $\forall n\geq n_0$, the probability that the quantity 
\[\mathcal{A}_n = \sup_{\phi} \int{\frac{ \left[\lambda_1(x)K_w*P_n(x) + (1-\lambda_1(x))p_{\phi^T}(x)\right]^{-\gamma}}{p_{\phi}^{-\gamma}(x)} dx}\] 
is upper bounded independently of $n$, where $\lambda_1(x)\in(0,1)$, is greater than $1-\eta_n$ for $\eta_n\rightarrow 0$;
\item for any $\varepsilon>0$, there exists $n_0$ such that $\forall n\geq n_0$, the probability that the quantity 
\[\mathcal{B}_n = \sup_{\phi}\frac{1}{n}\sum_{i=1}^n{\frac{ \left[\lambda_2(y_i)K_w*P_n(y_i) + (1-\lambda_2(y_i))p_{\phi^T}(y_i)\right]^{-\gamma-1}}{p_{\phi}^{-\gamma}(y_i)}}\] 
is upper bounded independently of $n$, where $\lambda_2(y_i)\in(0,1)$, is greater than $1-\eta_n$ for $\eta_n\rightarrow 0$;
\item $\inf_{\phi:\|\phi-\phi^T\|\geq\varepsilon} P_T h(P_T,\phi) > P_T h(P_T,\phi^T)$,
\end{enumerate}
then the minimum dual $\varphi-$divergence estimator defined by (\ref{eqn:NewMDphiDE}) is consistent whenever it exists.
\end{theorem}
\begin{proof}
Let $x$ and $a\neq 0,1$ be real numbers. By the mean value theorem, there exists $\lambda(x)\in(0,1)$ such that:
\[\left(K_w*P_n\right)^{a}(x) - p_{\phi^T}^{a}(x) = a \left[\lambda(x)K_w*P_n(x) + (1-\lambda(x))p_{\phi^T}(x)\right]^{a-1}\left(K_w*P_n(x) - p_{\phi^T}(x)\right)\]
This implies both identities:
\begin{eqnarray*}
\frac{\left(K_w*P_n\right)^{-\gamma+1}(x) - p_{\phi^T}^{-\gamma+1}(x)}{p_{\phi}^{-\gamma}(x)} & = & (1-\gamma)\frac{ \left[\lambda_1(x)K_w*P_n(x) + (1-\lambda_1(x))p_{\phi^T}(x)\right]^{-\gamma}}{p_{\phi}^{-\gamma}(x)} \times \left(K_w*P_n(x) - p_{\phi^T}(x)\right)\\
\frac{\left(K_w*P_n\right)^{-\gamma}(y_i) - p_{\phi^T}^{-\gamma}(y_i)}{p_{\phi}^{-\gamma}(y_i)} & = & -\gamma\frac{ \left[\lambda_2(y_i)K_w*P_n(y_i) + (1-\lambda_2(y_i))p_{\phi^T}(y_i)\right]^{-\gamma-1}}{p_{\phi}^{-\gamma}(y_i)}  \times \left(K_w*P_n(y_i) - p_{\phi^T}(y_i)\right)
\end{eqnarray*}
Let $n$ be sufficiently large such that:
\[\sup_{x}\left|K_w*P_n(x) - p_{\phi^T}(x)\right| \leq \min\left(\varepsilon, \frac{\varepsilon}{\mathcal{A}_n}, \frac{\varepsilon}{\mathcal{B}_n}\right)\]
where $\mathcal{A}_n$ and $\mathcal{B}_n$ as defined in the theorem. This event occurs with probability greater than $1-\eta_n$ with $\eta_n\rightarrow 0$ by the strong consistency of the kernel. The remaining of the arguments is the same as for Theorem 2.
\end{proof}

Quantities $\mathcal{A}_n$ and $\mathcal{B}_n$ can be simplified according to the range of values of $\gamma$ in a way that $\lambda_1$ and $\lambda_2$ play no role in the calculus. These functions have no explicit formulas in general and the only information we have in hand is that they take their values in $(0,1)$. Indeed, if $\gamma>0$, using the convexity of function $t\mapsto t^{-\gamma}$ over $\mathbb{R}_+$ and the fact that both quantities $\lambda_1(x)$ and $1-\lambda_1(x)$ are upper bounded by 1, we may write:
\[\mathcal{A}_n \leq \sup_{\phi} \int{\frac{ \left[K_w*P_n(x)\right]^{-\gamma}}{p_{\phi}^{-\gamma}(x)} dx} + \int{\frac{ p_{\phi^T}(x)^{-\gamma}}{p_{\phi}^{-\gamma}(x)} dx}.\]
If $\gamma <0$, one may use the increasing property of function $t\mapsto t^{-\gamma}$ on $\mathbb{R}_+$ to deduce that:
\[\mathcal{A}_n \leq \sup_{\phi} \int{\frac{ \left[K_w*P_n(x)+p_{\phi^T}(x)\right]^{-\gamma}}{p_{\phi}^{-\gamma}(x)} dx}. \]
Moreover, for values of $\gamma$ in $(-1,0)$, we may use Jensen's inequality to go further and write:
\[\mathcal{A}_n \leq \sup_{\phi} \left(\int{\frac{ K_w*P_n(x)+p_{\phi^T}(x)}{p_{\phi}(x)} dx}\right)^{-\gamma}. \]
Similar upper bounds can be established for $\mathcal{B}_n$. When $\gamma>-1, \gamma\neq 0,1$, we use again the convexity of function $t\mapsto t^{-\gamma-1}$ to get the following upper bound:
\[\mathcal{B}_n \leq \sup_{\phi}\frac{1}{n}\sum_{i=1}^n{\frac{ \left[K_w*P_n(y_i)\right]^{-\gamma-1} + p_{\phi^T}(y_i)^{-\gamma-1}}{p_{\phi}^{-\gamma}(y_i)}}.\]
Finally, when $\gamma\leq-1$, we use the increasing property of function $t\mapsto t^{-\gamma-1}$ over $\mathbb{R}_+$. We get:
\[\mathcal{B}_n \leq \sup_{\phi}\frac{1}{n}\sum_{i=1}^n{\frac{ \left[K_w*P_n(y_i) + p_{\phi^T}(y_i)\right]^{-\gamma-1}}{p_{\phi}^{-\gamma}(y_i)}}.\]
%%%%%%%%%%%%%%%%%%%%%%%%%%%%%%%%%%%%%%%
\subsection{Case of power divergences with $\gamma\in(0,1)$}
In this paragraph, we try to derive simpler conditions than those in Theorems 2 and 3 . For the class of Cressie-Read defined by $\varphi_{\gamma}$ with $\gamma\in(0,1)$, the kernel-based MD$\varphi$DE has the form:
\begin{multline}
\hat{\phi}_{n} = \arginf_{\phi\in\Phi} \frac{1}{\gamma-1}\int{\left(K_{w}*P_n\right)^{1-\gamma}(x)p_{\phi}^{\gamma}(x)dx} - \frac{1}{n\gamma}\sum_{i=1}^n{\left(\frac{p_{\phi}(y_i)}{K_{w}*P_n(y_i)}\right)^{\gamma} } - \frac{1}{\gamma(\gamma-1)}.
\label{eqn:NewMDphiDEPDbis}
\end{multline}
Our main problem is always the study of the difference $P_nH(P_n,\phi)-P_nh(P,\phi)$ which is given by:
\begin{multline*}
P_nH(P_n,\phi)-P_nh(P,\phi) = \frac{1}{\gamma-1}\int{\left[\left(K_{w}*P_n\right)^{1-\gamma}-p_{\phi^T}^{1-\gamma}\right](x)p_{\phi}(x)dx} \\ - \frac{1}{n\gamma}\sum_{i=1}^n{\left(\frac{p_{\phi}}{K_{w}*P_n\times p_{\phi^T}}\right)^{\gamma}(y_i) \left(p_{\phi^T}^{\gamma}(y_i) - [K_{w}*P_n]^{\gamma}(y_i)\right)}.
\end{multline*}
The key idea here is to use the uniform continuity of both functions $t\mapsto t^{\gamma}$ and $t\mapsto t^{1-\gamma}$. 
\begin{theorem}
For power divergences defined by $\varphi_{\gamma}$ with $\gamma\in(0,1)$, suppose that:
\begin{enumerate}
\item the kernel estimator is strongly consistent, i.e. $\sup_x\left|K_w*P_n(x) - p_{\phi^T}(x)\right|\rightarrow 0$ in probability;
\item function $x\mapsto \left(\frac{p_{\phi}}{p_{\phi^T}(x)}\right)^{\gamma}(x)$ is $P_T-$integrable for any $\phi$ in $\Phi$;
\item the quantity $\mathcal{A} = \sup_{\phi}\int{p_{\phi}^{\gamma}(x)dx}$ is upper bounded;
\item for any $\varepsilon>0$, there exists $n_0$ such that $\forall n\geq n_0$, the probability that the quantity 
\[\mathcal{B}_n = \sup_{\phi}\frac{1}{n}\sum_{i=1}^n{\left(\frac{p_{\phi}}{K_{w}*P_n\times p_{\phi^T}}\right)^{\gamma}(y_i)}\] 
is upper bounded independently of $n$ is greater than $1-\eta_n$ for $\eta_n\rightarrow 0$;
\item $\inf_{\phi:\|\phi-\phi^T\|\geq\varepsilon} P_T h(P_T,\phi) > P_T h(P_T,\phi^T)$,
\end{enumerate}
then the minimum dual $\varphi-$divergence estimator defined by (\ref{eqn:NewMDphiDEPDbis}) is consistent whenever it exists.
\end{theorem}
\begin{proof} 
In order to prove consistency of the kernel-based MD$\varphi$DE, we follow the same steps in Theorem 2. To Verify condition (\ref{eqn:ConsistP1}), we use the decomposition (\ref{eqn:DecompProofIneq}). The second term in (\ref{eqn:DecompProofIneq}) goes to zero in probability using assumption 2 of the present theorem and the Glivenko-Cantelli theorem. In order to treat the first term, we use the uniform continuity of both functions $t\mapsto t^{\gamma}$ and $t\mapsto t^{1-\gamma}$. If
\[\left|p_{\phi^T}(x) - K_{w}*P_n(x)\right|<\delta_1(\varepsilon),\]
then
\begin{equation}
\left| \left(K_{w}*P_n\right)^{\gamma}(x)-\left(p_{\phi^T}\right)^{\gamma}(x) \right| <  \frac{\varepsilon\gamma}{\mathcal{B}_n}.
\label{eqn:HellingerLikeCond1}
\end{equation}
where $\mathcal{B}_n$ is as given here above. On the other hand, uniform continuity of $t\mapsto t^{1-\gamma}$ entails that if
\[\left|p_{\phi^T}(x) - K_{w}*P_n(x)\right|<\delta_2(\varepsilon),\]
then
\begin{equation}
\left| \left(K_{w}*P_n\right)^{1-\gamma}(x)-\left(p_{\phi^T}\right)^{1-\gamma}(x) \right| <  \frac{\varepsilon(1-\gamma)}{\sup_{\phi}\int{p_{\phi}^{\gamma}(x)dx}}.
\label{eqn:HellingerLikeCond2}
\end{equation}
Let $n$ be sufficiently large such that 
\[\sup_{x}|p_{\phi^T}(x) - K_{w}*P_n(x)| < \min\left(\delta_1(\varepsilon),\delta_2(\varepsilon)\right). \]
By the strong consistency assumption for the kernel estimator, this event happens with probability greater than $1-\eta_n$ where $\eta_n\rightarrow 0$. Thus, each of (\ref{eqn:HellingerLikeCond1}) and (\ref{eqn:HellingerLikeCond2}) occur with probability greater than $1-\eta_n$ where $\eta_n\rightarrow 0$.\\
The remaining of the argument is exactly the same as for Theorem 2.
\end{proof}

This result is clearly far more simpler than the one given in Theorem 2. Condition 3 here is already independent of $n$ and is deterministic which is not the case for assumption 4 in theorem 2 and assumption 3 in Theorem 3. Although assumption 4 does not contain unknown functions such as $\lambda_2$ defined for previous results, it has always a similar difficulty since it concerns a sum of strongly dependent terms.
%%%%%%%%%%%%%%%%%%%%%%%%%%%%%%%%%%%%%%%
\subsection{Case of power divergences with $\gamma\in(-1,0)$}
This time we will use the uniform continuity of function $t\mapsto t^{-\gamma}$ to prove that if $|K_w*P_n(x) - p_{\phi^T}(x)|<\delta_2$, then:
\[|K_w*P_n(x) - p_{\phi^T}(x)|<\frac{\varepsilon}{\sup_{\phi}\frac{1}{n}\sum{p_{\phi}^{\gamma}(y_i)}}\]
Thus, $\mathcal{B}_n$ of Theorem 3 is now replaced by the quantity 
\[\mathcal{B}_n = \sup_{\phi}\frac{1}{n}\sum_{i=1}^n{p_{\phi}^{\gamma}(y_i)}.\]
On the other hand, we rewrite the integral difference as follows:
\[\int{\frac{\left(K_w*P_n\right)^{-\gamma+1}(x) - p_{\phi^T}^{-\gamma+1}(x)}{p_{\phi}^{-\gamma}(x)}dx} = \int{\frac{\left[\left(K_w*P_n\right)^{\frac{-\gamma+1}{2}}(x) - p_{\phi^T}^{\frac{-\gamma+1}{2}}(x)\right]\left[\left(K_w*P_n\right)^{\frac{-\gamma+1}{2}}(x) + p_{\phi^T}^{\frac{-\gamma+1}{2}}(x)\right]}{p_{\phi}^{-\gamma}(x)}dx}.\]
Now, using the uniform continuity of function\footnote{notice that $\frac{-\gamma+1}{2}\in(0,1)$ since $\gamma\in(-1,0)$.} $t\mapsto t^{\frac{-\gamma+1}{2}}$, we may deduce that if $|K_w*P_n(x) - p_{\phi^T}(x)|<\delta_1$, then:
\[|K_w*P_n(x) - p_{\phi^T}(x)|<\frac{\varepsilon}{\sup_{\phi}\int{\frac{\left(K_w*P_n\right)^{\frac{-\gamma+1}{2}}(x) + p_{\phi^T}^{\frac{-\gamma+1}{2}}(x)}{p_{\phi}^{-\gamma}(x)}}} \]
Thus $\mathcal{A}_n$ of Theorem 3 is now replaced by the simpler quantity:
\[\mathcal{A}_n = \sup_{\phi}\int{\frac{\left(K_w*P_n\right)^{\frac{-\gamma+1}{2}}(x) + p_{\phi^T}^{\frac{-\gamma+1}{2}}(x)}{p_{\phi}^{-\gamma}(x)}}\]
The remaining of the argument follows similarly to previous theorem. We may state the following result.
\begin{theorem}
For the class of power divergences defined through the class of Cressie-Read functions $\varphi_{\gamma}$, assume that:
\begin{enumerate}
\item the kernel estimator is strongly consistent, i.e. $\sup_x\|K_w*P_n(x) - p_{\phi^T}(x)\|\rightarrow 0$ in probability;
\item function $x\mapsto \left(\frac{p_{\phi}}{p_{\phi^T}(x)}\right)^{\gamma}(x)$ is $P_T-$integrable for any $\phi$ in $\Phi$;
\item for any $\varepsilon>0$, there exists $n_0$ such that $\forall n\geq n_0$, the probability that the quantity 
\[\mathcal{A}_n = \sup_{\phi}\int{\frac{\left(K_w*P_n\right)^{\frac{-\gamma+1}{2}}(x) + p_{\phi^T}^{\frac{-\gamma+1}{2}}(x)}{p_{\phi}^{-\gamma}(x)}}\] 
is upper bounded independently of $n$ is greater than $1-\eta_n$ for $\eta_n\rightarrow 0$;
\item for any $\varepsilon>0$, there exists $n_0$ such that $\forall n\geq n_0$, the probability that the quantity 
\[\mathcal{B}_n = \sup_{\phi}\frac{1}{n}\sum_{i=1}^n{p_{\phi}^{\gamma}(y_i)}\] 
is upper bounded independently of $n$ is greater than $1-\eta_n$ for $\eta_n\rightarrow 0$;
\item $\inf_{\phi:\|\phi-\phi^T\|\geq\varepsilon} P_T h(P_T,\phi) > P_T h(P_T,\phi^T)$,
\end{enumerate}
then the minimum dual $\varphi-$divergence estimator defined by (\ref{eqn:NewMDphiDE}) is consistent whenever it exists.
\end{theorem}
This last result is the least complicated one among previous ones, since in the one hand, there is no unknown functions such as $\lambda_1$ or $\lambda_2$. On the other hand, the sum in $\mathcal{B}_n$ is over i.i.d. terms. According to the model and to the value of $\gamma$ in the interval $(0,1)$, one may either use the results of Theorems 3 or 4. The two general results are clearly restrictive, and one should for his own particular case study derive his set of conditions. Those results stay as a guide to proving further ones.\\
The remaining of this section is devoted to show how in a gaussian model, consistency of the kernel-based MD$\varphi$DE can be proved.

%%%%%%%%%%%%%%%%%%%%%%%%%%%%%%%%%%%%%%%%%ùù
\begin{example}
We take a simple and ordinary example of a gaussian model with unknown mean $\mu$ which is supposed to be in a close interval $[\mu_{\min},\mu_{\max}]$. We consider power divergences for which $\gamma\in(-1,0)$. The gaussian kernel is used. Assumption 1 is easily checked by considering the list of conditions in Theorem A in \cite{Silverman}. Assumption 2 is also very simple since
\[\left(\frac{p_{\phi}^{\gamma}}{p_{\phi^T}^{\gamma}}p_{\phi^T}(x) = e^{-\frac{1}{2}x^2 - \mu y + \frac{1}{2}\mu^2}\right).\]
We use Theorem 5 to prove consistency. We calculate constants $\mathcal{A}_n$ and $\mathcal{B}_n$. By the strong consistency of the kernel estimator, it suffices for $\mathcal{A}_n$ to study boundedness of the term which contains $p_{\phi^T}$.
\[\int{\frac{p_{\phi^T}^{\frac{-\gamma+1}{2}}(x)}{p_{\phi}^{-\gamma}(x)}} = c_1(\gamma)e^{\frac{\gamma^2-\gamma}{2(1+\gamma)}\mu^2}\]
for a constant $c_1$. This quantity is bounded since $\mu$ is supposed to be in a closed interval. Hence $\mathcal{A}_n$ is bounded and assumption 3 is now verified.\\
On the other hand, in order to study $\mathcal{B}_n$, it suffices to consider the quantity $\sup_{\phi}\int{p_{\phi}^{\gamma}p_{\phi^T}}$ by vertue of the Glivenko-Cantelli theorem\footnote{The Glivenko-Cantelli theorem states that both quantities $\sup_{\phi}\frac{1}{n}\sum_{i=1}^n{p_{\phi}^{\gamma}(y_i)}$ and $\sup_{\phi}\int{p_{\phi}^{\gamma}p_{\phi^T}}$ are uniformly close for sufficiently large $n$ independently of $\phi$, hence boundedness of either of them implies boundedness of the other.}. We have:
\[\int{p_{\phi}^{\gamma}p_{\phi^T}} = c_2(\gamma)e^{-\frac{-\gamma}{1+\gamma}\frac{\mu^2}{2}}.\]
for a constant $c_2$. Here again, since $\mu$ is supposed to be in a closed interval, the previous quantity is bounded. This entails that $\mathcal{B}_n$ is bounded and assumption 4 is fulfilled. \\
We move now to the last assumption. By the dual representation of the divergence, we have $P_Th(P_T,\phi) = D_{\varphi}(p_{\phi},p_{\varphi^T})$. This implies that :
\[P_Th(P_T,\phi) = \frac{1}{\gamma(\gamma-1)}e^{\frac{\gamma^2-\gamma}{2}\mu^2} - \frac{1}{\gamma(\gamma-1)}.\]
This function clearly verifies assumption 5 since it has a minimum at $\mu=0$ and this minimum is well separated.
\end{example}
\begin{remark}
The argument concerning the boundedness of the term $\int{(K_w*P_n)^{\frac{1-\gamma}{2}}/p_{\phi}^{-\gamma}}$ given above is not precise. We will give a more accurate one for the interested. By Jensen's inequality, one may write:
\begin{eqnarray*}
\int{\frac{(K_w*P_n)^{\frac{1-\gamma}{2}}(y)}{p_{\phi}^{-\gamma}(y)}dy} & = & \int{\left(\frac{K_w*P_n}{p_{\phi}^{\frac{-2\gamma}{1-\gamma}}}\right)^{\frac{1-\gamma}{2}}(y)dy} \\
 & \leq & \left(\int{\frac{K_w*P_n}{p_{\phi}^{\frac{-2\gamma}{1-\gamma}}}(y)dy}\right)^{\frac{1-\gamma}{2}}\\
 & \leq & e^{-\frac{\gamma}{2}\mu^2}\left(\frac{1}{nw}\sum_{i=1}^n{e^{-\frac{y_i^2}{2w^2}} \int{\exp\left[-\frac{1}{2}\left(\frac{1}{w^2}+\frac{2\gamma}{1-\gamma}\right)y^2 + \left(\frac{y_i}{w^2}+\frac{2\gamma\mu}{1-\gamma}\right)y\right]dy}}\right)^{\frac{1-\gamma}{2}}.
\end{eqnarray*}
We calculate each integral separately:
\[
\int{\exp\left[-\frac{1}{2}\left(\frac{1}{w^2}+\frac{2\gamma}{1-\gamma}\right)y^2 + \left(\frac{y_i}{w^2}+\frac{2\gamma\mu}{1-\gamma}\right)y\right]dy} = w c_3(\gamma,w) \exp\left[\frac{\left(y_i+\frac{2\gamma w^2}{1-\gamma}\mu\right)^2}{2 w^2\left(1+\frac{2\gamma w^2}{1-\gamma}\right)}\right],
\]
where $c_3(\gamma,w)=\sqrt{\frac{1-\gamma}{1-\gamma +2\gamma w^2}}$. We now proceed to estimate the sum over $i$:
\begin{eqnarray*}
\frac{1}{nw}\sum_{i=1}^n{e^{-\frac{y_i^2}{2w^2}} \int{\exp^{-\frac{1}{2}\left(\frac{1}{w^2}+\frac{2\gamma}{1-\gamma}\right)y^2 + \left(\frac{y_i}{w^2}+\frac{2\gamma\mu}{1-\gamma}\right)y}dy}} & = & \frac{1}{n}\sum_{i=1}^n{e^{-\frac{y_i^2}{2w^2}} c_3(\gamma,w) \exp\left[\frac{\left(y_i+\frac{2\gamma w^2}{1-\gamma}\mu\right)^2}{2 w^2\left(1+\frac{2\gamma w^2}{1-\gamma}\right)}\right]} \\
 & = & \frac{1}{n}c_3(\gamma,w)e^{\frac{2w^2\gamma^2\mu^2}{(1-\gamma)(1-\gamma+2\gamma w^2)}}\sum_{i=1}^n{e^{-\frac{\gamma}{1-\gamma}y_i^2 + \frac{2\gamma \mu}{1-\gamma+2\gamma w^2}y_i}} \\
& \leq & \frac{1}{n}c_3(\gamma,w)e^{\frac{2w^2\gamma^2\mu_{\max}^2}{(1-\gamma)(1-\gamma+2\gamma w^2)}}\sum_{i=1}^n{e^{-\frac{\gamma}{1-\gamma}y_i^2 + \frac{2\gamma \mu_{\min}}{1-\gamma+2\gamma w^2}y_i}}.
\end{eqnarray*}
The final step is to use a version of the law of large numbers for independent random variables such as the two series theorem of Kolomogrov (see \cite{Feller} Chap VII, Theorem 3) since the terms of the sum do not have the same probability law, but guided by the standard gaussian law. The general term of the mean and its first two moments are:
\begin{eqnarray*}
Z_i & = & \exp\left[-\frac{\gamma}{1-\gamma}y_i^2 + \frac{2\gamma \mu_{\min}}{1-\gamma+2\gamma w^2}y_i\right] \\
\mathbb{E}[Z_i] & = & \sqrt{\frac{1-\gamma}{1+\gamma}}\exp\left[\frac{1-\gamma}{1+\gamma}\frac{\gamma \mu_{\min}}{1-\gamma+2\gamma w^2}\right] \\
\mathbb{E}[Z_i^2] & = & \sqrt{\frac{1-\gamma}{1+3\gamma}}\exp\left[\frac{1-\gamma}{1+3\gamma}\frac{4\gamma \mu_{\min}}{1-\gamma+2\gamma w^2}\right].
\end{eqnarray*}
The variance exists only when $\gamma\in(-\frac{1}{3},0)$. It results that for this range of values, the Kolomogrov two series theorem applies and the average $\frac{1}{n}\sum_{i=1}^n{e^{-\frac{\gamma}{1-\gamma}y_i^2 + \frac{2\gamma \mu_{\min}}{1-\gamma+2\gamma w^2}y_i}}$ now converges in probability independently of $\mu$. Besides, the remaining factor$c_3(\gamma,w)e^{\frac{2w^2\gamma^2\mu_{\max}^2}{(1-\gamma)(1-\gamma+2\gamma w^2)}}$ also converges as $n$ goes to infinity (and $w$ goes to zero) to a constant (equal to 1).  Thus, boundedness of $\int{\frac{(K_w*P_n)^{\frac{1-\gamma}{2}}(y)}{p_{\phi}^{-\gamma}(y)}dy}$ is ensured.
\end{remark}
%%%%%%%%%%%%%%%%%%%%%%%%%%%%%%%%%%%%%%%%%%%%%%%%%%%%%%%%%%%%%%
\begin{example}
Let's take again the example of a gaussian model with a mean parameter $\mu$ unknown. Consider the class of power divergences with $\gamma\in(0,1)$. We verify assumptions of Theorem 4. We suppose also that the true distribution is the standard gaussian law $\mathcal{N}(0,1)$. Let's consider a gaussian kernel. We have:
\[K_w*P_T(x) = \frac{1}{\sqrt{2\pi(1+w^2)}} e^{-\frac{x^2}{2(1+w^2)}}.\]
In this example, it suffices to study consistency of the kernel-based MD$\varphi$DE for a fixed window. Indeed, the minimum of $P_Th(P_T,\phi)$ concides with the minimum of function $P_TH(P_T,\phi)$. This entails that if for a fixed window, the kernel-based MD$\varphi$DE is consistent with respect to the minimum of $P_TH(P_T,\phi)$, so does it with respect to the minimum of $P_Th(P_T,\phi)$ which is the true parameter $\mu^T$. The corresponding list of conditions can easily be derived from Theorem 4. Indeed, consistency of the kernel is no longer needed \footnote{In the Basu-Lindsay approach, this can happen if one can find a \emph{transparent} kernel.}. Points 2 and 3 are kept as they are. We replace $p_{\phi^T}$ in assumption 4 by the smoothed distribution $K_w*P_{\phi^T}$. Point 5 becomes with respect to $P_TH(P_T,\phi)$ instead of $P_Th(P_T,\phi)$. The arguments of the proof are the same. We only need to use the Glivenk-Cantelli theorem instead of the strong consistency of the kernel. Notice that point 4 is very hard so that we only verify it when $w^2>\frac{1}{2}$.\\
We first check our claim that both $P_TH(P_T,\mu)$ and $P_Th(P_T,\mu)$ have the same minimum. The minimum of $P_Th(P_T,\mu)$ is attained when $\mu=0$. We calculate an exact form of function $P_TH(P_T,\mu)$. We have:
\begin{eqnarray*}
\int_{\mathbb{R}}{\frac{p_{\mu}^{\gamma}(x)}{p_{\phi^T}^{\gamma-1}(x)}dx} & = & \sqrt{\frac{1+w^2}{1+\gamma w^2}} e^{-\frac{\gamma(1-\gamma)}{2(1+\gamma w^2)}\mu^2}\\
\int_{\mathbb{R}}{\frac{p_{\mu}^{\gamma}(x)}{p_{\phi^T}^{\gamma}(x)} p_T(x)dx} & = & \sqrt{\frac{1+w^2}{(\gamma+1)w^2+1}} e^{-\frac{\gamma(w^2+1-\gamma)}{2(1+(\gamma+1)w^2)}\mu^2}
\end{eqnarray*}
and thus:
\[P_TH(P_T,\mu)=\frac{1}{\gamma-1}\sqrt{\frac{1+w^2}{1+\gamma w^2}} e^{-\frac{\gamma(1-\gamma)}{2(1+\gamma w^2)}\mu^2} - \frac{1}{\gamma}\sqrt{\frac{1+w^2}{(\gamma+1)w^2+1}} e^{-\frac{\gamma(w^2+1-\gamma)}{2(1+(\gamma+1)w^2)}\mu^2} - \frac{1}{\gamma(\gamma-1)}.\]
Figure \ref{fig:ObjFun} shows the curve of this function for several values of $w$ and $\gamma$. It is clear that the infimum is unique and nicely separated.
\begin{figure}[ht]
\centering
\includegraphics[scale=0.4]{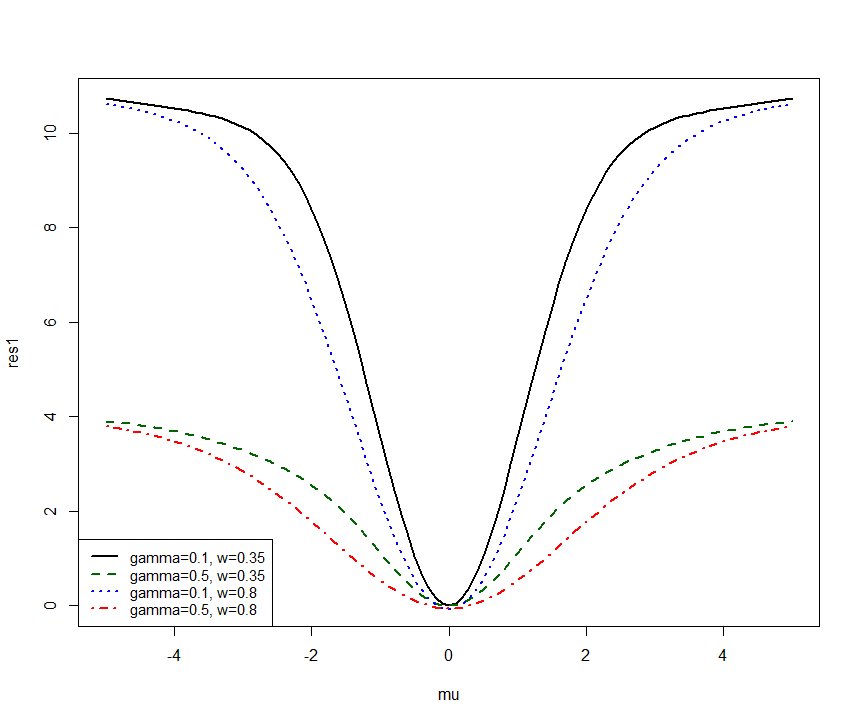}
\caption{Function $P_TH(P_T,\mu)$ for different windows and divergences. They all have an infimum at zero.}
\label{fig:ObjFun}
\end{figure}
It is easy to see that the derivative of $P_TH(P_T,\mu)$ with respect to $\mu$ has a unique zero at $\mu=0$. Besides function $\mu\mapsto P_TH(P_T,\mu)$ is strictly decreasing on $(-\infty,0)$ and strictly increasing on $(0,\infty)$. This is sufficient to prove our claim. Besides assumption 4 becomes well verified.\\
We now move to verify the $P_T-$integrability of $x\mapsto \left(\frac{p_{\phi}}{p_{\phi^T}(x)}\right)^{\gamma}(x)$. We have:
\[\left(\frac{p_{\phi}}{p_{\phi^T}(x)}\right)^{\gamma}(x) p_T(x) = \frac{(1+w^2)^{\gamma/2}}{\sqrt{2\pi}}e^{-\frac{(\gamma+1) w^2 +1}{2(1+w^2)}x^2 + \gamma x\mu - \gamma\mu^2/2}\]
which is clearly integrable.\\
Assumption 3 demands that the quantity $\sup_{\phi}\int{p_{\phi}^{\gamma}(x)dx}$ is finite. We have:
\[\int{p_{\phi}^{\gamma}(x)dx} = \frac{1}{\sqrt{\gamma}}, \quad \forall \mu\in\mathbb{R}.\]
Hence the supremum over $\mu$ is equal to $1/\sqrt{\gamma}$ and assumption 2 is verified.\\
Assumption 4 demands that the quantity $\sup_{\phi,n}\frac{1}{n}\sum_{i=1}^n{\left(\frac{p_{\phi}}{K_{w}*P_n\times p_{\phi^T}}\right)^{\gamma}(y_i)}$ is finite. We have:
\begin{eqnarray*}
K_w*P_n(y_i) & = & \frac{1}{n w} \sum_{j=1}^n{K\left(\frac{y_i-y_j}{w}\right)}\\
 & = & \frac{K(0)}{nw} + \frac{n-1}{n w} \frac{1}{n-1}\sum_{j\neq i}{K\left(\frac{y_i-y_j}{w}\right)}.
\end{eqnarray*}
The problem in the previous average is that the random variables inside are dependent but identically distributed and we cannot apply directly the law of large numbers. Let's try and calculate an almost sure upper bound manually. We have:
\begin{eqnarray*}
\frac{1}{n-1}\sum_{j\neq i}{K\left(\frac{y_i-y_j}{w}\right)} & = & e^{-\frac{y_i^2}{2w^2}} \frac{1}{n-1}\sum_{j\neq i}{e^{-\frac{y_j^2}{2w^2} + \frac{y_iy_j}{w^2}}} \\
 & \geq & e^{-\frac{y_i^2}{2w^2}} \frac{1}{n-1}\sum_{j\neq i}{1-\frac{y_j^2}{2w^2} + \frac{y_iy_j}{w^2}} \\
& \geq & e^{-\frac{y_i^2}{2w^2}}\left[1 - \frac{1}{2w^2}\frac{1}{n-1}\sum_{j\neq i}{y_j^2} + \frac{y_i}{w^2}\frac{1}{n-1}\sum_{j\neq i}{y_j}\right].
\end{eqnarray*}
Now, the averages $\frac{1}{n-1}\sum_{j\neq i}{y_j^2}$ and $\frac{1}{n-1}\sum_{j\neq i}{y_j}$ are sums of i.i.d. random variables. They are distributed independently of $i$ by $\frac{1}{n-1}\chi^2(n-1)+1$ and $\mathcal{N}\left(0,\frac{1}{n-1}\right)$ respectively. Moreover, the distribution of $y_i$ is independent of $i$. Hence, the distribution of the random variable $Z_n = 1 - \frac{1}{2w^2}\frac{1}{n-1}\sum_{j\neq i}{y_j^2} + \frac{y_i}{w^2}\frac{1}{n-1}\sum_{j\neq i}{y_j}$ is independent of $i$. On the other hand, this random variable converges in probability (using the law of large numbers and the Slutsky's lemma) to $1 - \frac{1}{2w^2}$ which is strictly positive since $w^2>\frac{1}{2}$. Thus one can deduce the existence of $n_0$ independent of $i$ such that, for $n\geq n_0$ the probability of the event $\{Z_n \geq 1 - \frac{1}{2w^2} - c\}$ is greater than $1-\eta_n$ for $\eta_n\rightarrow 0$. The value of $c$ is chosen such that $c < 1 - \frac{1}{2w^2}$. This entails that:
\begin{eqnarray*}
\frac{1}{n}\sum_{i=1}^n{\left[\frac{p_{\phi}(y_i)}{p_{\phi^T}(y_i)\times K_{w}*P_n(y_i)}\right]^{\gamma}} & \leq & \frac{1}{n}\sum_{i=1}^n{\left[\frac{p_{\phi}(y_i)}{p_{\phi^T}(y_i)e^{-\frac{y_i^2}{2w^2}}(1 - \frac{1}{2w^2} - c)} \right]^{\gamma}}.
\end{eqnarray*}
It suffices now to prove the boundedness of the sum which is a mean of i.i.d. random variables and one can use the law of large numbers to conclude an approximation, and the Glivenko-Cantelli theorem in order to conclude a result about the supremum over $\phi$. The limit in probability is given by:
\[\int{\left[\frac{p_{\phi}(y_i)}{p_{\phi^T}(y_i)e^{-\frac{y_i^2}{2w^2}}} \right]^{\gamma} dP_T} = \int{\exp\left[-\frac{(\gamma+1)w^4 + (1-\gamma)w^2-\gamma}{2w^2(1+w^2)}y^2 + \gamma y\mu - \frac{\gamma}{2}\mu^2\right]}\]
It is clear that $(\gamma+1)w^4 + (1-\gamma)w^2-\gamma$ needs to be positive in order for the theory to be applicable. A simple calculus shows that $w$ needs to verify the following condition:
\begin{equation}
w^2 > \frac{\gamma-1+\sqrt{5\gamma^2+2\gamma+1}}{2(\gamma+1)}.
\label{eqn:BwCond1}
\end{equation}
Under this condition on $w$, the previous integral can be calculated and is given by:
\[\int{\left[\frac{p_{\phi}(y_i)}{p_{\phi^T}(y_i)e^{-\frac{y_i^2}{2w^2}}} \right]^{\gamma} dP_T} = \frac{1}{\sqrt{a}}\exp\left[-\left(\frac{1}{2}-\frac{\gamma}{2a}\right)\gamma\mu^2\right]\]
where $a = \frac{(\gamma+1)w^4 + (1-\gamma)w^2-\gamma}{w^2(1+w^2)}$. In order for the supremum over $\mu$ to exist, we need that $a>\gamma$, i.e. $(\gamma+1)w^4 + (1-\gamma)w^2-\gamma>\gamma w^2(1+w^2)$. This happens if $w$ verify:
\begin{equation}
w^2\geq \frac{2\gamma-1+\sqrt{4\gamma^2+1}}{2}.
\label{eqn:BwCond2}
\end{equation}
Condition (\ref{eqn:BwCond2}) contains (\ref{eqn:BwCond1}), and hence is the one to be more interesting. For example, for $\gamma=0.1$, the corresponding condition on the window is $w\geq 0.33$. For a 100-sample from the standard gaussian distribution, the window corresponding to the Silverman's rule of thumb is in average $0.35$ whereas the Sheather and Jones' window is in average $0.39$. When adding $10\%$ outliers from a gaussian distribution $\mathcal{N}(10,1)$, these values become $0.41$ and $0.427$ in average. It is important to notice that, the preceding analysis is very simplistic and was based on the naive and relatively harsh inequality $e^x\geq x+1$. Thus, one should be able, using a more rigorous analysis, to get a better lower bounds on the bandwidth of the window.\\
Finally, It is important to notice that the existence of the kernel-based MD$\varphi$DE $\hat{\mu}_n$ is guaranteed since function $\mu\mapsto P_nH(P_n,\mu)$ has the form of the function $e^{-\mu^2}$. Besides, it has a limit equals to 0 as $\mu$ tends to $\pm\infty$. Moreover, it is continuous on $\mathbb{R}$ as a function of $\mu$, hence it is bounded and the infimum exists.
\end{example}
\subsection{Influence Function for a given window}
%\textcolor{red}{The previous arguments about the \emph{perfect} choice of the kernel window are very hopeful.} 
In practice, the choice of the window is based on methods such as cross-validation, gaussian approximations or even based on personal experience. Thus, it is interesting to study the robustness properties supposing that the window is generated by an external tool.\\
We will use the influence function (IF) approach which, although being limited to the existence of a noise-component, is easy to calculate in general\footnote{This is regardless of the theoretical justifications of its existence.} and gives an aspect of the robustness of an estimator whenever the IF is bounded. We derive here in this paragraph the influence function of the new MD$\varphi$DE for the class of power divergences. The general case of function $\varphi$ seems to give an incomprehensive formula, and is not as interesting as the case of power divergences. Recall that the later contains many classical divergences such as the Hellinger, the Pearson's $\chi^2$ and the Neymann's one.\\
Let $C$ be a functional which gives for a probability distribution $P$ the estimator corresponding to the argument of the infimum of $PH(P,\phi)$ defined earlier, i.e.
\[C(P) = \arginf_{\phi\in\Phi}\int{\varphi'\left(\frac{p_{\phi}}{K_{w}*P}\right)(x)p_{\phi}(x)dx} - \int{\varphi^{\#}\left(\frac{p_{\phi}(y)}{K_{w}*P(y)}\right) dP(x)}.\]
Hence, $C(P_n)$ is non other than the estimator given by (\ref{eqn:NewMDphiDE}) for a given $w$. Fisher consistency is translated by $C(P_{\phi^T})=\phi^T$. This is unfortunately not verified in general when the window is supposed to be calculated by an external tool, because the dual formula is a priori a lower bound of $D_{\varphi}(P_{\phi},P_{\phi^T})$, and we cannot be sure that it would verify the same identifiability property, i.e. $D(Q,P)=0$ iff $P=Q$ whenever $\varphi$ is strictly convex. Example 1 shows, however, a case where Fisher consistency is attained for any value of the window $w$.\\

The influence function measures the impact of a small perturbation in the distribution $P$ on the resulting estimator. It is hence defined by:
\[\text{IF}(P,Q) = \lim_{\varepsilon\rightarrow 0} \frac{C\left((1-\varepsilon)P+\varepsilon Q\right) - C(P)}{\varepsilon}\]
We generally detect the influence of an outlier $x_0$ by observing what happens when we replace $P$ by $(1-\varepsilon)P + \varepsilon \delta_{x_0}$.\\ 
In the literature of M-estimates, one may derive the IF from the estimating equation. For power divergences, the estimating equation corresponding to $P$ is given by:
\begin{equation}
\frac{\gamma}{\gamma-1}\int{\frac{p_{C(P)}^{\gamma-1}\nabla p_{C(P)}}{(K_w*P)^{\gamma-1}}(x)dx} = \int{\frac{p_{C(P)}^{\gamma-1}\nabla p_{C(P)}}{(K_w*P)^{\gamma}}(x)dP(x)}. 
\label{eqn:EstimEq}
\end{equation}
The influence function is obtained by "deriving"\footnote{The arginf function is a troubelsome function} the two sides with respect to $\varepsilon$ after having replaced $P$ by $(1-\varepsilon)P+\varepsilon Q$. The following result give the formula of the IF for power divergences when the noise is generated by an arbitrary distribution $Q$ or when an outlier is present.
\begin{theorem}
The influence function of the kernel-based MD$\varphi$DE defined by (\ref{eqn:NewMDphiDE}) for a given window is given by:
\begin{multline}
\text{IF}(P_T,Q) = \gamma A^{-1}\int{\frac{p_{C(P_T)}^{\gamma-1}\left[K_w*Q\right]\nabla p_{C(P_T)}}{(K_w*P_T)^{\gamma}}\left(1-\frac{p}{K*P_T}\right)(x) dx} + A^{-1}\int{\frac{p_{C(P_T)}^{\gamma-1}\nabla p_{C(P_T)}}{(K_w*P_T)^{\gamma}}(x)dQ(x)}.
\label{eqn:IFGeneralDef}
\end{multline}
If $C$ is Fisher consistent, i.e. $C(P_T) = \phi^T$, then the influence function is given by:
\begin{multline}
\text{IF}(P_T,Q) = \gamma A^{-1}\int{\frac{p_{\phi^T}^{\gamma-1}\left[K_w*Q\right]\nabla p_{\phi^T}}{(K_w*P)^{\gamma}}\left(1-\frac{p}{K*P}\right)(x) dx} + A^{-1}\int{\frac{p_{\phi^T}^{\gamma-1}\nabla p_{\phi^T}}{(K_w*P)^{\gamma}}(x)dQ(x)}.
\label{eqn:IFFisherConsist}
\end{multline}
Finally, if $Q = \delta_{x_0}$, then the IF is given by:
\begin{multline}
\text{IF}(P_T,x_0) = \frac{\gamma}{w} A^{-1}\int{\frac{p_{C(P_T)}^{\gamma-1}\left[K_w*\delta_{x_0}\right]\nabla p_{C(P_T)}}{(K_w*P_T)^{\gamma}}\left(1-\frac{p}{K_w*P_T}\right)(x) dx} + A^{-1}\frac{p_{C(P_T)}^{\gamma-1}\nabla p_{C(P_T)}}{(K_w*P_T)^{\gamma}}(x_0)
\label{eqn:IFFisherConsistOutliers}
\end{multline}
\end{theorem}
\begin{proof}
By deriving the left hand side of (\ref{eqn:EstimEq}), we get:
\[\frac{\gamma}{\gamma-1} \int{\frac{\left[(\gamma-1)\nabla p_{C(P)}\left(\nabla p_{C(P)}\right)^t+p_{C(P)}J_{p_{C(P)}}\right]p_{C(P)}^{\gamma-2}}{(K_w*P)^{\gamma-1}}} \text{IF}(P,Q) - \gamma\int{\frac{p_{C(P)}^{\gamma-1}\left[K_w*(Q-P)\right]\nabla p_{C(P)}}{(K_w*P)^{\gamma}}(x) dx}.\]
The right hand side gives:
\begin{multline*}
\int{\frac{\left[(\gamma-1)\nabla p_{C(P)}\left(\nabla p_{C(P)}\right)^tp_{C(P)}^{\gamma-2}+p_{C(P)}^{\gamma-1}J_{p_{C(P)}}\right]}{(K_w*P)^{\gamma}}(x)dP(x)}\text{IF}(P,Q) - \gamma\int{\frac{p_{C(P)}^{\gamma-1}\left[K_w*(Q-P)\right]\nabla p_{C(P)}}{(K_w*P)^{\gamma+1}}(x) dP(x)}\\
+ \int{\frac{p_{C(P)}^{\gamma-1}\nabla p_{C(P)}}{(K_w*P)^{\gamma}}(x)(dQ-dP)(x)}.
\end{multline*}
Let $A$ be the matrix given by:
\begin{eqnarray*}
A = \int{\left(\frac{\gamma}{\gamma-1} - \frac{p(x)}{K_w*P}\right)\frac{\left[(\gamma-1)\nabla p_{C(P)}\left(\nabla p_{C(P)}\right)^t+p_{C(P)}J_{p_{C(P)}}\right]p_{C(P)}^{\gamma-2}}{(K_w*P)^{\gamma-1}}}.
\end{eqnarray*}
We have now:
\begin{multline*}
A\; \text{IF}(P,Q) = \gamma\int{\frac{p_{C(P)}^{\gamma-1}\left[K_w*(Q-P)\right]\nabla p_{C(P)}}{(K_w*P)^{\gamma}}(x) dx} + \int{\frac{p_{C(P)}^{\gamma-1}\nabla p_{C(P)}}{(K_w*P)^{\gamma}}(x)(dQ-dP)(x)} \\
-\gamma\int{\frac{p_{C(P)}^{\gamma-1}\left[K_w*(Q-P)\right]\nabla p_{C(P)}}{(K_w*P)^{\gamma+1}}(x) dP(x)}
\end{multline*}
which, assuming $A$ is invertible and using the estimating equation (\ref{eqn:EstimEq}), may be rewritten as:
\[
\text{IF}(P,Q) = \gamma A^{-1}\int{\frac{p_{C(P)}^{\gamma-1}\left[K_w*Q\right]\nabla p_{C(P)}}{(K_w*P)^{\gamma}}\left(1-\frac{p}{K*P}\right)(x) dx} + A^{-1}\int{\frac{p_{C(P)}^{\gamma-1}\nabla p_{C(P)}}{(K_w*P)^{\gamma}}(x)dQ(x)}.
\]
The remaining of the proof is a simple substitution of $C(P)$ by $\phi^T$ when $P=P_{\phi^T}$, and replacing $Q$ by the dirac measure on a point $x_0$.
\end{proof}

\begin{remark}
The form of the IF is somewhat similar to the IF of the classical MD$\varphi$DE defined by (\ref{eqn:MDphiDEClassique}). \cite{TomaBronia} show that the IF of the classical MD$\varphi$DE is given by:
\[\text{IF}(P_T,x) = J^{-1}\int{\frac{\nabla p_{\phi^T}}{p_{\phi^T}}}\]
where $J$ is the information matrix given by $\int{\frac{\nabla p_{\phi^T} (\nabla p_{\phi^T})^t}{p_{\phi^T}}}$. Going back to the IF of the new MD$\varphi$DE given by (\ref{eqn:IFFisherConsistOutliers}), replacing $K_w*P_T$ by $p_{\phi^T}$ cancels the first term whereas the second term gives $A^{-1}\int{\frac{\nabla p_{\phi^T}}{p_{\phi^T}}}$, where $A = J + \frac{1}{\gamma-1}J_{p_{\phi^T}}$. \\
Intuitively, our modification has resulted in the term $\frac{p_{\phi^T}^{\gamma}}{(K_w*P)^{\gamma}}$ which could oblige the IF to be bounded in some cases. This is the ratio between the true density and the smoothed one. When $\gamma>0$, it is surprising that the IF becomes \emph{more} bounded as the ratio between the true distribution and the smoothed one decreases, which means that the smoothing is producing over estimation at the tail of the distribution.  
\end{remark}
\begin{example}
We resume the univariate gaussian example. Let's calculate the IF given by (\ref{eqn:IFFisherConsistOutliers}) since, as already seen in Example 2, the new MD$\varphi$DE is Fisher consistent.\\
The quantity $\frac{p_{\phi^T}^{\gamma-1}\nabla p_{\phi^T}}{(K_w*P)^{\gamma}}(x_0)$ is the only term which varies. It is given by:
\begin{eqnarray*}
\frac{p_{\phi^T}^{\gamma-1}\nabla p_{\phi^T}}{(K_w*P)^{\gamma}}(x_0) & = & (1+w^2)^{\gamma/2}x_0e^{-\frac{x_0^2}{2}}e^{-(\gamma-1)\frac{x_0^2}{2}} e^{\gamma\frac{x_0^2}{2(1+w^2)}} \\
 & = & (1+w^2)^{\gamma/2} x_0e^{-\frac{\gamma w^2}{2(1+w^2)}x_0^2}
\end{eqnarray*}
Hence, this quantity is bounded as soons as $\gamma>0$. The second quantity is an integral which needs to exist and be finite. We have:
\begin{eqnarray*}
\frac{p_{\phi^T}^{\gamma-1}K((x-x_0)/w)\nabla p_{\phi^T}}{(K_w*P)^{\gamma}}\left(1-\frac{p}{K_w*P}\right)(x) & = &  \frac{(1+w^2)^{\gamma/2}}{w} xe^{-\frac{\gamma w^2}{2(1+w^2)}x^2} e^{-\frac{(x-x_0)^2}{2w^2}} \frac{\frac{1}{\sqrt{1+w^2}}e^{-\frac{x^2}{2(1+w^2)}} - e^{-\frac{x^2}{2}}}{\frac{1}{\sqrt{1+w^2}}e^{-\frac{x^2}{2(1+w^2)}}} \\
 & = & \frac{(1+w^2)^{\frac{\gamma+1}{2}}}{w} \exp\left[-\frac{\gamma w^4 + 1}{2w^2(1+w^2)}x^2 + \frac{xx_0}{w^2} -\frac{x_0^2}{2w^2}\right] \times \\
 &  & \qquad \qquad \qquad \qquad \qquad \qquad \quad \left(\frac{1}{\sqrt{1+w^2}}e^{-\frac{x^2}{2(1+w^2)}} - e^{-\frac{x^2}{2}}\right)
\end{eqnarray*}
It is clear now that if $\gamma>0$, the integral exists. We should not forget that the integral term also depends on $x_0$. The dominating term is $e^{-x_0^2}$, so that the integral term is bounded as a function of $x_0$ as soon as the integral exists.\\
It remains to show that the term $A$ exists and is invertible. Since $\nabla p_{\phi^T} = xe^{-x^2/2}$, and $J_{p_{\phi^T}} = (1+x^2)e^{-x^2/2}$, then:
\[
\frac{\left[(\gamma-1)\nabla p_{C(P)}\left(\nabla p_{C(P)}\right)^t+p_{C(P)}J_{p_{C(P)}}\right]p_{C(P)}^{\gamma-2}}{(K_w*P)^{\gamma-1}} = \sqrt{\frac{1+w^2}{2\pi}} (1+\gamma x^2)e^{-\frac{\gamma w^2 + 1}{2(1+w^2)}x^2}.\]
Hence,
\begin{eqnarray*}
A & = & \sqrt{\frac{1+w^2}{2\pi}}\frac{\gamma}{\gamma-1}\int{(1+\gamma x^2)e^{-\frac{\gamma w^2 + 1}{2(1+w^2)}x^2} dx}  - \frac{1+w^2}{\sqrt{2\pi}}\int{(1+\gamma x^2)e^{-\frac{\gamma w^2 + w^2 + 1}{2(1+w^2)}x^2} dx} \\
 & = & \sqrt{\frac{1+w^2}{2\pi}}\frac{\gamma}{\gamma-1} \left(\sqrt{\frac{2\pi}{a}} + \gamma\sqrt{\frac{2\pi}{a^3}}\right) - \frac{1+w^2}{\sqrt{2\pi}} \left(\sqrt{\frac{2\pi}{b}} + \gamma\sqrt{\frac{2\pi}{b^3}}\right)
\end{eqnarray*}
where $a = \frac{\gamma w^2 + 1}{1+w^2}$ and $b = \frac{\gamma w^2 + w^2 + 1}{1+w^2}$. It is clear that for $\gamma\in(0,1)$, the two terms constituting $A$ have the same sign, hence $A$ cannot be zero since it is the sum of two negative terms. However, if $\gamma>1$, $A$ may by zero for some cases. Indeed, $A$ is 0 whenever $\gamma^2 (1+\gamma+2\gamma w^2)^2 (1+(\gamma+1)w^2)^3 - (\gamma-1)(1+w^2)(1+\gamma+(\gamma+2)w^2)^2 = 0$. Notice that function $w\mapsto \gamma^2 (1+\gamma+2\gamma w^2)^2 (1+(\gamma+1)w^2)^3 - (\gamma-1)(1+w^2)(1+\gamma+(\gamma+2)w^2)^2$ is equal to $2\gamma-1>0$ when $w=0$, whereas it has a $-\infty$ limit at $+\infty$. Thus, it passes by zero since it is a continuous function. \\
Previous arguments permit us to conclude for sure that for $\gamma\in(0,1)$, the influence function of the estimator defined by (\ref{eqn:NewMDphiDE}) is bounded in the gaussian model independently of the bandwidth of the gaussian kernel. Moreover, it is unbounded for $\gamma<0$. Hence, one can hope to get a robust estimation when $\gamma\in(0,1)$. However, further investigations are needed for the case of $\gamma<0$.
\end{example}
%%%%%%%%%%%%%%%%%%%%%%%%%%%%%%%%%%%%%%%%%%%%%%%
% =============================================
% -----------------------------------------------------
% =============================================
%%%%%%%%%%%%%%%%%%%%%%%%%%%%%%%%%%%%%%%%%%%%%%%
\section{The Basu-Lindsay approach}\label{subsec:BLapproach}
The idea of smoothing the empirical distribution was at first employed to avoid the problem of absolute continuity of the model with respect to $dP_n$ when we use the later to replace the true distribution in (\ref{eqn:PhiDivergence}), see \cite{Beran} for the case of the Hellinger distance. \cite{BasuLindsay} argue that the use of such methods require consistency and rates of convergence for the kernel estimator. They propose to smooth not only the empirical distribution, but also the model. For example, if the smoothing is by convolution with a symmetric kernel $K$ such as the gaussian kernel, the Basu-Lindsay approach is summarized in the following two lines:
\begin{eqnarray}
p_{\phi}^*(x) & = & \frac{1}{w}\int_{\mathbb{R}}{p_{\phi}(y) K\left(\frac{x-y}{w}\right)dy}; \nonumber\\
\hat{\phi} & = & \arginf_{\phi\in\Phi} \int_{\mathbb{R}}{\varphi\left(\frac{p_{\phi}^*(x)}{K_{n,w}(x)}\right)K_{n,w}(x)dx},
\label{eqn:BasuLindsayDiv}
\end{eqnarray}
where $K_{n,w}(x) = \frac{1}{nw}\sum{K\left(\frac{x-y_i}{w}\right)}$ is the Parzen-Rosenblatt symmetric-kernel estimator. The authors prove the robustness of (\ref{eqn:BasuLindsayDiv}) using the residual adjustment function (RAF), see \cite{LindsayRAF}, since the corresponding influence function is generally unbounded, keeping first order efficiency in hand. There is still the choice of the kernel and its window, since their theoretical study demands a \emph{transparency assumption} of the kernel\footnote{The transparency assumption here means that the smoothed score function (derivative of the log-likelihood) is proportional to the non smoothed one. The proportion rate can only be a function of the parameters.} which is not verified in general. A transparent kernel ensures no loss of information when smoothing the model density. They also show in simple examples that even when we use non transparent kernel, loss of information is not big provided that we are using a convenient kernel.\\
For example, in the gaussian model $\mathcal{N}(\mu,\sigma^2)$, the gaussian kernel verifies the transparency property. Besides, the smoothed model is merely a gaussian density with variance equal to $\sigma^2+h^2$. Thus, the Basu-Lindsay approach appears as if we are calculating a divergence between a \emph{weighted} version of the model and the kernel estimator.
%%%%%%%%%%%%%%%%%%%%%%%%%%%%%%%%%%%%%%%%%%%%%%%%%
\subsection{Smoothing-the-model's effect}\label{subsec:SmoothingModelEffect}
The Basu-Lindsay approach seems to be more sensitive to the choice of the kernel than standard methods. For example, let's take the case of densities defined on $(0,\infty)$ (with zero possibly included). Simple examples of such distributions are Weibull distributions and generalized Pareto distributions (GPDs). It is well-known that estimation based on symmetric kernels is biased near zero. Thus, smoothing the model with such kernels will result in similar bias near zero. Figure \ref{fig:SmoothingeffectGamma} shows the influence of a gaussian kernel on a GPD model. The smoothed model has a peak near zero and decreases then towards zero, and hence largely underestimates the values of the "non smoothed" model near zero. Thus, the divergence calculates a distance between a biased estimator of the true distribution and a biased model, and there is no intuitive guarantee of what should give the minimization of such a function. Standard methods which do not smooth the model would suffer less from this sort of problems since the bias is only in the kernel estimator.\\
Simulation results show that among the three methods which use a kernel estimator (Beran's approach, the Basu-Lindsay approach and our kernel-based MD$\varphi$DE) the Basu-Lindsay approach is the most sensitive one. Under the model, all three methods do not give satisfactory results in comparison to the MLE (or the classical MD$\varphi$DE) when we use symmetric kernels. When outliers are present, even the Basu-Lindsay estimator still gives a better result than the MLE.\\
\begin{figure}[h]
\centering
\includegraphics[scale=0.42]{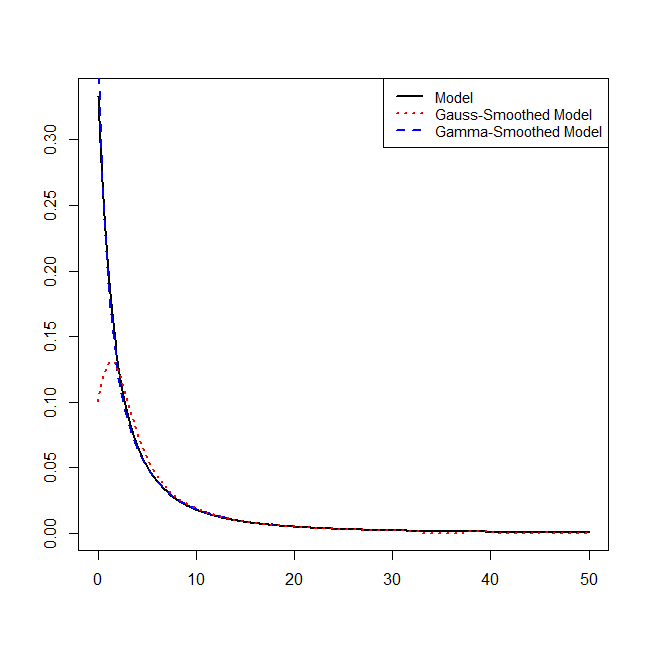}
\caption{Smoothing the model with a gaussian kernel results in a great loss in information. The use of an asymmetric kernel such as the the reciprocal inverse gaussian (RIG) seems to be a good alternative.}
\label{fig:SmoothingeffectGamma}
\end{figure}

\noindent The solution for the previous problem is of course to either use a bias-correction method, see \cite{BiasCorrSurvey}, or to use asymmetric kernels which do not suffer from the boundary bias, see \cite{Libengue}. A more intriguing example is a Weibull distribution with shape parameter in $(0,1)$. The density function explodes to infinity as we approach from zero\footnote{Of course, if we are defining the Weibull distribution with a location parameter, the pdf explodes to infinity near the value of the location parameter.}. Cases such as GPD models can be treated efficiently using bias-correction methods since these assume that the support is semi-closed. Models which has singularities such as the Weibull model can be treated using asymmetric kernels such as gamma kernels or reciprocal inverse gaussian kernels\footnote{Asymmetric kernels have an attractive property that they can treat both bounded and unbounded densities.}. These methods can be employed to recover a good performance in the Basu-Lindsay approach and give better results for the Beran and our kernel-based MD$\varphi$DE.\\ 
Let's see how this kind of solution can be applied on the Basu-Lindsay approach. We discuss only the case of asymmetric kernels since similar arguments apply for bias-correction methods. Let $\hat{f}$ be the asymmetric-kernel estimator defined by:
\[\hat{f}(x) = \frac{1}{n c(y_1,\cdots,y_n)}\sum_{i=1}^n{K_{x,w}(y_i)},\]
where $K_{x,w}$ is the asymmetric kernel calculated at observation $y_i$, and $c(y_1,\cdots,y_n)$ is a constant which ensures integrability to 1. For example, $K$ is the gamma kernel:
\[K_{x,w}(y) = \frac{y^{x/w}}{\Gamma(1+x/w) h^{1+x/w}} e^{-y/w}, \qquad \text{for } y\in[0,\infty),\]
where $\Gamma$ is the classical gamma function. Estimator $\hat{f}$ can no longer be defined as the convolution between the asymmetric kernel and the empirical distribution in the same way as symmetric ones. Thus, the smoothed model in the Basu-Lindsay approach can no longer be obtained by simple convolution. It is given by:
\[p_{\phi}^*(x) = \int_{0}^{\infty}{\frac{1}{c(y)}K_{x,w}(y)p_{\phi}(y)dy},\]
where $c(y)$ is a function which normalizes the kernel for each value of $y$ in order to be a density. It is given by:
\[c(y) = \int_{0}^{\infty}{K_{x,w}(y)dx}.\]
Unfortuantely, this normalization function cannot be calculated but numerically. Taking into account the number of integrations needed to perform such a task and the calculus of the $\varphi-$divergence afterwards which also needs numerical integration, we get a great complexity. In comparison to the classical approach of \cite{Beran}, the calculus of the smoothed model imposes two extra embedded integrals making the calculus of the $\varphi-$divergence very difficult on two levels. The first one is the execution time, and the second one is the subtlety of the whole calculus since all these integrals are carried out over slow decreasing functions on the half line\footnote{The calculs of bounded integrals is far more simple than infinit integrals. Besides, a slow decreasing function (at the border of the its domain), even if it is smooth, is harder to be handled by numerical integration methods than fast decreasing ones.}.\\
\begin{remark}
We were unable to use asymmetric kernels in the Basu-Lindsay approach, because integration calculus (three embedded ones) failed even when restricting the calculus of the normalizing function $c(y)$ on a finite interval. The execution time using the statistical tool \cite{Rtool} on an i7 laptop with 8G RAM took 12 minutes for a simple calculus of the smoothed model. One can imagine now the execution time of the $\varphi-$divergence and finally the optimization over $\phi$. The method should work if one can handle efficiently the problem of numerical integrations and give close results to the case when we do not smooth the model.
\end{remark}
\begin{remark}
The use of the normalization function is necessary to get a very small loss of information. If it is not used, there will be a similar underestimation near zero to the case of symmetric kernels when applied on models defined on a semi-closed intervals.
\end{remark}
Very recently, \cite{vKDE} have proposed a method which does not contain a normalization function. Their approach is based on the so called Mellin transform to approximate the distribution function and then derive an estimate of the density function. Their estimator called as varying kernel density estimator (vKDE) is defined by:
\begin{equation}
\hat{f}_{\alpha}(x) = \frac{1}{n}\sum_{i=1}^n{\frac{1}{y_i}\frac{1}{\Gamma(\alpha)}\left(\frac{\alpha x}{y_i}\right)^{\alpha}\exp\left(-\frac{\alpha x}{y_i}\right)}.
\label{eqn:MTKDE}
\end{equation}
This estimator is different from estimators defined based on symmetric or asymmetric kernels as explained by the authors. They provide a bias-corrected version of this estimator to reduce the bias at the boundary. Nevertheless, we prefer to use (\ref{eqn:MTKDE}) because it integrates to 1 and the Basu-Lindsay approach can be performed more efficiently and reasonably in comparison to the use of asymmetric kernels when working with distributions defined on the half line. The parameter $\alpha$ is a natural number, and (\ref{eqn:MTKDE}) is $L1$--consistent as $\alpha$ goes to infinity under suitable conditions. It even achieves the optimal rate of convergence for MSE and MISE. \\
It is important to notice that $\hat{f}_{\alpha}(0)=1$ for $\alpha\geq 1$. Thus, it is preferable to be used for densities which have value equal to 0 at 0 or for densities which are defined on $(0,\infty)$. In kernel-based estimation procedures, the value at zero is not important because it disappears in integration calculus. Besides, no observation will have exactly the value zero.
%%%%%%%%%%%%%%%%%%%%%%%%%%%%%%%%%%%%%%%%%%%%%%%
% =============================================
% -----------------------------------------------------
% =============================================
%%%%%%%%%%%%%%%%%%%%%%%%%%%%%%%%%%%%%%%%%%%%%%%

\section{Advantages and disadvantages of the new reformulation}
The new reformulation of the minimum dual $\varphi-$divergence has apparently many advantages in comparison with the classical approach and the Basu-lindsay's estimator. We list some of these points.
\begin{itemize}
\item[$\bullet$] The role of the kernel estimator appears directly in the formula of the IF and the ratio between the true distribution and its smoothed version is the part which controls the boundedness of the IF. The method also inherits its robustness from the fact that it approximates a $\varphi-$divergence. Simulation results that it still copes with the performance of both the MLE and classical MD$\varphi$DE when we are under the model.\\
\item[$\bullet$] In comparison with the classical MD$\varphi$DE, our new approach has omitted the double optimization by approximating the argument of the supremum in the dual representation. This constitutes a very important step since on the one hand, the double optimization requires a greater execution time which is of order equals to the square of the time needed for a simple optimization\footnote{There is also the initialization problems for each internal optimization calculus.}. On the other hand, the \emph{supremal} form of the objective function to be minimized afterwords creates further complications in studying the regularity (continuity and differentiability) of the supremal function which play an important role in optimization methods. It is true that optimization methods for non differentiable functions exist, nevertheless, these methods suffer from low convergence speed rates in comparison to methods which use the gradient of the objective function such as first order gradient descent (or the hessian matrix such as second order gradient descent and the BFGS).\\
\item[$\bullet$] Our approach contains only one integration which should be calculated numerically, whereas the smooth-of-the-model techinque in the Basu-Lindsay's estimator creates another integration which should be calculated numerically in general. Besides this calculus intervenes inside an external integration calculus. Thus the number of numerical integrations is highly increased depending on the difficulty of the external integration\footnote{Difficulty comes from a bad shape of the integrand sometimes or irregularities. It also comes from functions with low decreasing rate at infinity for infinite integrals.}. Besides, the use of asymmetric kernels or bias-correction methods is not possible since these methods add another internal integral unless one solve all these integrals efficiently, see Sec. \ref{subsec:BLapproach} for more details.
\item[$\bullet$] There is a difference between our new approach and direct smoothing techniques. Although the performances and estimates are close, our approach keeps the philosophy of approximating a divergence between the model and the empirical distribution. In the Basu-Lindsay's approach or classical methods of inserting a kernel in (\ref{eqn:PhiDivergence}) such as \cite{Beran}, the divergence is calculated between the empirical distribution and the (smoothed) model.\\
\end{itemize}
We list some of the drawbacks of our approach:
\begin{itemize}
\item Our method still suffers, similarly to the Basu-Lindsay's approach and any method which uses kernels, from the problem of choosing the kernel and the window. This problem stays minor as long as we are working with regular densities which converges to zero at both extremities of support. When the density tends to infinity on the border or does not converge to zero, asymmetric kernels or bias-correction methods are needed. Unfortunately, these two tools, although efficient, lack a good and a general method for the choice of the window. \\
\item Unlike the Basu-Lindsay's approach, we were not able to reformulate a general condition such as \emph{kernel transparency} in order to avoid the need to consistency of the kernel estimator in some cases\footnote{For our kernel-based MD$\varphi$DE, the gaussian location model does not need consistency of the kernel when the gaussian kernel is used.}. Note, however, that this transparency condition is still a very hard task, and if it is not verified, consistency of the kernel is needed.\\
\item As we will see in the simulation paragraph, our kernel-based MD$\varphi$DE has apparently traded some of its efficiency with a robustness properties. It is therefore not as good as the MLE and the classical MD$\varphi$DE under the model.\\
\item Our approach is not suitable for working with multidimensional distributions, since in higher dimensions, the so called \emph{curse of dimensionality} appears and the neighborhoods of observed data becomes void. Thus the calculus of the kernel estimator would require much more data than univariate problems. This is not the case of the classical approach. We still can use projection-based nonparameteric estimators to replace the kernel and do the job. We also present hereafter a particular solution to contamination models which can be generalized directly to multivariate cases.\\
\end{itemize}

% -----------------------------------------------------------------------
%%%%%%%%%%%%%%%%%%%%%%%%%%%%%%%%%%%%%%%%%%%%%%%%%%%%%%%%%%%
%========================================================
%%%%%%%%%%%%%%%%%%%%%%%%%%%%%%%%%%%%%%%%%%%%%%%%%%%%%%%%%%%
% -----------------------------------------------------------------------
\section{The Dual \texorpdfstring{$\varphi-$}{TEXT}divergence estimator}\label{sec:DphiDE}
\subsection{General facts and comments}
The dual $\varphi-$divergence estimator (D$\varphi$DE) was defined in \cite{BroniatowskiKeziou2007} as the argument of the supremum in (\ref{eqn:DivergenceDef}). It is defined by:
\begin{equation}
\hat{\alpha}_n = \argsup_{\alpha\in\Phi}\left\{\int{\varphi'\left(\frac{p_{\phi}}{p_{\alpha}}\right)(x)p_{\phi}(x)dx} - \frac{1}{n}\sum_{i=1}^n{\left[\frac{p_{\phi}}{p_{\alpha}} \varphi'\left(\frac{p_{\phi}}{p_{\alpha}}\right) - \varphi'\left(\frac{p_{\phi}}{p_{\alpha}}\right)\right](y_i)}\right\}
\label{eqn:DphiDE}
\end{equation}
This estimator is far more simple than the classical MD$\varphi$DE defined by (\ref{eqn:MDphiDEClassique}) since it needs only a simple optimization over $\alpha$ for a given choice of the escort parameter $\phi$. Besides, this estimator is proved to be robust in some models from an IF point of view, provided a suitable choice of the escort parameter. Indeed, the IF is given by (see \cite{TomaBronia}):
\[\text{IF}(y|\phi) = \left[\int{J_f(x)p_{\phi^T}(x)dx}\right]^{-1}\left[\int{\left(\frac{p_{\phi}}{p_{\phi^T}}\right)^{\gamma}(x)\nabla_{\phi} p_{\phi^T}(x) dx} - \left(\frac{p_{\phi}}{p_{\phi^T}}\right)^{\gamma}(y)\frac{\nabla_{\phi} p_{\phi^T}(y)}{p_{\phi^T}(y)}\right]\]
where:
\[f(\alpha,\phi,y) = \int{\frac{p_{\phi}^{\gamma}}{p_{\alpha}^{\gamma-1}}p_{\phi}dx} - \left[\frac{p_{\phi}}{p_{\alpha}}(y)\right]^{\gamma}\]
Previous papers which discussed the choice of the escort parameter have either let the choice arbitrary in the region where the IF is bounded (this is the case of \cite{TomaBronia}), or proposed to use robust estimates for the escort parameters (this is the case of \cite{Cherfi} and \cite{Frydlova}). The first idea is very complicated since we have no idea about the true value of the parameters and a bad choice of the escort parameter even inside the region where the IF is bounded does not ensure a good result. In \cite{Frydlova} and \cite{Cherfi}, experimental results show that the D$\varphi$DE in a normal model is very close to the escort parameter and coincide with the escort parameter when the later is equal to the MLE. The last fact can be easily verified following the proof of Theorem 6 in \cite{Broniatowski2014}. Indeed, one may show that the MLE is a zero of the estimating equation of the D$\varphi$DE and has a definit negative jacobian matrix of the corresponding objective function. On the other hand, the use of a robust escort parameter is not always a good idea. We discuss these two ideas on two examples.
\begin{example}
We resume the two-component gaussian mixture example. We have already shown that the classical MD$\varphi$DE has an unbounded IF in this model in paragraph \ref{sec:NonRobsutMDphiDE}. The IF of the D$\varphi$DE is not the same. We will try and give some conditions on the escort parameter in order to make it bounded. The first term in the influence function is a matrix which is independent of $y$ and is constant. Supposing that it is invertible, our job is to investigate both the existence of the integral, which is also a constant, and the remaining term which changes according to $y$. The integral exists since the the fraction is of order $e^{ax}$ whereas the derivative is of order $e^{-x^2}$. The remaining term needs to be studied extensively. The fraction $\frac{\nabla_{\phi} p_{\phi^T}(y)}{p_{\phi^T}(y)}$ was already studied in the case of the MD$\varphi$DE. We, therefore, need only to study the fraction $\left(\frac{p_{\phi}}{p_{\phi^T}}\right)^{\gamma}$.
\begin{eqnarray*}
\left(\frac{p_{\phi}}{p_{\phi^T}}\right)^{\gamma} & = & \left(\frac{\lambda e^{-\frac{1}{2}(y-\mu_1)^2} + (1-\lambda)e^{-\frac{1}{2}(y-\mu_1)^2}}{\lambda^T e^{-\frac{1}{2}(y-\mu_1^T)^2} + (1-\lambda^T)e^{-\frac{1}{2}(y-\mu_1^T)^2}}\right)^{\gamma} \\
 & = & \left(\frac{\lambda + (1-\lambda) e^{y(\mu_2-\mu_1) + \frac{1}{2}\mu_1^2-\frac{1}{2}\mu_2^2}}{\lambda^T + (1-\lambda^T) e^{y(\mu_2^T-\mu_1^T) + \frac{1}{2}(\mu_1^T)^2-\frac{1}{2}(\mu_2^T)^2}} e^{y(\mu_1 - \mu_1^T) + \frac{1}{2}(\mu_1^T)^2 - \frac{1}{2}\mu_1^2}\right)^{\gamma} \\
 & = & \left(\frac{1-\lambda + \lambda e^{y(\mu_1-\mu_2) + \frac{1}{2}\mu_2^2-\frac{1}{2}\mu_1^2}}{1-\lambda^T + \lambda^T e^{y(\mu_1^T-\mu_2^T) + \frac{1}{2}(\mu_2^T)^2-\frac{1}{2}(\mu_1^T)^2}} e^{y(\mu_2 - \mu_2^T) + \frac{1}{2}(\mu_2^T)^2 - \frac{1}{2}\mu_2^2}\right)^{\gamma}
\end{eqnarray*}
When $y$ tends to $-\infty$, if $\mu_1>\mu_1^T$, then the second line shows that the fraction gives a finite limit equals to 0. Otherwise, it gives $+\infty$. When $y$ tends to $+\infty$, if $\mu_2<\mu_2^T$, the third line shows that the fraction gives a finite limit equals to 0. Otherwise, it gives $+\infty$. We need to incorporate this with the terms of the vector $\frac{\nabla_{\phi} p_{\phi^T}(y)}{p_{\phi^T}(y)}$. The derivative with respect to $\lambda$, is already bounded, and hence no additional condition is needed. The derivative with respect to $\mu_1$ is also bounded at $+\infty$. However, at $-\infty$ it is of order $y$. Still, it vanishes against the term $e^{\gamma y(\mu_1 - \mu_1^T) + \frac{\gamma}{2}(\mu_1^T)^2 - \frac{\gamma}{2}\mu_1^2}$ which comes from the fraction $\left(\frac{p_{\phi}}{p_{\phi^T}}\right)^{\gamma}$ under conditions $\gamma>0$ and $\mu_1>\mu_1^T$. Finally, the derivative with respect to $\mu_2$ is treated similarly.\\
We conclude that provided that the matrix term is invertible, the influence function of the D$\varphi$DE is bounded whenever the escort parameter verify either of the following conditions according to the value of $\gamma$:
\begin{eqnarray}
\mu_1>\mu_1^T, \quad \mu_2<\mu_2^T \qquad & \text{if } & \gamma>0 \label{eqn:MixGaussRobustCond1}\\
\mu_1<\mu_1^T, \quad \mu_2>\mu_2^T \qquad & \text{if } & \gamma<0 \label{eqn:MixGaussRobustCond2}
\end{eqnarray}
The use of a robust escort parameter verifying the set of conditions (\ref{eqn:MixGaussRobustCond1}, \ref{eqn:MixGaussRobustCond2}) leads to a more robust parameter than the escort. However, the use of a robust escort parameter which does not fulfill the set of conditions (\ref{eqn:MixGaussRobustCond1}, \ref{eqn:MixGaussRobustCond2}) has a negative impact on the resulting estimator. In our simulations in Sec. \ref{sec:Simulations}, we have analyzed the mixture whose true set of parameters is $(\lambda^T=0.35,\mu_1^T=-2,\mu_2^T=1.5)$. We used our new MD$\varphi$DE (with a Silverman's rule for the window) as an escort parameter. The divergence criterion is the Hellinger divergence which corresponds to $\gamma=0.5$. Thus, we are in the context of condition (\ref{eqn:MixGaussRobustCond1}). The new MD$\varphi$DE verify this condition and the resulting D$\varphi$DE has in average a better error, see table \ref{tab:DphiDEGaussMixEx}. In the same table, we give another escort parameter which as good as the previous one depending on our two error criteria, and even slightly better. If we calculate the D$\varphi$DE using this escort parameter which clearly does not verify condition (\ref{eqn:MixGaussRobustCond1}), the resulting estimator does not give a better estimate than the escort. It is clearly worse since the error has nearly been doubled. \\
\begin{table}[h]
\centering
\begin{tabular}{|c|c|c|}
\hline
Estimator & $\chi^2$ & Total variation \\
\hline
$\hat{\phi}_1=(\hat{\lambda}=0.349,\hat{\mu}_1=-1.767,\hat{\mu}_2=1.377)$ & 0.155 & 0.087 \\
$\hat{\phi}_2 = (\hat{\lambda}=0.36,\hat{\mu}_1=-2.2, \hat{\mu}_2=1.7)$ & 0.142 & 0.079\\
\hline
D$\varphi$DE($\hat{\phi}_1$) & 0.142 & 0.076 \\
D$\varphi$DE($\hat{\phi}_2$) & 0.213 & 0.115 \\
\hline
\end{tabular}
\label{tab:DphiDEGaussMixEx}
\caption{The influence of a robust escort parameter on the D$\varphi$DE in a mixture of two gaussian components. The error is calculated between the true distribution and the estimated one, see Sec. \ref{sec:Simulations}.}
\end{table}
\end{example}
%%%%%%%%%%%%%%%
\begin{example}
Let $p_{\phi}$ be a generalized Pareto distribution:
\[p_{\nu,\sigma}(y) = \frac{1}{\sigma}\left(1+\nu\frac{y}{\sigma}\right)^{-1-\frac{1}{\nu}},\quad \text{for } y\geq 0.\]
The shape and the scale are supposed to be unknown and equal to $\nu^T=0.7, \sigma^T=3$. It is necessary for the IF of the D$\varphi$DE to be bounded\footnote{The IF contains an inverse of a $2\times 2$ matrix which cannot be simply calculated. Since it is a mere constant, we only discussed the other terms in the IF.} following the value of $\gamma$ to locate the shape of the escort parameter with respect to the true value of the shape parameter. If $\gamma\in(0,1)$, it is necessary for the IF to be bounded that $\nu<\nu^T$. If $\gamma<0$, then the IF can be bounded whenever $\nu>\nu^T$. Our simulation results in paragraph \ref{subsec:SimulationGPD} shows that for $\gamma=0.5$ (the hellinger divergence), the D$\varphi$DE calculated using a robust escort parameter (our kernel-based MD$\varphi$DE) has deteriorated the performance significantly. The total variation distance corresponding to the escort parameter is 0.05 whereas the total variation distance corresponding to the D$\varphi$DE is $0.12$. The escort parameter gives an estimate of the shape parameter $0.766$ which seems to be a good estimate. It is worth noting that it still gives better results than those obtained using MLE which gives a total variation distance equal to $0.195$.\\
\end{example}
The past two examples\footnote{See the remaining of the simulations for more examples.}form an opposed result to the conjuncture of both articles \cite{Frydlova} and \cite{Cherfi} about the use of robust escort parameter. The use of a robust escort is a gamble and does not guarantee a better estimator than the escort itself. Thus, we are taking a great risk by using the D$\varphi$DE. Notice, finally, that the D$\varphi$DE is still more robust than the MLE and the classical MD$\varphi$DE even if the IF is not bounded.
%%%%%%%%%%%%%%%%%%%%%%%%%%%%%%%%%%%%%%%%%%%%%%%%%%%%%%%%%%%%%%%%%%%
\subsection{Relation with the density power divergences}
The density power divergence (MDPD) was first introduced by \cite{BasuMPD}. It is defined by:
\begin{eqnarray}
\hat{\phi}_n & = & \arginf_{\phi} \int{p_{\phi}^{1+a}}(z) dz - \frac{a+1}{a}\frac{1}{n}\sum_{i}^n{p_{\phi}^{a}(y_i)} \nonumber \\
 & = & \arginf_{\phi} \mathbb{E}_{\phi}\left[p_{\phi}^a\right] - \frac{a+1}{a}\mathbb{E}_n\left[p_{\phi}^{a}\right]
\label{eqn:MPDdef}
\end{eqnarray}
Let's look at the D$\varphi$DE for power divergences with $\gamma=-a<0$. It is given by:
\begin{eqnarray}
\hat{\alpha}_n & = & \argsup_{\alpha\in\Phi} \frac{1}{\gamma-1}\int{\frac{p_{\theta}^{\gamma}}{p_{\alpha}^{\gamma-1}}(x)dx} - \frac{1}{\gamma}\frac{1}{n}\sum_{i=1}^n{\left[\frac{p_{\theta}}{p_{\alpha}}\right]^{\gamma}(y_i)}\nonumber\\
 & = & \argsup_{\alpha\in\Phi} -\frac{1}{a+1}\int{\frac{p_{\alpha}^{a+1}}{p_{\theta}^{a}}(x)dx} + \frac{1}{a}\frac{1}{n}\sum_{i=1}^n{\left[\frac{p_{\alpha}}{p_{\theta}}\right]^{a}(y_i)}\nonumber\\
& = & \arginf_{\alpha\in\Phi} \int{\frac{p_{\alpha}^{a+1}}{p_{\theta}^{a}}(x)dx} - \frac{1+a}{a}\frac{1}{n}\sum_{i=1}^n{\left[\frac{p_{\alpha}}{p_{\theta}}\right]^{a}(y_i)}\nonumber \\
& = & \arginf_{\alpha\in\Phi}\mathbb{E}_{\alpha}\left[\left(\frac{p_{\alpha}}{p_{\theta}}\right)^a\right] - \frac{a+1}{a}\mathbb{E}_n\left[\left(\frac{p_{\alpha}}{p_{\theta}}\right)^a\right]
\label{DphiDENeg}
\end{eqnarray}
A simple comparison between (\ref{eqn:MPDdef}) and (\ref{DphiDENeg}) gives that, the D$\varphi$DE seems to be a penalized form of the MDPD. This penalization by a density $p_{\theta}$ creates a big trouble from a robustness point of view. The robustness of the D$\varphi$DE is now not only controlled by the divergence power $a=-\gamma$ but also through $p_{\theta}$. We have seen in the previous paragraph that the robustness of the D$\varphi$DE in a two-component gaussian mixture varies according to the position of $\theta$ from the true set of parameters. The difficulty of this escort parameter constitutes the only drawback of the D$\varphi$DE in comparison to the MPD. It is still a positive point in its favor. Indeed, the penalization by $p_{\theta}$ can be reread in the spirit of \cite{BroniaKeziou2006}. The ratio $p_{\alpha}/p_{\theta}$ is the Radon-Nikoyme density of $P_{\alpha}$ with respect to $P_{\phi}$. Thus, one can define the D$\varphi$DE even if $P_{\alpha}$ is not absolutely continuous with respect to the Lebesgue measure on $\mathbb{R}$. This fact cannot be done in the MDPD.\\

%Possibility of adapting RAF here !
%Personal point of view: not alternative of MLE in the case of non contamination - problem of calculating integrals, escort and divergence.

% -----------------------------------------------------------------------
%%%%%%%%%%%%%%%%%%%%%%%%%%%%%%%%%%%%%%%%%%%%%%%%%%%%%%%%%%%
%========================================================
%%%%%%%%%%%%%%%%%%%%%%%%%%%%%%%%%%%%%%%%%%%%%%%%%%%%%%%%%%%
% -----------------------------------------------------------------------
\section{A solution for the case of contamination models}
We define a contamination model to be the following mixture model:
\[P_T = (1-\varepsilon)P_{\phi^T} + \varepsilon Q\]
for $\varepsilon\in[0,1)$ which should be small. We have already seen here above that the main problem in the classical MD$\varphi$DE is that the dual representation (\ref{eqn:ParametricDualForm}) largely underestimate the divergence between the true distribution and the model when the data is contaminated. Since the supremum is attained when $p_{\alpha}=p_T$, the model $p_{\alpha}$ cannot cope with the contamination part $\varepsilon Q$ keeping a \emph{good} distance from the main part of the distribution $(1-\varepsilon)P_{\phi^T}$. In order to reestablish the supremum attainment, or at least reduce the gap between the dual representation and the true value of the divergence $D_{\varphi}(P_{\phi},P_T)$, we propose to replace $p_{\alpha}$ by a \emph{contaminated model} $(1-\lambda)p_{\alpha}+\lambda q_{\theta}$. This corresponds to the use of the following class of functions in the dual formula of the divergence (\ref{eqn:GeneralDualRep}):
\[\mathcal{F}_{\theta} = \left\{\varphi'\left(\frac{p_{\phi}}{(1-\lambda)p_{\alpha}+\lambda q_{\theta}}\right), \alpha\in\Phi\subset\mathbb{R}^d, \theta\in\Theta\subset\mathbb{R}^{d'},\lambda\in[0,1)\right\}\]
The minimum dual $\varphi-$divergence estimator can now be defined by:
\begin{equation}
\hat{\phi}_n = \arginf_{\phi\in\Phi} \sup_{\alpha\in\Phi, \theta\in\Theta,\lambda\in[0,1)} \left\{\int{\varphi'\left(\frac{p_{\phi}}{(1-\lambda)p_{\alpha}+\lambda q_{\theta}}\right)(x)p_{\phi}(x)dx} - \frac{1}{n}\sum_{i=1}^n{\varphi^{\#}\left(\frac{p_{\phi}}{(1-\lambda)p_{\alpha}+\lambda q_{\theta}}\right)}\right\}
\label{eqn:MDphiDEContaminModel}
\end{equation}
When we replace $P_n$ by $P_T$, the supremum is attained whenever $\alpha=\phi^T, \lambda=\varepsilon, q_{\theta}=dQ/dx$. Hence, if we are under the model, i.e. $\varepsilon=0$  and $dQ/dx$ is in the submodel $(q_{\theta})_{\theta\in\Theta}$, the previous estimator is Fisher consistent unlike our estimator defined by (\ref{eqn:NewMDphiDE}).\\
The estimator defined by (\ref{eqn:MDphiDEContaminModel}) is clearly a modified version of the classical MD$\varphi$DE defined by (\ref{eqn:MDphiDEClassique}), and show what the classical approach misses. The advantage of such an approach is that we can use it in multidimensional problems without further modifications (unlike our first approach given in paragraph \ref{subsec:KernelSolution}). On the other hand, the choice of a model for the contamination part may be easier than the choice of the kernel and its window in the MD$\varphi$DE defined by (\ref{eqn:NewMDphiDE}), since there is already a whole theory in the literature of time series for modeling the contamination (noise) in a dataset.\\
The influence function of the estimator defined by (\ref{eqn:MDphiDEContaminModel}) can be calculated similarly to the classical MD$\varphi$DE (see \cite{TomaBronia}). However, the general case when $dQ/dx$ is not a member of the submodel $(q_{\theta})_{\theta\in\Theta}$ is very complicated. There is still a simple case when the contamination $Q$ is a member of the submodel $(q_{\theta})_{\theta\in\Theta}$. In this case, the attainement of the supremum in the dual representation permits us to use directly $D_{\varphi}(P_{\phi},P_T)$ which is proved to be a robust tool, see \cite{Donoho}. Thus, the influence function is unbounded and is the same as the influence function of a $\varphi-$divergence $D_{\varphi}(P_{\phi},P_T)$ which is the same as the IF of the MLE (and the classical MD$\varphi$DE), see \cite{LindsayRAF}.

% -----------------------------------------------------------------------
%%%%%%%%%%%%%%%%%%%%%%%%%%%%%%%%%%%%%%%%%%%%%%%%%%%%%%%%%%%
%========================================================
%%%%%%%%%%%%%%%%%%%%%%%%%%%%%%%%%%%%%%%%%%%%%%%%%%%%%%%%%%%
% -----------------------------------------------------------------------
\section{Simulation study}\label{sec:Simulations}
We summarize the results of 100 experiments by giving the average of the estimates and the error committed, and the corresponding standard deviation. We consider two error criteria. The total variation distance (TVD) and the Chi square divergence between the true distribution and the estimated one. These criteria are defined as follows:
\[\sqrt{\chi^2(p_{\phi},p_{\phi^T})} = \sqrt{\int{\frac{\left(p_{\phi}(y)-p_{\phi^T}(y)\right)^2}{p_{\phi^T}(y)}dy}},\quad \text{TVD}(p_{\phi},p_{\phi^T}) = \sup_{A\in\mathcal{B}_n(\mathbb{R})}\left|dP_{\phi}(A) - dP_{\phi^T}(A)\right|.\]
We prefer to use the Chi square divergence, because it measures the relative error between two probability laws. Hence, the error committed on sets where the true distribution attributes small values is penalized in a similar way to sets where the true distribution attributes large values. We use also the TVD because it has the property of measuring the largest error committed when measuring a set $A$ using the estimated distribution instead of the true one. The TVD can be directly calculated using the $L1$ distance. Indeed, the Shceff\'e lemma (see \cite{Meister} page 129.) states that:
\[\sup_{A\in\mathcal{B}_n(\mathbb{R})}\left|dP_{\phi}(A) - dP_{\phi^T}(A)\right| = \frac{1}{2}\int_{\mathbb{R}}{\left|p_{\phi}(y) - p_{\phi^T}(y)\right|dy}.\]
We consider the Hellinger divergence for estimators based on $\varphi-$divergences. The parameter vector is estimated using five methods:
\begin{enumerate}
\item Maximum likelihood (MLE) which is calculated using EM for mixture models;
\item The classical MD$\varphi$DE defined by (\ref{eqn:MDphiDEClassique});
\item Our kernel-based MD$\varphi$DE defined by (\ref{eqn:NewMDphiDE}) with different choices for the kernel and its bandwidth;
\item The Basu-Lindsay approach with different choices for the kernel and its bandwidth;
\item The dual $\varphi$--divergence estimator (D$\varphi$DE) defined by (\ref{eqn:DphiDE}) with escort parameter the result of our kernel-based MD$\varphi$DE with the best choice of the kernel and window among presented possibilities; 
\item The minimum power density estimator (MPD) of \cite{BasuMPD} defined by (\ref{eqn:MPDdef}) for $a\in\{0.1,0.25,0.5,0.75,1\}$.
\end{enumerate}
We give for each experiment a summary of the results with comments, and precise the used kernels and the corresponding windows choices. We finally give an overall conclusion with some practical remarks.\\
Optimization were done using the Nelder-Mead algorithm. Integrations calculus were done using function \texttt{distrExIntegrate} of package \texttt{distrEx} which is a slight modification of the standard function \texttt{integrate}. It performs a Gauss-Legendre quadrature when function \texttt{integrate} returns an error. We have noticed that functions such as \texttt{integral} of package \texttt{pracma}\footnote{Function \texttt{integral} includes a variety of adaptive numerical integration methods such as Kronrod-Gauss quadrature, Romberg's method, Gauss-Richardson quadrature, Clenshaw-Curtis (not adaptive) and (adaptive) Simpson's method.	}, although has a good performance, is slow. Besides, function \texttt{int} of package \texttt{rmutil}, which uses either the Romberg method or algorithm 614 of the collected algorithms from ACM, seems to underestimate the value of the integral in slightly difficult circumstances such as heavy tailed distributions. For example, when we used it to calculate the classical MD$\varphi$DE in the GPD case, it gave robust results because it underestimated the infinity part of the integral (forged thresholding effect). Finally, during some experiences on GPD observations and Weibull distributions based on the Basu-Lindsay approach, function \texttt{distrExIntegrate} failed to converge and function \texttt{integral} was used to attain a result.  \\
Our simulation study covers the following models:
\begin{enumerate}
\item Gaussian model with unknown mean and variance;
\item Two gaussian mixtures with two components where the proportion and the two means are unknown;
\item Generalized Pareto distribution with unknown shape and scale;
\item Three Weibull mixtures with two components where the proportion and the two shapes are unknown.
\end{enumerate}
Outliers were added in the original data in many ways which will be specified according to each case. We have either added noise outside the support of the dataset or by dispersing the noise over the whole dataset. We have also used different distributions to produce the noise.\\
For the first two models, we only used a gaussian kernel with window chosen using either Silverman's rule (nrd0 in the statistical tool R) or Sheather and Jones' rule (SJ). For the heavy tailed models which are defined on half the real line, we needed to use non classical kernels such as asymmetric kernels (RIG: reciprocal inverse gaussian and GA: gamma kernels) and the varying KDE of \cite{vKDE} denoted here as MT (Mellin transform) defined here above by (\ref{eqn:MTKDE}), followed by the value of the bandwidth $\alpha\in\{5,10,15,20\}$. In the GPD model and the first Weibull mixture, we present a simple comparison between symmetric kernels and other non classical methods and showed the advantage of the later in such context. We therefore avoided using symmetric kernels for other Weibull mixtures. For the Basu-Lindsay approach, we did not implement asymmetric kernels, see discussion in paragraph \ref{subsec:SmoothingModelEffect}. We only used the varying KDE.\\
In what concerns the rule for deciding the window for the non classical kernels, we have tried out the cross-validation method (CV), but it resulted always in large (small for the varying KDE) and inconvenient windows especially when outliers are inserted. We were, therefore, obliged to use fixed windows in order to obtain good results. For each kernel and method, the window value or the rule used to calculate it is written next to it. More details can be found at each paragraph.
%%%%%%%%%%%%%%%%%%%%%%%%%%%%%%%%%%%%%%%%%%%%%%%%%%%%%%%%%%%%%%%%%%%%%%
\subsection{Univariate gaussian model}
We consider the gaussian distribution $\mathcal{N}(\mu,\sigma^2)$ when both parameters $\mu$ and $\sigma$ are unknown. We generate at each run a 100-sample of the standard gaussian distribution $\mathcal{N}(0,1)$. Outliers are added simply by replacing the 10 largest values in the sample by the value 10.\\
The maximum likelihood estimator of the parameters are simply the empirical mean and variance $\hat{\mu}=\frac{1}{100}\sum{y_i}, \hat{\sigma}^2 = \frac{1}{99}\sum{(y_i-\hat{\mu})^2}$. For methods which need kernels, we used a gaussian kernel with two rules for the window; Silverman's rule and Sheather and Jones' one. We calculate the power density estimator (MPD) for values of the tradeoff parameter $a\in\{0.1,0.25,0.5,0.75,1\}$. The D$\varphi$DE was calculated using the kernel-based MD$\varphi$DE as an escort with the Silverman's rule. Estimation results are summarized in table \ref{tab:EstimGauss}. Estimation error is calculated in table \ref{tab:ErrGauss}. When we are under the model, all compared methods give the same result with very slight differences. As we add $10\%$ outliers, the classical MD$\varphi$DE and the MLE give the same result which is positively deviated from the true mean with a large variance. This is already expected by virtue of the result of \cite{Broniatowski2014}. Other methods, ours included, give robust results except for MPD with $a=0.1$. Our estimator (for both windows choices) is at the same level of efficiency as the MLE under the model. Besides, the window choice seems irrelevant for methods based on kernels but for Beran's method where Silverman's rule is a bit better. The MPD seems to give the best tradeoff between efficiency and robustness for $a=0.5$ conquering other methods. The kernel-based MD$\varphi$DE and the Basu-Lindsay approaches give slightly better efficiency which is traded with slightly lower robustness in comparison to the result of MPD with $a=0.5$.\\

\begin{table}[hp]
%\resizebox{\columnwidth}{!}{
\centering
\begin{tabular}{|l|c|c|c|c||c|c|c|c|}
\hline
\multirow{2}{2.5cm}{Estimation method} & \multicolumn{4}{|c||}{No Outliers} & \multicolumn{4}{|c|}{$10\%$ Outliers}\\
\cline{2-9}
  & $\mu$ & sd$(\mu)$ & $\sigma$ & sd$(\sigma)$ & $\mu$ & sd$(\mu)$ & $\sigma$ & sd$(\sigma)$\\
 \hline
 \hline
\multicolumn{9}{|c|}{Hellinger} \\
\hline
\hline
Classical MD$\varphi$DE & 0.005 & 0.111 & 0.983 & 0.082 & 0.833 & 0.103 & 3.157 & 0.039\\
\hline
New MD$\varphi$DE - Silverman & 0.005 & 0.113 & 0.967 & 0.081 & -0.187 & 0.114 & 0.810 & 0.069\\
New MD$\varphi$DE - SJ & 0.005 & 0.113 & 0.973 & 0.082 & -0.191 & 0.114 & 0.800 & 0.068\\
\hline
Basu-Lindsay - Silverman & 0.005 & 0.114 & 0.968 & 0.081 & -0.191 & 0.114 & 0.805 & 0.068\\
Basu-Lindsay - SJ & 0.005 & 0.113 & 0.970 & 0.081 & -0.193 & 0.114 & 0.799 & 0.067\\
\hline
Beran - Silverman & 0.005 & 0.113 & 1.024 & 0.087 & -0.191 & 0.114 & 0.878 & 0.075\\
Beran - SJ & 0.005 & 0.112 & 1.048 & 0.089 & -0.192 & 0.114 & 0.853 & 0.073\\
\hline
\hline
MPD 0.1 & 0.005 & 0.112 & 0.983 & 0.082 & 0.319 & 0.111 & 2.451 & 0.079\\
MPD 0.25 & 0.006 & 0.112 & 0.983 & 0.083 & -0.145 & 0.114 & 0.854 & 0.074\\
MPD 0.5 & 0.008 & 0.117 & 0.979 & 0.087 & -0.115 & 0.116 & 0.875 & 0.081\\
MPD 0.75 & 0.010 & 0.123 & 0.975 & 0.093 & -0.093 & 0.120 & 0.894 & 0.089\\
MPD 1 & 0.012 & 0.129 & 0.971 & 0.098 & -0.077 & 0.124 & 0.910 & 0.094\\
\hline
\hline
D$\varphi$DE & 0.005 & 0.112 & 0.982 & 0.082 & -0.164 & 0.114 & 0.873 & 0.080\\
\hline
\hline
MLE & 0.005 & 0.111 & 0.988 & 0.082 & 0.833 & 0.103 & 3.172 & 0.039 \\
\hline
\end{tabular}
%}
\caption{The mean value and the standard deviation of the estimates in a 100-run experiment in the standard gaussian model. The divergence criterion is the Hellinger divergence. The escort parameter of the D$\varphi$DE is taken as the new MD$\varphi$DE with the Silverman bandwidth choice.}
\label{tab:EstimGauss}
\end{table}

\begin{table}[hp]
%\resizebox{\columnwidth}{!}{
\centering
\begin{tabular}{|l|c|c|c|c||c|c|c|c|}
\hline
\multirow{2}{2.5cm}{Estimation method} & \multicolumn{4}{|c||}{No Outliers} & \multicolumn{4}{|c|}{$10\%$ Outliers}\\
\cline{2-9}
  & $\chi^2$ & sd($\chi^2$) & TVD & sd(TVD) & $\chi^2$ & sd($\chi^2$) & TVD & sd(TVD)\\
 \hline
 \hline
\multicolumn{9}{|c|}{Hellinger} \\
\hline
\hline
Classical MD$\varphi$DE & 0.104 & 0.052 & 0.054 & 0.026 & 8.503 & 0.113 & 0.516 & 0.002\\
\hline
New MD$\varphi$DE - Silverman & 0.106 & 0.052 & 0.056 & 0.028 & 0.230 & 0.063 & 0.136 & 0.041\\
New MD$\varphi$DE - SJ & 0.105 & 0.052 & 0.055 & 0.027 & 0.239 & 0.062 & 0.141 & 0.041\\
\hline
Basu-Lindsay - Silverman & 0.105 & 0.052 & 0.055 & 0.028 & 0.235 & 0.062 & 0.139 & 0.040\\
Basu-Lindsay - SJ & 0.105 & 0.052 & 0.055 & 0.027 & 0.240 & 0.062 & 0.142 & 0.040\\
\hline
Beran - Silverman & 0.114 & 0.063 & 0.054 & 0.025 & 0.191 & 0.067 & 0.110 & 0.042\\
Beran - SJ & 0.125 & 0.076 & 0.057 & 0.026 & 0.205 & 0.066 & 0.119 & 0.042\\
\hline
D$\varphi$DE & 0.104 & 0.052 & 0.054 & 0.026 & 0.183 & 0.068 & 0.105 & 0.042\\
\hline
\hline
MPD 0.1 & 0.104 & 0.051 & 0.053 & 0.026 & 5.772 & 0.356 & 0.411 & 0.013\\
MPD 0.25 & 0.105 & 0.052 & 0.054 & 0.026 & 0.185 & 0.066 & 0.107 & 0.042\\
MPD 0.5 & 0.110 & 0.054 & 0.057 & 0.028 & 0.165 & 0.068 & 0.094 & 0.042\\
MPD 0.75 & 0.116 & 0.060 & 0.060 & 0.032 & 0.152 & 0.070 & 0.086 & 0.043\\
MPD 1 & 0.121 & 0.066 & 0.063 & 0.036 & 0.144 & 0.070 & 0.080 & 0.043\\
\hline
\hline
MLE & 0.104 & 0.052 & 0.053 & 0.025 & 8.522 & 0.111 & 0.518 & 0.002\\
\hline
\end{tabular}
%}
\caption{The mean value of errors committed in a 100-run experiment with the standard deviation. The divergence criterion is the Hellinger divergence. The escort parameter of the D$\varphi$DE is taken as the new MD$\varphi$DE with the Silverman bandwidth choice.}
\label{tab:ErrGauss}
\end{table}
\clearpage
%%%%%%%%%%%%%%%%%%%%%%%%%%%%%%%%%%%%%%%%%%%%%%%%%%%%%%%%%%
% ===============================================================
%%%%%%%%%%%%%%%%%%%%%%%%%%%%%%%%%%%%%%%%%%%%%%%%%%%%%%%%%%
\subsection{Mixture of two gaussian components}
We show in this paragraph several simulations from a two-component gaussian mixture where the data is contaminated or not by a $10\%$ of outliers. We present two mixtures. The first one has the following parameters $\lambda = 0.35, \mu_1 = -2, \mu_2 = 1.5$. The second one has closer components (means). Its parameters are $\lambda = 0.45, \mu_1 = -0.5, \mu_2 = 2$. Variances of both components are supposed to be fixed at 1. The two mixtures are ploted in figure \ref{fig:TheTwoMixtures}. We are only interested in the means and the proportion of each class. Contamination was done for the first mixture by adding in the original sample to the 5 lowest values random observations from the uniform distribution $\mathcal{U}[-5,-2]$. We also added to the 5 largest values random observations from the uniform distribution $\mathcal{U}[2,5]$. Estimation results are summarized in table \ref{tab:EstimGaussMix}. Estimation error is calculated in table \ref{tab:ErrGaussMix}. For the second mixture, contamination was done by adding in the original sample to the 5 lowest values random observations from the uniform distribution $\mathcal{U}[-3,-1]$. We add to to the 5 largest values random observations from the uniform distribution $\mathcal{U}[1,3]$. Estimation results are summarized in table \ref{tab:EstimGaussMixCloserMeans}. Estimation error is calculated in table \ref{tab:ErrGaussMixCloserMeans}. Maximum likelihood estimates are calculated using the EM algorithm. \\
\begin{figure}[ht]
\centering
\includegraphics[scale=0.3]{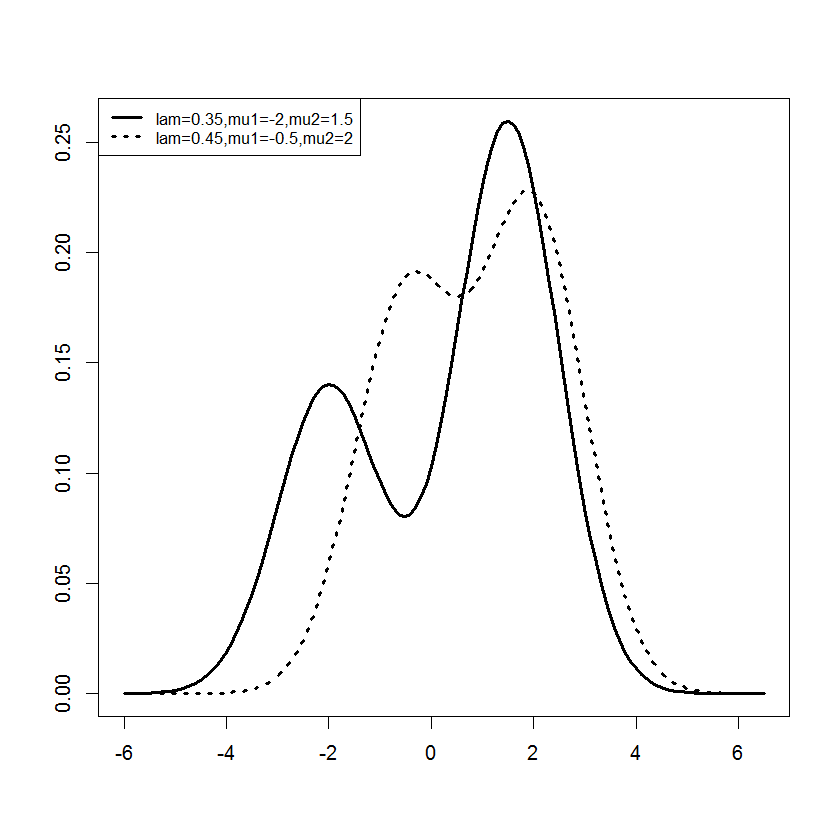}
\caption{The two gaussian mixtures.}
\label{fig:TheTwoMixtures}
\end{figure}

\noindent In what concerns the first mixture (table \ref{tab:ErrGaussMix}): When we are under the model, all compared methods give the same performance. When outliers are added, both classical MD$\varphi$DE and MLE are not robust and give the same result. Other methods provide robust results. The choice of the window has a clearer influence than in the gaussian case. The Silverman's rule gives better results for kernel-based approaches. Error values are close for robust methods and MPD 0.1 is the best one (unlike the univariate gaussian).\\
In what concerns the second mixture: When we are under the model, slight differences appear in favor of the classical MD$\varphi$DE and the MLE (calculated using EM). When we add the outliers, these two estimators fail. MPD for $a=0.1,0.25$ and the Basu-Lindsay approach also fail in the eye of the $\chi^2$ distance. Our kernel-based MD$\varphi$DE have close robustness to the remaining estimators; the MPD for $a=0.5$ and Beran's method. The $\chi^2$ error is more sensitive and show higher differences in favor of our approach against the Basu-Lindsay approach and the minimum power density for small values of the tradeoff parameter. This was basically because of some experiences which failed to converge to a model where the two components are near 0 and considered the second component as the negative noised part of the data. Thus a great relative error has occurred. \\
\begin{table}[hp]
\resizebox{\columnwidth}{!}{
\centering
\begin{tabular}{|l|c|c|c|c|c|c||c|c|c|c|c|c|}
\hline
\multirow{2}{2.5cm}{Estimation method} & \multicolumn{6}{|c||}{No Outliers} & \multicolumn{6}{|c|}{$10\%$ Outliers}\\
\cline{2-13}
  & $\lambda$ & sd($\lambda$) & $\mu_1$ & sd$(\mu_1)$ & $\mu_2$ & sd$(\mu_2)$ & $\lambda$ & sd($\lambda$) & $\mu_1$ & sd$(\mu_1)$ & $\mu_2$ & sd$(\mu_2)$\\
 \hline
 \hline
\multicolumn{13}{|c|}{Hellinger} \\
\hline
\hline
Classical MD$\varphi$DE &0.360 & 0.054 & -1.989 & 0.204 & 1.493 & 0.136 & 0.342 & 0.064 & -2.617 &0.288 & 1.713 & 0.172\\
\hline
New MD$\varphi$DE - Silverman & 0.360 & 0.054 & -1.993 & 0.208 & 1.499 & 0.133 & 0.349 & 0.058 & -1.767 &0.226 & 1.377 & 0.135\\
New MD$\varphi$DE - SJ & 0.359 & 0.054 & -1.981 & 0.206 & 1.490 & 0.134 & 0.346 & 0.059 & -1.706 &0.218 & 1.333 & 0.136\\
\hline
Basu-Lindsay - Silverman & 0.361 & 0.055 & -1.979 & 0.207 & 1.490 & 0.139 & 0.339 & 0.062 & -1.927 &0.305 & 1.377 & 0.158\\
Basu-Lindsay - SJ & 0.360 & 0.054 & -1.977 & 0.203 & 1.486 & 0.135 & 0.346 & 0.059 & -1.751 &0.227 & 1.339 & 0.140\\
\hline
Beran - Silverman & 0.371 & 0.050 & -1.985 & 0.203 & 1.546 & 0.132 & 0.369 & 0.053 & -1.788 & 0.218 & 1.477 & 0.134\\
Beran - SJ & 0.366 & 0.052 & -1.983 & 0.204 & 1.522 & 0.134 & 0.355 & 0.056 & -1.743 & 0.217 & 1.384 & 0.136\\
\hline
D$\varphi$DE & 0.361 & 0.054 & -1.988 & 0.203 & 1.492 & 0.136 & 0.355 & 0.056 & -2.132 &0.224 & 1.605 & 0.137\\
\hline
\hline
MPD 0.1 & 0.360 & 0.054 & -1.991 & 0.207 & 1.493 & 0.134 & 0.346 & 0.059 & -2.052 & 0.243 & 1.452 & 0.144\\
MPD 0.25 & 0.360 & 0.053 & -1.994 & 0.213 & 1.492 & 0.133 & 0.351 & 0.057 & -1.832 & 0.223 & 1.394 & 0.134\\
MPD 0.5 & 0.360 & 0.053 & -1.997 & 0.226 & 1.489 & 0.136 & 0.353 & 0.056 & -1.819 & 0.218 & 1.404 & 0.132\\
\hline
\hline
MLE (EM) & 0.360 & 0.054 & -1.989 & 0.204 & 1.493 & 0.136 & 0.342 & 0.064 & -2.617 &0.288 & 1.713 & 0.172\\
\hline
\end{tabular}
}
\caption{The mean value and the standard deviation of the estimates in a 100-run experiment in a two-components gaussian mixture. The divergence criterion is the Hellinger divergence. The escort parameter of the D$\varphi$DE is taken as the new MD$\varphi$DE with the Silverman bandwidth choice.}
\label{tab:EstimGaussMix}
\end{table}

\begin{table}[hp]
%\resizebox{\columnwidth}{!}{
\centering
\begin{tabular}{|l|c|c|c|c||c|c|c|c|}
\hline
\multirow{2}{2.5cm}{Estimation method} & \multicolumn{4}{|c||}{No Outliers} & \multicolumn{4}{|c|}{$10\%$ Outliers}\\
\cline{2-9}
  & $\chi^2$ & sd($\chi^2$) & TVD & sd(TVD) & $\chi^2$ & sd($\chi^2$) & TVD & sd(TVD)\\
 \hline
 \hline
\multicolumn{9}{|c|}{Hellinger} \\
\hline
\hline
Classical MD$\varphi$DE & 0.113 & 0.044 & 0.064 & 0.025 & 0.335 & 0.102 & 0.150 & 0.034\\
\hline
New MD$\varphi$DE - Silverman & 0.113 & 0.045 & 0.064 & 0.025 & 0.155 & 0.059 & 0.087 & 0.033\\
New MD$\varphi$DE - SJ & 0.113 & 0.045 & 0.064 & 0.025 & 0.179 & 0.061 & 0.101 & 0.035\\
\hline
Basu-Lindsay - Silverman & 0.115 & 0.043 & 0.065 & 0.024 & 0.155 & 0.073 & 0.085 & 0.033\\
Basu-Lindsay - SJ & 0.113 & 0.043 & 0.064 & 0.024 & 0.170 & 0.062 & 0.096 & 0.035\\
\hline
Beran - Silverman & 0.113 & 0.046 & 0.064 & 0.025 & 0.132 & 0.050 & 0.073 & 0.027\\
Beran - SJ & 0.112 & 0.045 & 0.063 & 0.025 & 0.157 & 0.057 & 0.087 & 0.032\\
\hline
D$\varphi$DE & 0.112 & 0.044 & 0.064 & 0.025 & 0.142 & 0.061 & 0.076 & 0.031\\
\hline
\hline
MPD 0.1 & 0.113 & 0.044 & 0.064 & 0.025 & 0.124 & 0.052 & 0.069 & 0.029\\
MPD 0.25 & 0.114 & 0.045 & 0.064 & 0.025 & 0.140 & 0.054 & 0.079 & 0.030\\
MPD 0.5 & 0.117 & 0.047 & 0.065 & 0.025 & 0.138 & 0.053 & 0.078 & 0.030\\
\hline
\hline
MLE & 0.113 & 0.044 & 0.064 & 0.025 & 0.335 & 0.102 & 0.150 & 0.034\\
\hline
\end{tabular}
%}
\caption{The mean value of errors committed in a 100-run experiment with the standard deviation. The divergence criterion is the Hellinger divergence. The escort parameter of the D$\varphi$DE is taken as the new MD$\varphi$DE with the Silverman bandwidth choice.}
\label{tab:ErrGaussMix}
\end{table}

%%%%%%%%%%%%%%%%%%%%%%%%%%%%%%%%%%%%%%%%%%%%%%%%%%%%%%%%%%%%%%%%%%%%%%%%%%%%%%%
%%%%%%%%%%%%%%%%%%%%%%%%%%%%%%%%%%%%%%%%%%%%%%%%%%%%%%%%%%%%%%%%%%%%%%%%%%%%%%%

\begin{table}[hp]
\resizebox{\columnwidth}{!}{
\centering
\begin{tabular}{|l|c|c|c|c|c|c||c|c|c|c|c|c|}
\hline
\multirow{2}{2.5cm}{Estimation method} & \multicolumn{6}{|c||}{No Outliers} & \multicolumn{6}{|c|}{$10\%$ Outliers}\\
\cline{2-13}
  & $\lambda$ & sd($\lambda$) & $\mu_1$ & sd$(\mu_1)$ & $\mu_2$ & sd$(\mu_2)$ & $\lambda$ & sd($\lambda$) & $\mu_1$ & sd$(\mu_1)$ & $\mu_2$ & sd$(\mu_2)$\\
 \hline
 \hline
\multicolumn{13}{|c|}{Hellinger} \\
\hline
\hline
Classical MD$\varphi$DE & 0.457 & 0.077 & -0.487 & 0.240 & 2.006 & 0.187 & 0.437  & 0.128 & -0.860 & 0.478 & 2.192 & 0.343\\
\hline
New MD$\varphi$DE - Silverman &  0.457 & 0.077 & -0.488 & 0.242 & 2.006 & 0.191 & 0.444 & 0.098 & -0.409 & 0.376 & 1.873 & 0.240\\
New MD$\varphi$DE - SJ & 0.456  & 0.077 & -0.490 & 0.242 & 2.009 & 0.191 & 0.443  & 0.098 & -0.381 & 0.376 & 1.851 & 0.235\\
\hline
Basu-Lindsay - Silverman & 0.460 & 0.079 & -0.470 & 0.247 & 2.004 & 0.189 & 0.406 & 0.150 & -0.834 & 0.880 & 1.89 & 0.386\\
Basu-Lindsay - SJ & 0.460  & 0.078 & -0.472 & 0.246 & 2.008 & 0.190 & 0.410  & 0.144 & -0.762 & 0.888 & 1.857 & 0.352\\
\hline
Beran - Silverman & 0.464 & 0.066 & -0.533 & 0.221 & 2.080 & 0.180 & 0.456 & 0.076 & -0.494 & 0.233 & 2.012 & 0.225\\
Beran - SJ & 0.465 & 0.064 & -0.541 & 0.213 & 2.096 & 0.178 & 0.453 & 0.080 & -0.454 & 0.230 & 1.964 & 0.219\\
\hline
D$\varphi$DE & 0.457  & 0.077 & -0.487 & 0.239 & 2.006 & 0.187 & 0.447 & 0.086 & -0.661 &0.283 & 2.100 & 0.231\\
\hline
\hline
MPD 0.1 & 0.456 & 0.077 & -0.492 & 0.238 & 2.005 & 0.191 & 0.424 & 0.142 & -0.843 & 0.872 & 2.015 & 0.504\\
MPD 0.25 & 0.456 & 0.076 & -0.497 & 0.236 & 2.003 & 0.199 & 0.441 & 0.097 & -0.505 &0.443 & 1.912 & 0.243\\
MPD 0.5 & 0.455 & 0.076 & -0.503 & 0.241 & 2.000 & 0.212 & 0.453 & 0.080 & -0.394 &0.234 & 1.906 & 0.205\\
\hline
\hline
MLE & 0.457 & 0.077 & -0.487 & 0.240 & 2.006 & 0.187 &  0.432 & 0.146 & -0.964 & 0.706 & 2.222 & 0.593\\
\hline
\end{tabular}
}
\caption{The mean value and the standard deviation of the estimates in a 100-run experiment in a two-components gaussian mixture with close means. The divergence criterion is the Hellinger divergence. The escort parameter of the D$\varphi$DE is taken as the new MD$\varphi$DE with the Silverman bandwidth choice.}
\label{tab:EstimGaussMixCloserMeans}
\end{table}

\begin{table}[hp]
%\resizebox{\columnwidth}{!}{
\centering
\begin{tabular}{|l|c|c|c|c||c|c|c|c|}
\hline
\multirow{2}{2.5cm}{Estimation method} & \multicolumn{4}{|c||}{No Outliers} & \multicolumn{4}{|c|}{$10\%$ Outliers}\\
\cline{2-9}
  & $\chi^2$ & sd($\chi^2$) & TVD & sd(TVD) & $\chi^2$ & sd($\chi^2$) & TVD & sd(TVD)\\
 \hline
 \hline
\multicolumn{9}{|c|}{Hellinger} \\
\hline
\hline
Classical MD$\varphi$DE &  0.108 & 0.050 & 0.061 & 0.029 & 0.294 & 0.528 & 0.122 & 0.044\\
\hline
New MD$\varphi$DE - Silverman & 0.110  & 0.051 & 0.062 & 0.029 & 0.156 & 0.245 & 0.081 & 0.044\\
New MD$\varphi$DE - SJ & 0.109 & 0.052 & 0.062 & 0.029 &  0.163 & 0.242 & 0.085 & 0.042\\
\hline
Basu-Lindsay - Silverman & 0.110 & 0.050 & 0.062 & 0.029 & 0.961 & 3.366 & 0.097 & 0.067\\
Basu-Lindsay - SJ & 0.110  & 0.050 & 0.063 & 0.029 & 0.982  & 3.606 & 0.092 & 0.067\\
\hline
Beran - Silverman & 0.113 & 0.050 & 0.062 & 0.027 & 0.114 & 0.053 & 0.065 & 0.031\\
Beran - SJ & 0.114 & 0.051 & 0.062 & 0.026 & 0.111 & 0.053 & 0.064 & 0.032\\
\hline
D$\varphi$DE &  0.108 & 0.050 & 0.061 & 0.029 & 0.150 & 0.075 & 0.081 & 0.034\\
\hline
\hline
MPD 0.1 & 0.108 & 0.050 & 0.062 & 0.029 & 2.745 & 10.73 & 0.090 & 0.067\\
MPD 0.25 & 0.110 & 0.051 & 0.063 & 0.029 & 0.589 & 4.676 & 0.072 & 0.042\\
MPD 0.5 & 0.114 & 0.052 & 0.065 & 0.030 & 0.121 & 0.059 & 0.072 & 0.037\\
\hline
\hline
MLE & 0.108  & 0.050 & 0.061 & 0.029 & 1.813  & 6.76 & 0.130 & 0.057\\
\hline
\end{tabular}
%}
\caption{The mean value of errors committed in a 100-run experiment with the standard deviation in a mixture of two gaussian components with close means. The divergence criterion is the Hellinger divergence. The escort parameter of the D$\varphi$DE is taken as the new MD$\varphi$DE with the Silverman bandwidth choice.}
\label{tab:ErrGaussMixCloserMeans}
\end{table}

\clearpage
%%%%%%%%%%%%%%%%%%%%%%%%%%%%%%%%%%%%%%%%%%%%%%%%%%%%%%%%%%%%%%%%%%%%%%%%%%%%%
%%%%%%%%%%%%%%%%%%%%%%%%%%%%%%%%%%%%%%%%%%%%%%%%%%%%%%%%%%%%%%%%%%%%%%%%%%%%%
\subsection{Generalized Pareto distribution}\label{subsec:SimulationGPD}
We show in this paragraph several simulations from the generalized Pareto distribution (GPD) where the data is contaminated or not by a $10\%$ of outliers. A GPD with a fixed location at zero, a scale parameter $\sigma>0$ and a shape parameter $\nu>0$ is defined by:
\[p_{\nu,\sigma}(y) = \frac{1}{\sigma}\left(1+\nu\frac{y}{\sigma}\right)^{-1-\frac{1}{\nu}},\quad \text{for } y\geq 0\]
We generate 100 samples. Each sample contains 100 observations drawn independently from the same distribution GPD($\nu = 0.7, \sigma=3$). Outliers are added by replacing 10 observations (chosen randomly) from each sample by observations from the distribution GPD($\nu=1,\sigma=10,\mu=500$) where $\mu$ is the location parameter. Estimation results are summarized in table \ref{tab:EstimGPD}. Estimation error is calculated in table \ref{tab:ErrGPD}. The maximum likelihood estimator was calculated using the \texttt{gpd.fit} function of package \texttt{ismev}.\\
In the litterature of nonparametric density estimation, it is mentioned everywhere that symmetric kernels are not suitable for densities defined on half the real line because of the boundary effect. We, however, still use them here for the sake of comparison when they are employed inside an estimation criterion. For more details, we invite the reader to revisit paragraph \ref{subsec:SmoothingModelEffect}.\\
When we are under the model, All presented methods except for the Basu-Lindsay approach attained the same efficiency of the MLE and sometimes even better for given choices of the kernel or the tradeoff parameter. Our kernel-based MD$\varphi$DE attained a similar performance to the MLE for \emph{all} non classical kernels and the corresponding choices of the window. Beran's method attained this performance only with the varying KDE (MT 5,10,15,20). MPD attained this level only for small values of $a$ (0.25 and 0.1). Other kernel choices were not very successful except for our kernel-based MD$\varphi$DE with a gaussian kernel and a Silverman's rule. This may be some indication of small sensitivity to the kernel used. \\

When outliers are added, performance of kernel-based methods is slightly deteriorated whereas other methods are greatly influenced, and the error is at least doubled; MPD for all cases included. The use of asymmetric kernels seems to be the most convenient for a GPD model. Our kernel-based MD$\varphi$DE seems to give the best result (in $\chi^2$ and TVD) for all kernels and corresponding windows keeping a great marge in its favor in comparison with other methods. \\

Why does the Basu-Lindsay approach give bad results in a GPD model using a gaussian kernel? A natural answer is that the gaussian kernel is not suited for densities which do not go to zero at both extremities of the domain of definition of the true distribution as was already indicated in Sect. \ref{subsec:BLapproach}. It is well known that symmetric kernels have the so-called boundary effect or bias. In the Basu-Lindsay approach, this fact has a double bad effect. The first is on the kernel estimator which no longer is appropriate to replace the true distribution near zero. The second is on the smoothed model. When the model is smoothed with a gaussian kernel, a great loss in information occurs in comparison to the original model, see Fig (\ref{fig:SmoothingeffectGamma}). Now that both the kernel estimator and the smoothed model are "corrupted", the divergence between them is no longer related to the divergence between the model and the empirical distribution. The use of asymmetric kernels or bias-correction methods was not possible practically because these methods provide non normalized estimators, see paragraph \ref{subsec:SmoothingModelEffect} Remark 5. This causes further difficulties in numerical integrations while smoothing the model, and requires higher execution time than possible. We therefore used the non classical kernel estimator based on the Mellin transform defined by (\ref{eqn:MTKDE}). This estimator is normalized by construction and is free of boundary bias. Results based on such estimator are a clear improvement.\\
Last but not least, it is worth noting that both asymmetric kernels gave very close results for all kernel-based methods. In the remaining experiences, we will only be using the reciprocal inverse gaussian (RIG) kernel.\\
\Remark
The nature of the heavy tail of the GPD (slow decrease at infinity) made integration calculus difficult, and some integration functions failed to give fairly correct results. We, therefore, and in order to avoid integration on an infinite interval $[0,\infty)$, propose to use a quantile trick which is translated by the change of variable:
\[\int_{0}^{\infty}{\varphi'\left(\frac{p_{\phi}}{p_{\alpha}}\right)(x) p_{\phi}(x) dx} = \int_{0}^1{\varphi'\left(\frac{p_{\phi}}{p_{\alpha}}\right)p_{\phi}(\mathbb{F}_{\phi}^{-1}(y))dy}\]
where $\mathbb{F}_{\phi}^{-1}(y) = \frac{\sigma}{\nu}((1-y)^(-\nu)-1)$ is the quantile of the GPD probability law $P_{\phi}$. Although this idea may appear ineffective since it does not change anything in the integral (the quantile funtion takes back values from $[0,1)$ into $[0,\infty)$), it was the savior from using other integration functions such as function \texttt{int} which work, but largely underestimate the true value, see discussion at the beginning of this section. In fact, integration methods perform in general better when integrating on a finite interval than when integrating on an infinite one.\\

\begin{table}[hp]
%\resizebox{\columnwidth}{!}{
\centering
\begin{tabular}{|l|c|c|c|c||c|c|c|c|}
\hline
\multirow{2}{2.5cm}{Estimation method} & \multicolumn{4}{|c||}{No Outliers} & \multicolumn{4}{|c|}{$10\%$ Outliers}\\
\cline{2-9}
  & $\nu$ & sd$(\nu)$ & $\sigma$ & sd$(\sigma)$ & $\nu$ & sd$(\nu)$ & $\sigma$ & sd$(\sigma)$\\
 \hline
 \hline
\multicolumn{9}{|c|}{Hellinger} \\
\hline
\hline
Classical MD$\varphi$DE & 0.721  & 0.174 & 3.029 & 0.575 & 1.655 & 0.113 & 2.694 & 0.491\\
\hline
New MD$\varphi$DE - Gauss Silverman & 0.463 & 0.142 & 2.719 & 0.586 & 0.571 & 0.197 & 2.427 & 0.599\\
New MD$\varphi$DE - Gauss SJ & 0.343 & 0.108 & 2.858 & 0.597 &  0.368 & 0.141 & 2.798 & 0.569\\
New MD$\varphi$DE - RIG CV & 0.528 & 0.140 & 3.125 & 0.611 &  0.775 & 0.202 & 2.844 & 0.571\\
New MD$\varphi$DE - RIG Nrd0 & 0.562 & 0.139 & 3.133 & 0.605 & 0.817 & 0.219 & 2.815 & 0.545\\
New MD$\varphi$DE - RIG SJ & 0.522 & 0.129 & 3.138 & 0.616 & 0.688 & 0.191 & 2.903 & 0.574\\
New MD$\varphi$DE - GA CV & 0.530 & 0.139 & 3.117 & 0.610 & 0.766  & 0.204 & 2.833 & 0.577\\
New MD$\varphi$DE - GA Nrd0 & 0.564 & 0.139 & 3.112 & 0.601 & 0.814  & 0.211 & 2.787 & 0.544\\
New MD$\varphi$DE - GA SJ & 0.520 & 0.126 & 3.135 & 0.607 & 0.691  & 0.185 & 2.895 & 0.576\\
New MD$\varphi$DE - MT 5 & 0.641 & 0.156 & 3.217 & 0.615 & 1.202 & 0.161 & 2.806 & 0.510\\
New MD$\varphi$DE - MT 10 & 0.607 & 0.153 & 3.272 & 0.628 & 1.090 & 0.195 & 2.876 & 0.552\\
New MD$\varphi$DE - MT 15 & 0.588 & 0.150 & 3.307 & 0.636 & 1.026 & 0.206 & 2.920 & 0.565\\
New MD$\varphi$DE - MT 20 & 0.573 & 0.148 & 3.331 & 0.643 & 0.979 & 0.212 & 2.956 & 0.577\\
\hline
Basu-Lindsay - Gauss Silverman & 0.128  & 0.125 & 6.022 & 1.522 & 0.122  & 0.109 & 7.151 & 2.025\\
Basu-Lindsay - Gauss SJ & 0.078  & 0.066 & 4.603 & 1.057 &  0.097 & 0.087 & 4.843 & 1.316\\
Basu-Lindsay - MT 5 & 0.833  & 0.156 & 2.232 & 0.651 & 0.765 & 0.189 & 2.937 & 0.666\\
Basu-Lindsay - MT 10 & 0.853  & 0.197 & 2.297 & 0.659 & 0.777 & 0.193 & 2.880 & 0.704\\
Basu-Lindsay - MT 15 &  0.881 & 0.176 & 2.293 & 0.517 & 1.164 & 0.169 & 2.893 & 0.530\\
Basu-Lindsay - MT 20 &  0.907 & 0.180 & 2.337 & 0.603 & 0.936 & 0.206 & 2.694 & 0.580\\
\hline
Beran - Gauss Nrd0 & 0.216  & 0.108 & 5.165 & 1.218 & 0.197 & 0.125 & 6.084 & 1.546\\
Beran - Gauss SJ & 0.231  & 0.108 & 3.988 & 0.919 & 0.229 & 0.134 & 4.135 & 0.939\\
Beran - RIG CV & 0.516  & 0.134 & 3.890 & 0.832 & 0.833 & 0.218 & 3.944 & 0.745\\
Beran - RIG Nrd0 & 0.515  & 0.138 & 4.441 & 1.026 & 0.878 & 0.233 & 4.229 & 0.954\\
Beran - RIG SJ & 0.507  & 0.136 & 3.813 & 0.787 & 0.732 & 0.200 & 3.641 & 1.113\\
Beran - GA CV & 0.486  & 0.134 & 3.936 & 0.847 & 0.745  & 0.207 & 4.097 & 0.822\\
Beran - GA Nrd0 & 0.475  & 0.139 & 4.510 & 0.998 & 0.778 & 0.220 & 4.547 & 1.032\\
Beran - GA SJ & 0.503  & 0.133 & 3.780 & 0.773 & 0.703 & 0.186 & 3.589 & 0.781\\
Beran - MT 5 & 0.711  & 0.150 & 3.384 & 0.640 & 1.339 & 0.140 & 2.979 & 0.551\\
Beran - MT 10 & 0.665  & 0.150 & 3.315 & 0.620 & 1.231 & 0.155 & 2.900 & 0.530\\
Beran - MT 15 & 0.637  & 0.154 & 3.310 & 0.640 & 1.164 & 0.169 & 2.893 & 0.530\\
Beran - MT 20 & 0.627  & 0.156 & 3.302 & 0.637 & 0.936 & 0.206 & 2.694 & 0.580\\
\hline
D$\varphi$DE & 0.720  & 0.179 & 3.026 & 0.580 & 1.45 & 0.290 & 2.749 & 0.524\\
\hline
\hline
MPD 1 & 0.729 & 0.402 & 3.023 & 0.660 & 1.039 & 0.483  & 3.273 & 0.681 \\
MPD 0.75 & 0.716 & 0.331 & 3.025 & 0.631 & 1.021 & 0.416  & 3.242 & 0.645 \\
MPD 0.5 & 0.715 & 0.263 & 3.023 & 0.603 & 1.028 & 0.361  & 3.171 & 0.605 \\
MPD 0.25 & 0.722 & 0.200 & 3.019 & 0.581 & 1.292 & 0.240  & 2.955 & 0.532 \\
MPD 0.1 & 0.723 & 0.175 & 3.019 & 0.568 & 1.564 & 0.154  & 2.779 & 0.500 \\
\hline
\hline
MLE &   0.719 & 0.174 & 3.031 & 0.58 & 1.654 & 0.113  & 2.695 & 0.492 \\
\hline
\end{tabular}
%}
\caption{The mean value and the standard deviation of the estimates in a 100-run experiment in the GPG model. The divergence criterion is the Neymann Chi square divergence or the Hellinger. The escort parameter of the D$\varphi$DE is taken as the new MD$\varphi$DE with the Silverman bandwidth choice.}
\label{tab:EstimGPD}
\end{table}

\begin{table}[hp]
%\resizebox{\columnwidth}{!}{
\centering
\begin{tabular}{|l|c|c|c|c||c|c|c|c|}
\hline
\multirow{2}{2.5cm}{Estimation method} & \multicolumn{4}{|c||}{No Outliers} & \multicolumn{4}{|c|}{$10\%$ Outliers}\\
\cline{2-9}
  & $\chi^2$ & sd($\chi^2$) & TVD & sd(TVD) & $\chi^2$ & sd($\chi^2$) & TVD & sd(TVD)\\
 \hline
 \hline
\multicolumn{9}{|c|}{Hellinger} \\
\hline
\hline
Classical MD$\varphi$DE &  0.099 & 0.077 & 0.044 & 0.026 & 1.027 & 0.195 & 0.142 & 0.014\\
\hline
New MD$\varphi$DE - Silverman & 0.159 & 0.056 & 0.087 & 0.034 & 0.171 & 0.070 & 0.097 & 0.044\\
New MD$\varphi$DE - SJ & 0.189  & 0.052 & 0.100 & 0.035 & 0.183  & 0.066 & 0.098 & 0.042\\
New MD$\varphi$DE - RIG CV & 0.109 & 0.045 & 0.058 & 0.027 & 0.114 & 0.065 & 0.053 & 0.029\\
New MD$\varphi$DE - RIG Nrd0 & 0.100 & 0.044 & 0.054 & 0.027 & 0.142 & 0.130 & 0.056 & 0.029\\
New MD$\varphi$DE - RIG SJ & 0.110 & 0.044 & 0.059 & 0.027 & 0.104 & 0.056 & 0.054 & 0.030\\
New MD$\varphi$DE - GA CV & 0.108 & 0.045 & 0.058 & 0.027 & 0.114 & 0.063 & 0.054 & 0.029\\
New MD$\varphi$DE - GA Nrd0 & 0.100 & 0.044 & 0.054 & 0.027 & 0.132 & 0.092 & 0.056 & 0.028\\
New MD$\varphi$DE - GA SJ & 0.109 & 0.044 & 0.058 & 0.027 & 0.104 & 0.056 & 0.054 & 0.030\\
New MD$\varphi$DE - MT 5 & 0.093 & 0.053 & 0.049 & 0.028 & 0.472 & 0.307 & 0.089 & 0.024\\
New MD$\varphi$DE - MT 10 & 0.095 & 0.050 & 0.051 & 0.028 & 0.336 & 0.243 & 0.078 & 0.026\\
New MD$\varphi$DE - MT 15 & 0.097 & 0.048 & 0.053 & 0.028 & 0.268 & 0.193 & 0.072 & 0.027\\
New MD$\varphi$DE - MT 20 & 0.099 & 0.047 & 0.054 & 0.029 & 0.226 & 0.154 & 0.068 & 0.028\\
\hline
Basu-Lindsay - Silverman & 0.301  & 0.08 & 0.179 & 0.048 & 0.361  & 0.110 & 0.214 & 0.061\\
Basu-Lindsay - SJ & 0.256  & 0.046 & 0.145 & 0.033 & 0.264  & 0.055 & 0.151 & 0.039\\
Basu-Lindsay - MT 5 & 0.155  & 0.082 & 0.090 & 0.047 & 0.100 & 0.077 & 0.051 & 0.036\\
Basu-Lindsay - MT 10 & 0.155  & 0.080 & 0.085 & 0.043 & 0.102 & 0.078 & 0.053 & 0.038\\
Basu-Lindsay - MT 15 & 0.140  & 0.107 & 0.071 & 0.050 & 0.421 & 0.278 & 0.086 & 0.025\\
Basu-Lindsay - MT 20 & 0.157  & 0.085 & 0.078 & 0.044 & 0.160 & 0.083 & 0.059 & 0.031\\
\hline
Beran - Gauss Nrd0 & 0.241  & 0.072 & 0.142 & 0.045 & 0.297 & 0.090 & 0.177 & 0.053\\
Beran - Gauss SJ & 0.199  & 0.049 & 0.109 & 0.034 & 0.207 & 0.044 & 0.114 & 0.032\\
Beran - RIG CV & 0.133  & 0.060 & 0.076 & 0.038 & 0.226 & 0.128 & 0.094 & 0.041\\
Beran - RIG Nrd0 & 0.164  & 0.085 & 0.097 & 0.051 & 0.306 & 0.235 & 0.114 & 0.054\\
Beran - RIG SJ & 0.123  & 0.060 & 0.069 & 0.039 & 0.146 & 0.097 & 0.070 & 0.048\\
Beran - GA CV & 0.136  & 0.060 & 0.078 & 0.038 &  0.195 & 0.100 & 0.094 & 0.044\\
Beran - GA Nrd0 & 0.169  & 0.078 & 0.101 & 0.048 & 0.267 & 0.186 & 0.121 & 0.057\\
Beran - GA SJ &  0.120 & 0.058 & 0.068 & 0.037 & 0.130 & 0.078 & 0.065 & 0.040\\
Beran - MT 5 & 0.103 & 0.067 & 0.052 & 0.030 & 0.915 & 0.729 & 0.111 & 0.022\\
Beran - MT 10 & 0.093 & 0.057 & 0.049 & 0.029 & 0.581 & 0.615 & 0.095 & 0.023\\
Beran - MT 15 & 0.094 & 0.054 & 0.050 & 0.029 & 0.421 & 0.278 & 0.086 & 0.025\\
Beran - MT 20 & 0.095 & 0.055 & 0.051 & 0.029 & 0.371 & 0.298 & 0.081 & 0.026\\
\hline
D$\varphi$DE &  0.099 & 0.077 & 0.048 & 0.028 & 0.843 & 0.407 & 0.120 & 0.030\\
\hline
\hline
MPD 1 & 0.211 & 0.310 & 0.068 & 0.038 & 0.477 & 0.665 & 0.089 & 0.047\\
MPD 0.75 & 0.204 & 0.389 & 0.062 & 0.034 & 0.424 & 0.545 & 0.085 & 0.043\\
MPD 0.5 & 0.141 & 0.160 & 0.056 & 0.030 & 0.419 & 0.515 & 0.082 & 0.039\\
MPD 0.25 & 0.106 & 0.082 & 0.049 & 0.028 & 0.669 & 0.441 & 0.104 & 0.030\\
MPD 0.1 & 0.099 & 0.083 & 0.047 & 0.027 & 0.955 & 0.326 & 0.133 & 0.019\\
\hline
\hline
MLE & 0.099 & 0.077 & 0.048 & 0.026 & 1.025  & 0.195 & 0.142 & 0.014\\
\hline
\end{tabular}
%}
\caption{The mean value of errors committed in a 100-run experiment with the standard deviation for the GPD model. The divergence criterion is the Neymann Chi square divergence or the Hellinger. The escort parameter of the D$\varphi$DE is taken as the new MD$\varphi$DE with the gamma kernel.}
\label{tab:ErrGPD}
\end{table}

\clearpage
%%%%%%%%%%%%%%%%%%%%%%%%%%%%%%%%%%%%%%%%%%%%%%%%%%%%%%%%%%%%%%%%%%%%%%%%%%

\subsection{Mixtures of Two Weibull Components}
We present the results of estimating three different two-component Weibull mixtures. The model has the following density:
\[p_{\phi}(x) = 2\lambda\nu_1 (2x)^{\nu_1-1} e^{-(2x)^{\nu_1}}+(1-\lambda)\frac{\nu_2}{2}\left(\frac{x}{2}\right)^{\nu_2-1} e^{-\left(\frac{x}{2}\right)^{\nu_2}}.\]
Scale parameters are supposed to be known and equal to $0.5$ for the first component and $2$ for the second component. The proportion is unknown and fixed at $0.35$. Shape parameters are supposed unknown. Our examples cover a variety of cases of a Weibull mixture where the density function has either a finite limit at zero or goes to infinity for one of the components:
\begin{enumerate}
\item a mixture with close modes $\nu_1 = 1.2, \nu_2 = 2$;
\item a mixture with one mode and with limit equal to infinity at zero $\nu_1 = 0.5, \nu_2 = 3$;
\item a mixture with no modes and with limit equal to infinity at zero $\nu_1 = 0.5, \nu_2 = 1$.
\end{enumerate}
We plot these mixtures in figure \ref{fig:ThreeWeibullMixtures}. Outliers were added in different ways to illustrate several scenarios. For the first mixture, outliers were added by replacing 10 observations of each sample chosen randomly by 10 observations drawn independently from a Weibull distribution with shape $\nu = 0.9$ and scale $\sigma = 3$. See tables (\ref{tab:EstimWeibullMixTwoModes}) and (\ref{tab:ErrWeibullMixTwoModes}). For the second mixture, we added to the 10 largest observations of each sample a random observation drawn from the uniform distribution $\mathcal{U}[2,10]$. See tables \ref{tab:EstimWeibullMixOneMode} and \ref{tab:ErrWeibullMixOneMode}. For the third one, outliers were added by replacing 10 observations, chosen randomly, of each sample by observations from the uniform distribution $\mathcal{U}[\max y_i, 75]$ after having verified that no observation in the overall data has exceeded the value 50.  See tables \ref{tab:EstimWeibullMixNoMode} and \ref{tab:ErrWeibullMixNoMode}.\\
\begin{figure}[h]
\centering
\includegraphics[scale=0.4]{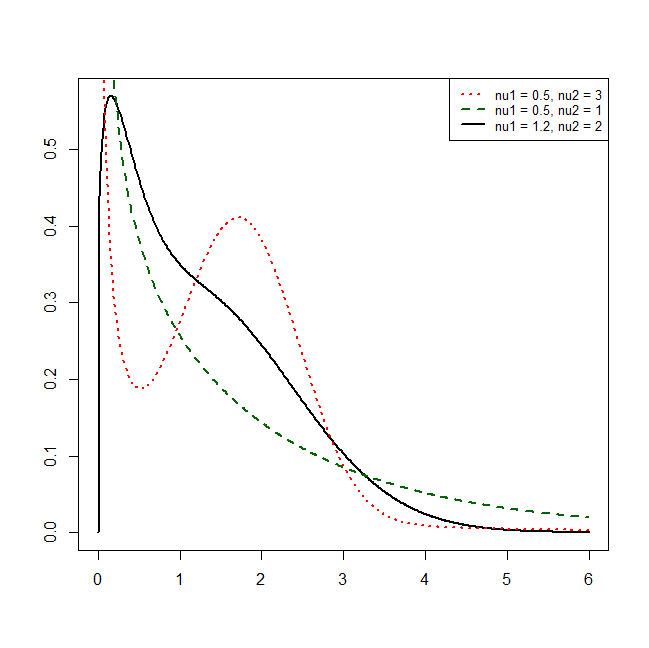}
\caption{The three Weibull mixtures used in our experience.}
\label{fig:ThreeWeibullMixtures}
\end{figure}

\noindent The caclulus of the $\chi^2$ divergence between the estimated model and the true distribution gave often infinity on all mixtures for all estimation methods even under the model. This is because small bias in the estimation of the shape parameter results in a great relative error in both the tail behavior and near zero. We therefore, only provide the TVD as an error criterion.\\  
The first Weibull mixture was the least complicated case. We were able to get satisfactory results for our kernel-based MD$\varphi$DE using a gaussian kernel. The two other mixtures were more challenging, and we needed to use asymmetric kernels to solve the problem of the bias near zero. It is worth noting that the Basu-Lindsay approach provided very bad estimates in the three mixtures which keeps it out of the competition. Note also that the use of a gaussian kernel gave very pleasant results for the first mixture in spite of the boundary bias. We excluded it from mixtures which have infinity limit at zero because it did not work well because of the large bias at zero.\\

For the first mixture, under the model all presented methods provide close results (and sometimes better) to the MLE except for the Basu-Lindsay approach with all available choices and Beran's method with the varying KDE (MT) for windows 5 and 10 which fail. Under contamination, our method gives better results than all other methods and have very close (even slightly better) performance to the MPD for tradeoff parameter higher than 0.25.\\
For the second mixture, the Basu-Lindsay approach failed again. Beran's method gave good result under the model only in one case; the RIG with window 0.01. The density power MPD worked very well only for a tradeoff parameter lower than 0.5 and gave a good compromise between robustness and efficiency. It gave the best compromise in the presented methods. Our kernel-based MD$\varphi$DE has close results to MPD with difference of $0.01$ in the TVD. It is worth noting that our kernel-based MD$\varphi$DE gave faire results for the two proposed kernels; the asymmetric kernel RIG for window 0.01 as before and the varying KDE MT for windows 10, 15 and 20. A fact which was not verified for other kernel-based methods showing again a less sensibility towards the kernel.\\
For the third mixture, the Basu-Lindsay approach did not give good results especially under the model. The only satisfactory results (which gave a good tradeoff between robustness and efficiency) were obtained by our kernel-based MD$\varphi$DE for RIG kernel with window 0.01, Beran's method with the same kernel and window and the MPD for $a=0.5$. Our method and Beran's gave the same result with difference of 0.015 in favor of the power density estimator. Better efficiency were obtained by other choices but on the cost of the robustness of the resulting estimator under contamination.
%%%%%%%%%%%%%%%%%%%
\begin{table}[hp]
\resizebox{\columnwidth}{!}{
\centering
\begin{tabular}{|l|c|c|c|c|c|c||c|c|c|c|c|c|}
\hline
\multirow{2}{2.5cm}{Estimation method} & \multicolumn{6}{|c||}{No Outliers} & \multicolumn{6}{|c|}{$10\%$ Outliers}\\
\cline{2-13}
  & $\lambda$ & sd($\lambda$) & $\nu_1$ & sd$(\nu_1)$ & $\nu_2$ & sd$(\nu_2)$ & $\lambda$ & sd($\lambda$) & $\nu_1$ & sd$(\nu_1)$ & $\nu_2$ & sd$(\nu_2)$\\
 \hline
 \hline
\multicolumn{13}{|c|}{Hellinger} \\
\hline
\hline
Classical MD$\varphi$DE & 0.355 & 0.066 & 1.245 & 0.228 & 2.054 & 0.237 & 0.410 & 0.257 & 1.045 & 0.255 & 1.718 & 0.849\\
\hline
New MD$\varphi$DE - Gauss Silverman & 0.384 & 0.067 & 1.221 & 0.244 & 2.138 & 0.291 & 0.348 & 0.076 & 1.121 & 0.265 & 1.822 & 0.319\\
New MD$\varphi$DE - Gauss SJ & 0.387 & 0.067 & 1.227 & 0.240 & 2.188 & 0.308 & 0.356 & 0.076 & 1.133 & 0.261 & 1.905 & 0.319\\
New MD$\varphi$DE - RIG 0.01 & 0.371 & 0.066 & 1.297 & 0.231 & 2.215 & 0.321 & 0.355 & 0.100 & 1.213 & 0.229 & 1.955 & 0.344\\
New MD$\varphi$DE - RIG 0.1 & 0.358 & 0.065 & 1.233 & 0.210 & 2.065 & 0.267 & 0.330 & 0.117 & 1.127 & 0.226 & 1.741 & 0.304\\
New MD$\varphi$DE - RIG SJ & 0.351 & 0.066 & 1.217 & 0.207 & 2.001 & 0.245 & 0.324 & 0.132 & 1.107 & 0.226 & 1.670 & 0.297\\
New MD$\varphi$DE - MT 5 & 0.328 & 0.112 & 1.301 & 0.235 & 1.809 & 0.192 & 0.363 & 0.229 & 1.195 & 0.213 & 1.592 & 0.356\\
New MD$\varphi$DE - MT 10 & 0.330 & 0.091 & 1.355 & 0.235 & 1.923 & 0.220 & 0.351 & 0.204 & 1.247 & 0.230 & 1.645 & 0.285\\
New MD$\varphi$DE - MT 15 & 0.327 & 0.076 & 1.383 & 0.234 & 1.973 & 0.237 & 0.348 & 0.199 & 1.275 & 0.233 & 1.680 & 0.294\\
New MD$\varphi$DE - MT 20 & 0.328 & 0.076 & 1.403 & 0.233 & 2.002 & 0.249 & 0.348 & 0.198 & 1.295 & 0.235 & 1.702 & 0.297\\
\hline
Basu-Lindsay - Gauss Silverman & 0.752 & 0.064 & 2.199 & 0.248 & 38.66 & 8.66 & 0.822  & 0.083 & 1.927 & 0.276 & 32.37 & 13.52 \\
Basu-Lindsay - Gauss SJ & 0.723 & 0.059 & 2.205 & 0.257 & 16.18 & 10.75 & 0.759  & 0.065 & 1.958 & 0.263 & 19.52 & 10.56 \\
Basu-Lindsay - MT 5 & 0.403 & 0.072 & 1.339 & 0.224 & 3.241 & 0.547 & 0.346 & 0.076 & 1.260 & 0.210 & 2.874 & 0.338 \\
Basu-Lindsay - MT 10 & 0.390 & 0.069 & 1.409 & 0.234 & 3.281 & 0.465 & 0.337 & 0.067 & 1.319 & 0.217 & 2.813 & 0.233 \\
Basu-Lindsay - MT 15 & 0.393 & 0.067 & 1.458 & 0.248 & 3.297 & 0.476 & 0.333 & 0.062 & 1.340 & 0.232 & 2.823 & 0.257 \\
Basu-Lindsay - MT 20 & 0.399 & 0.066 & 1.472 & 0.221 & 3.282 & 0.458 & 0.335 & 0.068 & 1.362 & 0.225 & 2.819 & 0.300 \\
\hline
Beran - Gauss Silverman & 0.254 & 0.058 & 1.313 & 0.087 & 2.010 & 0.200 & 0.182  & 0.074 & 1.174 & 0.162 & 1.703 & 0.253\\
Beran - Gauss SJ & 0.295 & 0.067 & 1.371 & 0.104 & 2.085 & 0.225 & 0.240 & 0.079 & 1.284 &0.127 & 1.794 & 0.266\\
Beran - RIG 0.01 & 0.368 & 0.064 & 1.240 & 0.198 & 2.147 & 0.277 & 0.339 & 0.094 & 1.151 &0.200 & 1.858 & 0.332\\
Beran - RIG 0.1 & 0.345 & 0.061 & 1.117 & 0.103 & 1.897 & 0.172 & 0.289 & 0.095 & 1.033 &0.125 & 1.570 & 0.247\\
Beran - RIG SJ & 0.320 & 0.060 & 1.069 & 0.074 & 1.725 & 0.138 & 0.260 & 0.123 & 0.997 &0.088 & 1.416 & 0.203\\
Beran - MT 5 & 0.453 & 0.307 & 1.146 & 0.178 & 1.386 & 0.180 & 0.626 & 0.349 & 1.055 & 0.172 & 1.461 & 0.531\\
Beran - MT 10 & 0.354 & 0.201 & 1.238 & 0.201 & 1.553 & 0.133 & 0.419 & 0.304 & 1.134 & 0.202 & 1.450 & 0.425\\
Beran - MT 15 & 0.334 & 0.153 & 1.286 & 0.211 & 1.664 & 0.143 & 0.404 & 0.277 & 1.178 & 0.188 & 1.500 & 0.370\\
Beran - MT 20 & 0.334 & 0.136 & 1.317 & 0.218 & 1.738 & 0.156 & 0.383 & 0.256 & 1.207 & 0.198 & 1.542 & 0.348\\
\hline
D$\varphi$DE & 0.356 & 0.066 & 1.248 & 0.232 & 2.069 & 0.278 & 0.332 & 0.142 & 1.113 & 0.248 & 1.700 & 0.289\\
\hline
\hline
MPD 1 & 0.358 & 0.087 & 1.238 & 0.252 & 2.127 & 0.521 & 0.343 & 0.113 & 1.167 & 0.239 & 2.005 & 0.517\\
MPD 0.75 & 0.353 & 0.073 & 1.236 & 0.237 & 2.088 & 0.397 & 0.341 & 0.108 & 1.164 & 0.235 & 1.951 & 0.432\\
MPD 0.5 & 0.354 & 0.068 & 1.238 & 0.230 & 2.071 & 0.345 & 0.336 & 0.105 & 1.159 & 0.237 & 1.860 & 0.344\\
MPD 0.25 & 0.354 & 0.066 & 1.239 & 0.226 & 2.053 & 0.272 & 0.324 & 0.131 & 1.132 & 0.235 & 1.699 & 0.321\\
MPD 0.1 & 0.355 & 0.066 & 1.242 & 0.227 & 2.048 & 0.238 & 0.394 & 0.241 & 1.091 & 0.215 & 1.780 & 0.792\\
\hline
\hline
MLE (EM) & 0.355 & 0.066 & 1.245 & 0.228 & 2.054 & 0.237 & 0.321 & 0.187 & 0.913 & 0.313 & 1.575 & 0.325\\
\hline
\end{tabular}
}
\caption{The mean value and the standard deviation of the estimates in a 100-run experiment on a two-component Weibull mixture ($\lambda=0.35,\nu_1=1.2,\nu_2=2$). The escort parameter of the D$\varphi$DE is taken as the new MD$\varphi$DE with the SJ bandwidth choice.}
\label{tab:EstimWeibullMixTwoModes}
\end{table}

\begin{table}[hp]
%\resizebox{\columnwidth}{!}{
\centering
\begin{tabular}{|l|c|c|c||c|c|c|}
\hline
\multirow{2}{2.5cm}{Estimation method} & \multicolumn{3}{|c||}{No Outliers} & \multicolumn{3}{|c|}{$10\%$ Outliers}\\
\cline{2-7}
  & mean & median & sd & mean & median & sd\\
 \hline
 \hline
\multicolumn{7}{|c|}{Hellinger} \\
\hline
\hline
Classical MD$\varphi$DE & 0.052 & 0.048 & 0.025 & 0.108 & 0.094 & 0.099\\
\hline
New MD$\varphi$DE - Gauss Silverman & 0.058 & 0.054 & 0.029 & 0.068 & 0.065 & 0.034\\
New MD$\varphi$DE - Gauss SJ & 0.058 & 0.053 & 0.029 & 0.064 & 0.061 & 0.031\\
New MD$\varphi$DE - RIG 0.01 & 0.058 & 0.052 & 0.030 & 0.059 & 0.057 & 0.030 \\
New MD$\varphi$DE - RIG 0.1 & 0.051 & 0.049 & 0.026 & 0.066 & 0.062 & 0.032 \\
New MD$\varphi$DE - RIG SJ & 0.050 & 0.050 & 0.026 & 0.071 & 0.066 & 0.032 \\
New MD$\varphi$DE - MT 5 & 0.057 & 0.055 & 0.025 & 0.081 & 0.074 & 0.032 \\
New MD$\varphi$DE - MT 10 & 0.054 & 0.053 & 0.026 & 0.075 & 0.071 & 0.032 \\
New MD$\varphi$DE - MT 15 & 0.054 & 0.054 & 0.026 & 0.073 & 0.069 & 0.032 \\
New MD$\varphi$DE - MT 20 & 0.055 & 0.054 & 0.027 & 0.073 & 0.069 & 0.031 \\
\hline
Basu Lindsay - Gauss Silverman & 0.298 & 0.289 & 0.042 & 0.247 & 0.253 & 0.050\\
Basu Lindsay - Gauss SJ & 0.252 & 0.256 & 0.051 & 0.242 & 0.246 & 0.044\\
Basu Lindsay - MT 5 & 0.127 & 0.141 & 0.046 & 0.121 & 0.111 & 0.042\\
Basu Lindsay - MT 10 & 0.133 & 0.136 & 0.039 & 0.117 & 0.111 & 0.036\\
Basu Lindsay - MT 15 & 0.134 & 0.141 & 0.039 & 0.118 & 0.110 & 0.038\\
Basu Lindsay - MT 20 & 0.132 & 0.138 & 0.039 & 0.117 & 0.109 & 0.039\\
\hline
Beran - Gauss Silverman & 0.068 & 0.062 & 0.028 & 0.082 & 0.081 & 0.031\\
Beran - Gauss SJ & 0.060 & 0.054 & 0.028 & 0.067 & 0.065 & 0.029 \\
Beran - RIG 0.01 & 0.052 & 0.048 & 0.026 & 0.060 & 0.058 & 0.029 \\
Beran - RIG 0.1 & 0.042 & 0.039 & 0.020 & 0.067 & 0.061 & 0.030 \\
Beran - RIG SJ & 0.045 & 0.044 & 0.017 & 0.079 & 0.076 & 0.030 \\
Beran - MT 5 & 0.099 & 0.097 & 0.016 & 0.125 & 0.125 & 0.022 \\
Beran - MT 10 & 0.073 & 0.070 & 0.021 & 0.102 & 0.100 & 0.028 \\
Beran - MT 15 & 0.064 & 0.060 & 0.022 & 0.092 & 0.089 & 0.030 \\
Beran - MT 20 & 0.059 & 0.055 & 0.023 & 0.086 & 0.084 & 0.030 \\
\hline
D$\varphi$DE & 0.053 & 0.049 & 0.027 & 0.068 & 0.065 & 0.031\\
\hline
\hline
MPD 1 & 0.065 & 0.061 & 0.034 & 0.068 & 0.064 & 0.030 \\
MPD 0.75 & 0.059 & 0.056 & 0.029 & 0.063 & 0.060 & 0.029 \\
MPD 0.5 & 0.056 & 0.052 & 0.029 & 0.061 & 0.056 & 0.029 \\
MPD 0.25 & 0.052 & 0.048 & 0.027 & 0.068 & 0.067 & 0.031 \\
MPD 0.1 & 0.051 & 0.048 & 0.026 & 0.088 & 0.083 & 0.039 \\
\hline
\hline
MLE & 0.052 & 0.048 & 0.025 & 0.095 & 0.098 & 0.035 \\
\hline
\end{tabular}
%}
\caption{The mean value with the standard deviation of the TVA committed in a 100-run experiment on a two-component Weibull mixture ($\lambda=0.35,\nu_1=1.2,\nu_2=2$). The escort parameter of the D$\varphi$DE is taken as the new MD$\varphi$DE with the SJ bandwidth choice.}
\label{tab:ErrWeibullMixTwoModes}
\end{table}

%%%%%%%%%%%%%%%%%%%
%%%%%%%%%%%%%%%%%%%
\begin{table}[hp]
\resizebox{\columnwidth}{!}{
\centering
\begin{tabular}{|l|c|c|c|c|c|c||c|c|c|c|c|c|}
\hline
\multirow{2}{2.5cm}{Estimation method} & \multicolumn{6}{|c||}{No Outliers} & \multicolumn{6}{|c|}{$10\%$ Outliers}\\
\cline{2-13}
  & $\lambda$ & sd($\lambda$) & $\nu_1$ & sd$(\nu_1)$ & $\nu_2$ & sd$(\nu_2)$ & $\lambda$ & sd($\lambda$) & $\nu_1$ & sd$(\nu_1)$ & $\nu_2$ & sd$(\nu_2)$\\
 \hline
 \hline
\multicolumn{13}{|c|}{Hellinger} \\
\hline
\hline
Classical MD$\varphi$DE & 0.344 & 0.059 & 0.497 & 0.079 & 3.063 & 0.476  & 0.376 & 0.053 & 0.339 & 0.030 & 2.892 & 0.484\\
\hline
New MD$\varphi$DE RIG - 0.01 & 0.330 & 0.061 & 0.540 & 0.140 & 3.170 & 0.503 & 0.338 & 0.061 & 0.432 & 0.105 & 3.055 & 0.583\\
New MD$\varphi$DE RIG - 0.1 & 0.371 & 0.063 & 0.468 & 0.138 & 3.045 & 0.452 & 0.392  & 0.072 & 0.372 & 0.085 & 2.927 & 0.464\\
New MD$\varphi$DE RIG - SJ & 0.395 & 0.072 & 0.442 & 0.134 & 3.013 & 0.443 & 0.424  & 0.086 & 0.354 & 0.082 & 2.916 & 0.459\\
New MD$\varphi$DE MT - 5 & 0.311 & 0.062 & 0.520 & 0.065 & 2.875 & 0.451 & 0.316 & 0.063 & 0.376 & 0.036 & 2.699 & 0.471\\
New MD$\varphi$DE MT - 10 & 0.302 & 0.062 & 0.548 & 0.077 & 2.903 & 0.433 &  0.306 & 0.062 & 0.384 & 0.039 & 2.727 & 0.448\\
New MD$\varphi$DE MT - 15 & 0.295 & 0.063 & 0.564 & 0.084 & 2.927 & 0.434 & 0.298 & 0.063 & 0.388 & 0.042 & 2.745 & 0.450\\
New MD$\varphi$DE MT - 20 & 0.289 & 0.063 & 0.575 & 0.091 & 2.943 & 0.437 & 0.291 & 0.063 & 0.392 & 0.044 & 2.758 & 0.454\\
\hline
Basu-Lindsay MT - 5 & 0.250 & 0.070 & 0.834 & 0.168 & 2.849 & 0.733 & 0.185 & 0.074 & 0.715 & 0.208 & 2.189& 0.155\\
Basu-Lindsay MT - 10 & 0.240 & 0.065 & 0.797 & 0.157 & 2.789 & 0.550 & 0.197 & 0.087 & 0.707 & 0.201 & 2.324& 0.132\\
Basu-Lindsay MT - 15 & 0.254 & 0.073 & 0.745 & 0.140 & 2.915 & 0.584 & 0.204 & 0.078 & 0.674 & 0.181 & 2.352& 0.092\\
\hline
Beran RIG - 0.01 & 0.298 & 0.058 & 0.647 & 0.082 & 3.017 & 0.437 & 0.295 & 0.057 & 0.486 & 0.081 & 2.842 & 0.460\\
Beran RIG - 0.1 & 0.234 & 0.054 & 0.652 & 0.105 & 2.374 & 0.245 & 0.216 & 0.053 & 0.408 & 0.056 & 2.149 & 0.291\\
Beran RIG - SJ & 0.194 & 0.056 & 0.653 & 0.134 & 1.936 & 0.246 & 0.142 & 0.065 & 0.402 & 0.144 & 1.601 & 0.325\\
Beran MT - 5 & 0.250 & 0.070 & 0.463 & 0.058 & 1.603 & 0.140 & 0.245 & 0.083 & 0.340 & 0.062 & 1.494 & 0.208\\
Beran MT - 10 & 0.278 & 0.066 & 0.501 & 0.069 & 2.005 & 0.181 & 0.275 & 0.079 & 0.354 & 0.033 & 1.868 & 0.260\\
Beran MT - 15 & 0.286 & 0.065 & 0.524 & 0.075 & 2.224 & 0.218 & 0.284 & 0.071 & 0.365 & 0.033 & 2.068 & 0.280\\
\hline
D$\varphi$DE & 0.343 & 0.059 & 0.5004 & 0.084 & 3.047 & 0.474 & 0.372 & 0.056 & 0.357 & 0.056 & 2.897 & 0.502\\
\hline
\hline
MDE 0.75 & 0.444 & 0.126 & 0.595 & 0.080 & 3.466 & 0.643 & 0.417 & 0.127 & 0.602 & 0.087 & 3.233 & 0.606\\
MDE 0.5 & 0.376 & 0.067 & 0.551 & 0.093 & 3.159 & 0.488 & 0.357 & 0.067 & 0.555 & 0.097 & 2.980 & 0.484\\
MDE 0.25 & 0.347 & 0.061 & 0.512 & 0.096 & 3.057 & 0.472 & 0.331 & 0.062 & 0.471 & 0.068 & 2.879 & 0.491\\
MDE 0.1 & 0.344 & 0.059 & 0.496 & 0.084 & 3.050 & 0.470 & 0.343 & 0.058 & 0.384 & 0.037 & 2.859 & 0.484\\
\hline
\hline
MLE (EM) & 0.344 & 0.059 & 0.498 & 0.079 & 3.063 & 0.476 & 0.376 & 0.053 & 0.339 & 0.303 & 2.892 & 0.482\\
\hline
\end{tabular}
}
\caption{The mean value and the standard deviation of the estimates in a 100-run experiment in a two-component Weibull mixture ($\lambda=0.35,\nu_1=0.5,\nu_2=3$). The escort parameter of the D$\varphi$DE is taken as the new MD$\varphi$DE with the Silverman bandwidth choice.}
\label{tab:EstimWeibullMixOneMode}
\end{table}

\begin{table}[hp]
%\resizebox{\columnwidth}{!}{
\centering
\begin{tabular}{|l|c|c|c||c|c|c|}
\hline
\multirow{2}{2.5cm}{Estimation method} & \multicolumn{3}{|c||}{No Outliers} & \multicolumn{3}{|c|}{$10\%$ Outliers}\\
\cline{2-7}
  & mean & median & sd & mean & median & sd\\
 \hline
 \hline
\multicolumn{7}{|c|}{Hellinger} \\
\hline
\hline
Classical MD$\varphi$DE & 0.060 & 0.055 & 0.024 & 0.096 & 0.094 & 0.025\\
\hline
New MD$\varphi$DE RIG - 0.01 & 0.074 & 0.070 & 0.034 & 0.076 & 0.073 & 0.039\\
New MD$\varphi$DE RIG - 0.1 & 0.079 & 0.064 & 0.053 & 0.099 & 0.086 & 0.062\\
New MD$\varphi$DE RIG - SJ & 0.091 & 0.075 & 0.068 & 0.120 & 0.099 & 0.078\\
New MD$\varphi$DE MT - 5 & 0.062 & 0.061 & 0.027 & 0.081 & 0.073 & 0.031\\
New MD$\varphi$DE MT - 10 & 0.066 & 0.064 & 0.028 & 0.076 & 0.070 & 0.030\\
New MD$\varphi$DE MT - 15 & 0.069 & 0.068 & 0.028 & 0.076 & 0.071 & 0.030\\
New MD$\varphi$DE MT - 20 & 0.072 & 0.073 & 0.029 & 0.076 & 0.071 & 0.030\\
\hline
Basu-Lindsay MT - 5 & 0.119 & 0.114 & 0.039 & 0.131 & 0.121 & 0.029  \\
Basu-Lindsay MT - 10 & 0.109 & 0.106 & 0.033 & 0.119 & 0.100 & 0.038  \\
Basu-Lindsay MT - 15 & 0.107 & 0.103 & 0.030 & 0.112 & 0.097 & 0.033  \\
\hline
Beran RIG - 0.01 & 0.077 & 0.080 & 0.026 & 0.066 & 0.063 & 0.029\\
Beran RIG - 0.1 & 0.105 & 0.104 & 0.025 & 0.112 & 0.108 & 0.038\\
Beran RIG - SJ & 0.157 & 0.032 & 0.032 & 0.193 & 0.180 & 0.053 \\
Beran MT - 5 & 0.182 & 0.183 & 0.025 & 0.207 & 0.202 & 0.032 \\
Beran MT - 10 & 0.127 & 0.127 & 0.028 & 0.153 & 0.146 & 0.037 \\
Beran MT - 15 & 0.102 & 0.104 & 0.029 & 0.126 & 0.121 & 0.036 \\
\hline
D$\varphi$DE & 0.060 & 0.057 & 0.024 & 0.091 & 0.088 & 0.027\\
\hline
\hline
MDP 0.75 & 0.103 & 0.083 & 0.067 & 0.097 & 0.083 & 0.065 \\
MDP 0.5 & 0.068 & 0.067 & 0.029 & 0.069 & 0.067 & 0.028 \\
MDP 0.25 & 0.062 & 0.058 & 0.026 & 0.064 & 0.062 & 0.029 \\
MDP 0.1 & 0.061 & 0.059 & 0.024 & 0.076 & 0.072 & 0.027 \\
\hline
\hline
MLE & 0.060 & 0.056 & 0.024 & 0.096 & 0.094 & 0.024 \\
\hline
\end{tabular}
%}
\caption{The mean value with the standard deviation of the TVA committed in a 100-run experiment on a two-component Weibull mixture ($\lambda=0.35,\nu_1=0.5,\nu_2=3$). The escort parameter of the D$\varphi$DE is taken as the new MD$\varphi$DE with the SJ bandwidth choice.}
\label{tab:ErrWeibullMixOneMode}
\end{table}

%%%%%%%%%%%%%%%%%%%
%%%%%%%%%%%%%%%%%%%
\begin{table}[hp]
\resizebox{\columnwidth}{!}{
\centering
\begin{tabular}{|l|c|c|c|c|c|c||c|c|c|c|c|c|}
\hline
\multirow{2}{2.5cm}{Estimation method} & \multicolumn{6}{|c||}{No Outliers} & \multicolumn{6}{|c|}{$10\%$ Outliers}\\
\cline{2-13}
  & $\lambda$ & sd($\lambda$) & $\nu_1$ & sd$(\nu_1)$ & $\nu_2$ & sd$(\nu_2)$ & $\lambda$ & sd($\lambda$) & $\nu_1$ & sd$(\nu_1)$ & $\nu_2$ & sd$(\nu_2)$\\
 \hline
 \hline
\multicolumn{13}{|c|}{Hellinger} \\
\hline
\hline
Classical MD$\varphi$DE & 0.367 & 0.102 & 0.550 & 0.104 & 1.054 & 0.194 & 0.352  & 0.158 & 0.273 & 0.050 & 1.051 & 0.407\\
\hline
New MD$\varphi$DE - 0.01 & 0.445 & 0.103 & 0.562 & 0.135 & 1.212 & 0.284 & 0.409 & 0.133 & 0.464 & 0.156 & 1.148 & 0.293\\
New MD$\varphi$DE - 0.1 & 0.432 & 0.101 & 0.502 & 0.141 & 1.139 & 0.241 & 0.460  & 0.210 & 0.378 & 0.125 & 1.114 & 0.302\\
New MD$\varphi$DE - SJ & 0.431 & 0.101 & 0.485 & 0.141 & 1.127 & 0.244 & 0.487 & 0.216 & 0.356 & 0.108 & 1.110 & 0.309\\
New MD$\varphi$DE MT - 5 & 0.350 & 0.158 & 0.619 & 0.134 & 1.006 & 0.211 & 0.436 & 0.313 & 0.375 & 0.121 & 1.245 & 1.177\\
New MD$\varphi$DE MT - 10 & 0.338 & 0.148 & 0.643 & 0.135 & 1.019 & 0.167 & 0.474 & 0.322 & 0.409 & 0.140 & 1.150 & 0.516\\
New MD$\varphi$DE MT - 15 & 0.335 & 0.148 & 0.658 & 0.135 & 1.029 & 0.161 & 0.456 & 0.321 & 0.411 & 0.146 & 1.292 & 1.689\\
\hline
Basu-Lindsay MT - 5 & 0.392 & 0.178 & 0.734 & 0.122 & 1.042 & 0.022 & 0.351 & 0.225 & 0.757 & 0.177 & 1.048 & 0.026\\
Basu-Lindsay MT - 10 & 0.340 & 0.149 & 0.742 & 0.103 & 1.037 & 0.024 & 0.260 & 0.175 & 0.712 & 0.147 & 1.039 & 0.024\\
Basu-Lindsay MT - 15 & 0.340 & 0.149 & 0.742 & 0.103 & 1.037 & 0.024 & 0.222 & 0.126 & 0.696 & 0.125 & 1.043 & 0.016\\
\hline
Beran - 0.01 & 0.370 & 0.098 & 0.685 & 0.091 & 1.125 & 0.188 & 0.381 & 0.211 & 0.572 & 0.183 & 1.058 & 0.215\\
Beran - 0.1 & 0.234 & 0.093 & 0.747 & 0.113 & 1.028 & 0.118 & 0.419 & 0.372 & 0.479 & 0.211 & 1.181 & 0.553\\
Beran RIG - SJ & 0.211 & 0.185 & 0.745 & 0.130 & 1.034 & 0.230 & 0.259 & 0.331 & 0.367 & 0.181 & 1.105 & 0.542\\
Beran MT - 5 & 0.302 & 0.205 & 0.584 & 0.129 & 0.867 & 0.120 & 0.471 & 0.388 & 0.376 & 0.128 & 1.097 & 0.738\\
Beran MT - 10 & 0.327 & 0.175 & 0.610 & 0.132 & 0.929 & 0.121 & 0.490 & 0.347 & 0.394 & 0.131 & 1.155 & 0.803\\
Beran MT - 15 & 0.331 & 0.165 & 0.623 & 0.128 & 0.962 & 0.128 & 0.470 & 0.340 & 0.400 & 0.132 & 1.174 & 0.893\\
\hline
D$\varphi$DE & 0.371 & 0.111 & 0.544 & 0.100 & 1.064 & 0.240 & 0.473 & 0.293 & 0.382 & 0.175 & 1.431 & 1.818\\
\hline
\hline
MPD 0.75 & 0.494 & 0.181 & 0.619 & 0.089 & 1.341 & 0.689 & 0.505 & 0.243 & 0.625 & 0.087 & 1.313 & 0.641\\
MPD 0.5 & 0.413 & 0.134 & 0.577 & 0.101 & 1.143 & 0.349 & 0.412 & 0.255 & 0.582 & 0.101 & 1.059 & 0.358\\
MPD 0.25 & 0.366 & 0.108 & 0.542 & 0.110 & 1.064 & 0.349 & 0.554 & 0.348 & 0.503 & 0.117 & 1.205 & 0.995\\
MPD 0.1 & 0.368 & 0.109 & 0.539 & 0.106 & 1.059 & 0.237 & 0.451 & 0.322 & 0.370 & 0.111 & 1.280 & 1.407\\
\hline
\hline
MLE (EM) & 0.372 & 0.108 & 0.549 & 0.100 & 1.055 & 0.192 & 0.417 & 0.194 & 0.291 & 0.073 & 1.114 & 0.468\\
\hline
\end{tabular}
}
\caption{The mean value and the standard deviation of the estimates in a 100-run experiment in a two-component Weibull mixture ($\lambda=0.35,\nu_1=0.5,\nu_2=1$). The escort parameter of the D$\varphi$DE is taken as the new MD$\varphi$DE with the Silverman bandwidth choice.}
\label{tab:EstimWeibullMixNoMode}
\end{table}

\begin{table}[hp]
%\resizebox{\columnwidth}{!}{
\centering
\begin{tabular}{|l|c|c|c||c|c|c|}
\hline
\multirow{2}{2.5cm}{Estimation method} & \multicolumn{3}{|c||}{No Outliers} & \multicolumn{3}{|c|}{$10\%$ Outliers}\\
\cline{2-7}
  & mean & median & sd & mean & median & sd\\
 \hline
 \hline
\multicolumn{7}{|c|}{Hellinger} \\
\hline
\hline
Classical MD$\varphi$DE & 0.056 & 0.055 & 0.026 & 0.124 & 0.114 & 0.035\\
\hline
New MD$\varphi$DE RIG - 0.01 & 0.079 & 0.073 & 0.039 & 0.090 & 0.082 & 0.044\\
New MD$\varphi$DE RIG - 0.1 &  0.079 & 0.065 & 0.059 & 0.112 & 0.101 & 0.050\\
New MD$\varphi$DE RIG - SJ & 0.076 & 0.065 & 0.041 & 0.129 & 0.117 & 0.065\\
New MD$\varphi$DE MT - 5 & 0.063 & 0.058 & 0.029 & 0.114 & 0.095 & 0.041\\
New MD$\varphi$DE MT - 10 & 0.067 & 0.063 & 0.028 & 0.112 & 0.102 & 0.038\\
New MD$\varphi$DE MT - 15 & 0.069 & 0.067 & 0.028 & 0.111 & 0.105 & 0.036\\
\hline
Basu-Lindsay MT - 5 & 0.095 & 0.067 & 0.078 & 0.118 & 0.087 & 0.088 \\
Basu-Lindsay MT - 10 & 0.094 & 0.074 & 0.073 & 0.112 & 0.088 & 0.080 \\
Basu-Lindsay MT - 15 & 0.093 & 0.072 & 0.067 & 0.103 & 0.088 & 0.063 \\
\hline
Beran RIG 0.01 & 0.079 & 0.081 & 0.028 & 0.089 & 0.087 & 0.033\\
Beran RIG 0.1 & 0.087 & 0.085 & 0.023 & 0.103 & 0.102 & 0.025\\
Beran RIG - SJ & 0.094 & 0.092 & 0.023 & 0.100 & 0.097 & 0.021 \\
Beran MT - 5 & 0.061 & 0.060 & 0.022 & 0.127 & 0.134 & 0.044 \\
Beran MT - 10 & 0.059 & 0.055 & 0.025 & 0.115 & 0.096 & 0.041 \\
Beran MT - 15 & 0.060 & 0.056 & 0.025 & 0.112 & 0.097 & 0.039 \\
\hline
D$\varphi$DE & 0.057 & 0.055 & 0.028 & 0.117 & 0.113 & 0.034\\
\hline
\hline
MPD 0.75 & 0.102 & 0.091 & 0.050 & 0.093 & 0.088 & 0.039 \\
MPD 0.5 & 0.072 & 0.067 & 0.032 & 0.075 & 0.074 & 0.033 \\
MPD 0.25 & 0.061 & 0.056 & 0.028 & 0.092 & 0.090 & 0.039 \\
MPD 0.1 & 0.058 & 0.055 & 0.027 & 0.108 & 0.087 & 0.039 \\
\hline
\hline
MLE & 0.056 & 0.055 & 0.026 & 0.122 & 0.117 & 0.029 \\
\hline
\end{tabular}
%}
\caption{The mean value with the standard deviation of errors committed in a 100-run experiment on a two-component Weibull mixture ($\lambda=0.35,\nu_1=0.5,\nu_2=1$). The escort parameter of the D$\varphi$DE is taken as the new MD$\varphi$DE with the SJ bandwidth choice.}
\label{tab:ErrWeibullMixNoMode}
\end{table}
\clearpage
%%%%%%%%%%%%%%%%%%%%%%%%%%%%%%%%%%%%%%%%%%%%%%%%%%%%%%%%%%%%%

\subsection{Concluding remarks and comments}
Simulation results, although do not cover a wide range of models and divergences, give some indications about the robustness and the efficiency of the compared results. They also present possible solutions for many difficult estimation problems by employing non classical kernel methods. We summarize the most important remarks based on our simulations presented in this last section.
\begin{itemize}
\item[$\bullet$] Both MLE and classical MD$\varphi$DE have the best efficiency under the model even in \emph{difficult} models with heavy tails where kernel-based approaches could not give a satisfactory result. In regular situations such as the gaussian model (mixtures included), all methods were equivalent under the model.\\
\item[$\bullet$] When contamination is present, the compared estimators gave results as expected. Both MLE and classical MD$\varphi$DE are not robust against contamination. The D$\varphi$DE guided by our kernel-based MD$\varphi$DE gave very good results under the model, however, when contamination is present it failed to improve the result obtained by the escort in difficult situations which is the case of the three Weibull mixtures and the GPD. It even gave very bad results some times in comparison to other estimation methods, but still better than MLE and the classical MD$\varphi$DE. \\
\item[$\bullet$] The Basu-Lindsay approach worked very well in regular situations and even showed improved efficiency in comparison to the Beran's method which is concordant to the result of \cite{BasuSarkar}. It gave surprisingly good results in the GPD model under contamination when we used the varying KDE in comparison to the situation under the model. Unfortunately, it did not give satisfactory results in any of the Weibull mixtures. This method seems very sensitive to the kernel under difficult situations since the model is already influenced by the kernel creating a loss of information.\\
\item[$\bullet$] The minimum density power divergence gave very good results in all situations but the GPD. The best tradeoff parameter from our set of candidates was $a=0.5$. \\
\item[$\bullet$] The Beran's method gave very good tradeoff (and many times the best) between robustness and efficiency in most of the situations, but not very well in the GPD model. The best choice of the kernel for GPD and Weibull mixtures was the RIG with window 0.01. It was sensitive to the choice of the kernel and its window in many situations.\\
\item[$\bullet$] Our kernel-based MD$\varphi$DE gave very good results in all situations and had close results to the MPD and Beran's methods. It gave the best results in the GPD model with very good compromise between efficiency and robustness. It is worth noting that our new estimator was less influenced by the choice of the kernel and the window than all kernel-based methods which participated in the comparison showing very promising and encouraging properties.\\
\item[$\bullet$] The use of the varying KDE (MT) gave best results under the model. We believe that one can get better results if we have a better method for choosing the window than the cross-validation procedure presented in \cite{vKDE}. Recall that the cross-validation method gave good results under the model but very bad ones under contamination. It chose a value $a=1$ for most of the samples.\\
\item[$\bullet$] We are surprised that the best window that corresponds to the best performance for asymmetric kernels and the varying KDE was the most extreme one (the least for asymmetric kernels and the largest for MT). Such a choice corresponds to a fluctuating nonparameteric density estimator. Apparently, the bias at the border played the most important part in estimation. Note that for the RIG kernel as the window becomes smaller, the estimator goes faster towards infinity at zero.
\end{itemize}
\bibliographystyle{plainnat}
\bibliography{NewMDphiDE}
\end{document}